\documentclass[12pt,letterpaper]{article}

\usepackage{amsmath}
\usepackage{amsfonts}
\usepackage{amssymb}
\usepackage{amsthm}
\usepackage{graphicx}
\usepackage[hidelinks]{hyperref}
\usepackage{mathrsfs}
\usepackage{xcolor}
\usepackage{multirow}
\usepackage{soul}
\usepackage{url}
\usepackage{todonotes}
\usepackage{verbatim}
\usepackage{titletoc}
\usepackage{color}
\usepackage{pdflscape}

\usepackage[letterpaper]{geometry}
\usepackage{setspace}
\allowdisplaybreaks

\usepackage[authoryear]{natbib}

\usepackage[subtle]{savetrees}
\usepackage[small,compact]{titlesec}

\newtheoremstyle{break}
  {\topsep}{\topsep}%
  {\itshape}{}%
  {\bfseries}{}%
  {\newline}{}%
\theoremstyle{break}
\newtheorem{theorem}{Theorem}
\newtheorem{proposition}{Proposition}

\newtheorem{lemma}{Lemma}

\newtheoremstyle{algorithms}
  {\topsep}{\topsep}%
  {\upshape}{}%
  {\bfseries}{}%
  {\newline}{}%
\theoremstyle{algorithms}
\newtheorem*{one sided}{Algorithm One-Sided}%
\newtheorem*{one sided star}{Algorithm One-Sided$^*$}%
\newtheorem*{two sided}{Algorithm Two-Sided}%
\newtheorem*{two sided star}{Algorithm Two-Sided$^*$}%
\newtheorem*{example}{Example:~Regression with Sign-Restricted Control Coefficients}

\usepackage[font=small, labelfont=bf]{caption}

\DeclareMathOperator{\argmin}{argmin}

\DeclareMathOperator{\diag}{Diag}

\begin{document}

\title{Short and Simple Confidence Intervals when the Directions of Some Effects are Known\thanks{
We thank Timothy Armstrong, Christopher Blattman, Michal Koles\'ar, 
Koohyun Kwon, Soonwoo Kwon, Julian Jamison, Margaret Sheridan, Liam Wren-Lewis, and Kaspar Wuthrich for helpful comments and suggestions and Chad Brown for excellent research assistance.}} 



\author{Philipp Ketz\footnote{Paris School of Economics, philipp.ketz@psemail.eu} \qquad Adam McCloskey\footnote{Department of Economics, University of Colorado, adam.mccloskey@colorado.edu}  
}

\maketitle

\begin{abstract}

We provide adaptive confidence intervals on a parameter of interest in the presence of nuisance parameters when some of the nuisance parameters have known signs.  The confidence intervals are adaptive in the sense that they tend to be short at and near the points where the nuisance parameters are equal to zero.  We focus our results primarily on the practical problem of inference on a coefficient of interest in the linear regression model when it is unclear whether or not it is necessary to include a subset of control variables whose partial effects on the dependent variable have known directions (signs).  Our confidence intervals are trivial to compute and can provide significant length reductions relative to standard confidence intervals in cases for which the control variables do not have large effects.  At the same time, they entail minimal length increases at any parameter values.  We prove that our confidence intervals are asymptotically valid uniformly over the parameter space and illustrate their length properties in an empirical application to a factorial design field experiment and a Monte Carlo study calibrated to the empirical application.

\bigskip

\noindent \textsc{Keywords:~confidence intervals, adaptive inference, uniform inference, sign restrictions, boundary problems}
\end{abstract}

\thispagestyle{empty}
\setcounter{page}{0}

\newpage

\section{Introduction}\label{sec:intro}

Consider the common empirical setting for which a researcher is interested in estimating the causal effect of one variable on another via a linear regression in the presence of one or more observed potential control variables.  The researcher believes that a regression including this full set of controls should not suffer from omitted variables bias but is uncertain whether it is necessary to include them all to overcome this bias.  To obtain more informative inference, the researcher would prefer not to include controls unnecessarily and knows that if some of these controls indeed influence the outcome variable, it must be in a known positive or negative direction.  Indeed, typical heuristic explanations for the potential inclusion of a control variable to mitigate omitted variables bias involve a known ``direction'' for the effect of the omitted variable on the outcome of interest.  In this paper, we develop confidence intervals (CIs) with desirable properties for these types of settings.

More specifically, we develop CIs for a parameter of interest in the presence of nuisance parameters with a known sign.  Our CIs are designed to have uniformly correct (asymptotic) coverage and desirable length properties across the entire parameter space while becoming particularly short when these nuisance parameters are small or zero.  In the regression context, this latter property is motivated by common practical situations for which the researcher believes the regression coefficients on a subset of control variables with known partial effect directions are likely to be small or zero.  In general, our CIs can be used for inference on a parameter in any well-behaved finite-dimensional model with a large-sample normally distributed estimator when some nuisance parameters are restricted above or below by zero (possibly after a location shift). This includes regression models estimated by ordinary, generalized and two stage least squares as well as models with bounded parameter spaces such as (G)ARCH \citep[see e.g.,][]{Bollerslev:86} and random coefficient models \citep[see e.g.,][]{BLP:95,Andrews:99}. Even though the standard ``constrained'' estimator is not normally distributed in large samples when the true parameter vector is at (or close to) the boundary of the parameter space \citep[see e.g.,][]{Andrews:99}, there often exists a ``quasi-unconstrained'' estimator that is \citep{Ket18}. While noting this generality, we mainly focus on regression models for ease of exposition.


To construct our CIs, we use the fact that knowledge of the signs of control variable coefficients, in addition to a standard consistent estimator of the covariance matrix of the underlying coefficient estimates, can be used to  determine the sign of the corresponding omitted variables biases incurred by omitting the corresponding control variables. In turn, standard one-sided CIs for the coefficient of interest based on regressions that omit some of these control variables maintain correct coverage. We show that a particular form of these latter CIs is expected excess length-optimal (among affine CIs\textemdash see Proposition \ref{prop:one-sided optimality} for details) when the corresponding control coefficients are equal to zero. It also has low expected excess length when the control coefficients are close to zero but its expected  excess length grows without bound as the control coefficients grow larger.  On the other hand, we show that standard one-sided CIs based upon the regression including all controls have the minimal maximum expected excess length (among affine CIs) over the parameter space that imposes the sign of the control coefficients.\footnote{In fact, we show that these two types of CIs are optimal at each quantile of the excess length distribution greater than one minus their nominal coverage probabilities.  See Proposition \ref{prop:one-sided optimality} below for details.}  They also have correct coverage and expected excess length that does not depend upon the true values of the control coefficients.  We propose adaptive one-sided CIs that utilize the strengths of both of these types of CIs by intersecting them.  We make use of the same logic for constructing two-sided CIs essentially by intersecting our lower- and upper- one-sided CIs.

In particular, we propose a computationally trivial method to find the subset of controls that is able to produce the largest expected (excess) length reductions when using this intersection principle.  In addition, the restricted parameter space implies that the coverage of these intersected CIs is lowest at its boundary.  This feature allows us to provide the user a simple means to compute the smallest CI endpoints that yield correct coverage uniformly across the parameter space via response surface regression output, rather than using a conservative Bonferroni correction.  Using our reported response surface regression coefficients, the user can immediately compute these CI endpoints as a function of one or two empirical correlation parameters, depending upon whether they are forming a one- or two-sided CI. A Stata package available in the SSC archive automatically computes the CIs we propose.\footnote{The package name is ``ssci''. Corresponding Matlab code is available on the authors' webpage.}

We show that our proposed CIs are uniformly asymptotically valid and characterize their length properties.  The latter depend upon the correlation structure of the underlying data and the true values of the unknown control coefficients.  For extreme values of correlation between the estimators of the coefficient of interest and sign-restricted controls, the expected (excess) length of our CIs can be close to 100\% smaller than that of standard CIs based upon the regression including all controls. For correlation values more likely to be encountered in practice, these expected (excess) length reductions can still exceed 30\% for commonly used confidence levels.  On the other hand, for a confidence level of 95\%, for example, our proposed two-sided CIs cannot be more than 2.28\% longer than the corresponding standard CI for \emph{any} realization of the data and the expected excess length of our one-sided CIs cannot be more than 3\% longer than that of the corresponding standard CI.

A leading example of where our proposed CIs should prove useful is in the context of factorial (or ``cross-cutting'') designs in field experiments. Take, for example, the $2 \times 2$ factorial design where two treatments are administered independently such that there are three treatment arms, the two ``main'' treatments (separately) and the combination of the two, and a control arm. The corresponding treatment effects can be consistently estimated by OLS using the ``long regression'', i.e., the regression of the outcome variable on a constant and three dummy variables, one for each treatment arm.\footnote{Recently, \cite*{Muralidharan19} have highlighted the importance of using the long regression, as opposed to the ``short'' regression, i.e., the regression of the outcome variable on a constant and a dummy for the treatment of interest, to avoid omitted variable biases and accompanying size distortions.} In many cases, researchers have prior knowledge about the signs of the main treatment effects.  For example, ethics boards for research grants are unlikely to fund experiments unless they are very likely to entail non-negative average treatment effects.  Moreover, experimenters often conduct pilot studies prior to conducting full scale experiments in part to confirm their prior beliefs about the direction of average treatment effects.\footnote{For example, possible negative effects of a treatment may be closely monitored during the pilot phase and, if realized, even lead to early termination of the experiment. Note also that imposing the absence of any negative effects on \textit{individual} participants (that could be associated with the treatment) is stronger than necessary, because our CIs only require sign restrictions on \textit{average} (treatment) effects in this context.} Such prior knowledge can then be used to obtain sign restrictions on the corresponding regression coefficients. In this context, our proposed CIs will be ``short'' if the \textit{estimated} effects of the main treatments are small and/or have the ``wrong'' sign, which is likely to occur if the unknown \textit{population} effects are small or zero. To illustrate the potential usefulness of our CIs in the context of factorial designs, we revisit \cite*{blattman2017} who study the effect of ``therapy'' and ``cash'' on violent and criminal behavior in Liberia using a $2 \times 2$ factorial design. Indeed, for some of the treatment effects under study, we find our proposed CIs to be up to 36\% shorter than the corresponding standard CIs.

\subsection{Relationship with the Literature}

Several results in the statistics and econometrics literatures provide bounds on the ability for CIs to simultaneously maintain uniformly correct coverage over a class of data-generating processes (DGPs) while adapting to a given subclass.  Here we develop CIs with this very goal in mind:~our CIs maintain correct coverage for the parameter of interest uniformly across the parameter space for the nuisance parameters while becoming shorter when these nuisance parameters are equal to zero.  Although most of this literature is devoted to nonparametric methods (e.g.,~\citealp{Low97}; \citealp{CL04}), the recent work of \cite{AK18} has produced similar implications for parametric models like those in the asymptotic versions of the problems we study.  Indeed, \cite{AKK20} provide bounds on the ability to shorten CIs while maintaining correct coverage for regression coefficients at points for which potential control coefficients are zero.  However, all of the aforementioned results rely upon an assumption of symmetry about zero for the underlying parameter space (among others).  Because we are interested in problems with sign-restricted nuisance parameters, the underlying parameter space is asymmetric and these results do not apply, allowing for us to achieve the goal of constructing CIs that become significantly shorter at empirically-relevant parameter values.

Depending on the application, it may, of course, be possible that a researcher has prior knowledge on the magnitude of the control variables' coefficients rather than their sign. In this case, the recent work by \cite{AKK20} can be employed \citep[see also][]{LM21}. Indeed, \cite*{Muralidharan19} study and suggest (among others) the CI proposed by \cite{AKK20} as a means to improve over standard CIs in the context of factorial designs. In particular, they argue that researchers may, depending on the application, be willing to assume prior knowledge of the maximum (absolute) value of an ``interaction effect'', e.g., the effect of providing two treatments jointly minus the sum of the two main treatment effects. Here, we provide complementary results to be applied in settings for which it is natural for researchers to know the direction, rather than the magnitude, of control variables' coefficients.


This is certainly not the first paper to produce CIs that adapt to subclasses of DGPs while retaining uniform control of coverage probability.  Several authors have provided such adaptive CIs for various smoothness classes and shape constraints in the nonparametric literature.  See, e.g.,~\cite{CL04}, \cite{CLX13}, \cite{Arm15}, \cite{KK20} and \cite{KK20b}.  Given our focus on finite-dimensional models, we are not concerned with the rate of convergence adaptation in this literature but rather finite-sample length adaptation for CIs.  Nevertheless, our CIs share some similarities with some of the CIs in this literature.  Like the ones we propose, the CIs of \cite{CL04}, \cite{KK20} and \cite{KK20b} are obtained by intersecting CIs that are optimal under different subclasses of DGPs.    Within this literature, \cite{KK20} and \cite{KK20b} are probably the closest studies to ours as they focus on nonparametric regression models with coordinate-wise monotone regression functions.  In addition, both \cite{KK20} and \cite{KK20b} provide a means of shortening adaptive nonparametric CIs relative to simple Bonferroni corrections in a similar spirit to our CI endpoints computed from response surface regression output.  However, in contrast to the existing literature on adaptive CIs for nonparametric models, we prove our CIs are uniformly asymptotically valid without assuming Gaussian disturbances or fixed regressors.

Finally, our work is related to the literature on uniform inference when nuisance parameters may be at or near a boundary, e.g.,~\cite{AG09}, \cite{McC17} and \cite{Ket18}.  While CIs with uniform asymptotic validity could in principle be computed by inverting the tests in this literature, this is often computationally prohibitive, especially when the nuisance parameter exceeds one or two dimensions.  Similarly, inverting weighted average power maximizing tests such as those of \cite{MM13} or \cite{EMW15} is computationally intractable for most realistic applications.  In contrast, our CIs are direct and trivial to compute since they do not rely on test inversion.  Moreover, our CIs are designed to have length properties that are desirable from a practical perspective without requiring the user to specify weights or tuning parameters to optimize over.

\subsection{Outline of Paper}

The remainder of this paper is organized as follows.  Section \ref{NMLSP} imparts the basic intuition of our CI constructions in a stylized asymptotic version of the inference problem we consider before providing computationally trivial algorithms for constructing the one- and two-sided CIs we propose in the general asymptotic setting.  Section \ref{FS} then shows how our CIs are constructed in practical finite-sample applications and provides theoretical results establishing their uniform asymptotic validity across a wide variety of applications.  In Section \ref{EA}, we illustrate the usefulness of our CIs in an empirical application of inference on treatment effects in a factorial design field experiment while Section \ref{Sims} examines their finite-sample properties in a simulation study calibrated to the empirical application.  Appendix \ref{sec:proofs} provides the mathematical proofs of our theoretical results and Appendix \ref{sec:reg_space} specifies a parameter space for the standard linear regression model that satisfies the requirements for some of our theoretical results.  Appendix \ref{AT} contains additional tables referenced in the text while the online supplemental appendix provides details on the numerical computations underlying some of the results in this paper.

Throughout this paper, we use the following notational conventions. For any two column vectors $a$ and $b$, we sometimes write $(a, b)$ instead of $(a', b')'$ and let $a \geq b$ denote the element-by-element inequality. Let $\mathbb{R}_{+} = [0,\infty)$, $\mathbb{R}_{+,\infty}=\mathbb{R}_{+}\cup\{\infty\}$, $\mathbb{R}_{\infty}=\mathbb{R}\cup\{\infty\}\cup\{-\infty\}$ and $z_{\xi}$ denote the $\xi^{th}$ quantile of the standard normal distribution.  For a square matrix $A$, $\diag(A)$ denotes the diagonal matrix with the same diagonal entries as $A$ and $\lambda_{\min}(A)$ and $\lambda_{\max}(A)$ denote its smallest and largest eigenvalues.

\section{Normal Means Large Sample Problem} \label{NMLSP}

LeCam's Limits of Experiments Theory provides that inference on the parameter of a well-behaved model is equivalent to inference on the mean of a Gaussian random vector with known variance matrix in large samples.  This powerful result incorporates regression models, instrumental variables models, maximum likelihood models and models estimated by the generalized method of moments.\footnote{This result holds under assumptions ensuring the model is well-behaved, amounting to the existence of an asymptotically normally distributed estimator for the (finite-dimensional) parameter vector in our context. For example, in the context of instrumental variables models or models estimated by the generalized method of moments, this does not allow for weak instruments or other forms of weak identification.  As alluded to in the Introduction, one can use the results of \cite{Ket18} to obtain this limit experiment even for models that may not be defined outside the parameter space, such as the random coefficients logit \citep{BLP:95} and (G)ARCH models for which variance parameters must be non-negative.  Indeed, such models provide other natural applications for the CIs we introduce in this paper.} See Chapter 9 of \cite{vdV98} and Chapter 13 of \cite{LR05} for textbook treatments of this theory. 
Since the variance matrix is known in this setting, each element of the Gaussian random vector can be scale-normalized so that the large sample inference problem reduces to inference on the mean vector $h$ from a single observation $Y\overset{d}\sim \mathcal{N}(h,\Omega)$, where $\Omega$ is a known correlation matrix.

It is often the case in econometric applications that the researcher is interested in constructing a CI for a scalar parameter of interest in the presence of nuisance parameters.  In addition, the researcher often has knowledge about the sign of the nuisance parameters.  For example, when performing inference on a single coefficient in the linear regression model when ``control'' variables may be included in the regression to mitigate potential omitted variable bias, the researcher often knows the direction of the partial effects of some of the controls from economic theory or logical reasoning.  In the large sample problem, this corresponds to conducting inference on a scalar $\beta$ from a single observation
\begin{equation}
\left(\begin{array}{c}
Y_{\beta} \\
Y_{\delta}
\end{array}\right)\sim \mathcal{N}
\left(\left(\begin{array}{c}
\beta \\
\delta
\end{array}\right), 
\Omega\right), \label{Limit Exp}
\end{equation}
where $\Omega$ is a known positive-definite correlation matrix and $\delta$ is a finite-dimensional nuisance parameter whose elements are known to be greater than or equal to zero.\footnote{The restriction $\delta\geq 0$ is without loss of generality because parameters without sign restrictions may be dropped from the analysis in the limiting problem and limiting Gaussian random variables corresponding to parameters restricted to be greater/less than or equal to a known number may be linearly transformed to conform to \eqref{Limit Exp}.}

In many contexts, it is natural for the researcher to desire a CI with the following properties:~(i) correct coverage $1-\alpha$ (coverage of at least $1-\alpha$) across the entire $\delta\geq 0$ parameter space, (ii) good length properties across the entire $\delta\geq 0$ parameter space and (iii) shortness when $\delta$ is equal or close to zero.  For example, if it is not obvious whether a regressor should enter as a control variable or not, it is sensible to desire an especially short CI when the unknown population regression coefficient is equal to or near zero (reflecting the researcher's uncertainty about whether it is an important variable) while maintaining correct coverage and decent length no matter the coefficient's magnitude.  In this section, we provide CI constructions for the large sample problem with this very goal in mind.  We begin by describing the intuition for the CIs in the simplest version of the problem and subsequently provide general formulations for both one- and two-sided CIs.

\subsection{Basic Intuition}

To communicate the basic intuition for our CIs, we specialize the large sample problem \eqref{Limit Exp} to the case for which $\delta$ is one-dimensional and the correlation between $Y_\beta$ and $Y_{\delta}$ is positive:
\begin{equation*}
\left(\begin{array}{c}
Y_{\beta} \\
Y_{\delta}
\end{array}\right)\sim \mathcal{N}
\left(\left(\begin{array}{c}
\beta \\
\delta
\end{array}\right), 
\left(\begin{array}{cc}
1 & \rho \\
\rho & 1
\end{array}\right)\right), 
\end{equation*}
where $\rho> 0$, $\beta$ is unrestricted, and $\delta \geq 0$.  Consider the formation of an upper one-sided CI for $\beta$ with the goal of satisfying properties (i)--(iii) above.  To illustrate the tension between properties (ii) and (iii), note that the standard CI that ignores the information in $Y_{\delta}$, i.e., 
$$CI_u(Y_{\beta})=[Y_{\beta}-z_{1-\alpha},\infty),$$ 
satisfies (ii) but not (iii) since its expected excess length is always simply equal to $z_{1-\alpha}$.\footnote{Expected excess length of an upper one-sided CI for $\beta$ is defined as $E[\beta - \text{lb}]$, where lb denotes the lower bound of the CI.}  On a more technical level, we show in Proposition \ref{prop:one-sided optimality}(i) below that this CI achieves the minimal maximum excess length quantile for all quantiles larger than $\alpha$ across the $\delta\geq 0$ parameter space.  That is, the standard CI is minimax for the problem we are interested in.  On the other hand,  Proposition \ref{prop:one-sided optimality}(ii) below shows that the CI that is excess length-optimal for all excess length quantiles larger than $\alpha$ when $\delta$ is known to equal zero is equal to~
\[\widetilde{CI}_u(Y_{\beta},\rho Y_{\delta})=\left[Y_\beta-\rho Y_\delta-\sqrt{1-\rho^2}z_{1-\alpha},\infty\right).\]
This CI satisfies (iii) but not (ii) since its expected excess length is equal to $\rho \delta +\sqrt{1-\rho^2}z_{1-\alpha}$, which diverges as $\delta \to \infty$.\footnote{Both CIs satisfy (i) in this context since we have assumed $\rho> 0$, see equation \eqref{adaptive cov}.}

In order to attain property (iii) but not at the expense of property (ii), we propose CIs with length performance designed to adapt  to the data.  Consider intersecting the two CIs $CI_u(Y_{\beta})$ and $\widetilde{CI}_u(Y_{\beta},\rho Y_{\delta})$ to simultaneously retain property (ii) of the former and property (iii) of the latter:
\begin{align*}
\widehat{CI}_u\left(Y_{\beta},\rho Y_{\delta};z_{1-\alpha+\gamma},\sqrt{1-\rho^2}z_{1-\gamma}\right)&=[Y_{\beta}-z_{1-\alpha+\gamma},\infty)\cap \left[Y_{\beta}-\rho Y_{\delta}-\sqrt{1-\rho^2}z_{1-\gamma},\infty\right) \\
&=\left[Y_{\beta}-\min\left\{z_{1-\alpha+\gamma},\rho Y_{\delta}+\sqrt{1-\rho^2}z_{1-\gamma}\right\},\infty\right)
\end{align*}
for some $\gamma\in(0,\alpha)$.  Note that $\widehat{CI}_u\left(Y_{\beta},\rho Y_{\delta};z_{1-\alpha+\gamma},\sqrt{1-\rho^2}z_{1-\gamma}\right)$ maintains correct coverage probability over the parameter space:
\begin{align}
& P\left(\beta\in\widehat{CI}_u\left(Y_{\beta},\rho Y_{\delta};z_{1-\alpha+\gamma},\sqrt{1-\rho^2}z_{1-\gamma}\right) \right)  \notag \\ =&P\left(\beta\geq Y_{\beta}-\min\left\{z_{1-\alpha+\gamma},\rho Y_{\delta}+\sqrt{1-\rho^2}z_{1-\gamma}\right\}\right) \notag \\
=&1-P\left(\beta< Y_{\beta}-\min\left\{z_{1-\alpha+\gamma},\rho Y_{\delta}+\sqrt{1-\rho^2}z_{1-\gamma}\right\}\right) \notag \\
\geq & 1-P(\beta<Y_{\beta}-z_{1-\alpha+\gamma})-P\left(\beta<Y_{\beta}-\rho Y_{\delta}-\sqrt{1-\rho^2}z_{1-\gamma}\right) \notag \\
\geq & 1-(\alpha-\gamma)-\gamma=1-\alpha \label{adaptive cov}
\end{align}
for all $(\beta,\delta)\in\mathbb{R}\times \mathbb{R}_+$, where the first inequality follows from the Bonferroni inequality and the second inequality uses the fact that
\begin{align*}
P\left(\beta<Y_{\beta}-\rho Y_{\delta}-\sqrt{1-\rho^2}z_{1-\gamma}\right)&=P\left(\beta<\beta-\rho\delta+\tilde{Z}_\rho-\sqrt{1-\rho^2}z_{1-\gamma}\right) \\
&=P\left(\tilde{Z}_\rho<-\rho\delta-\sqrt{1-\rho^2}z_{1-\gamma}\right) \leq P\left(\tilde{Z}_\rho<-\sqrt{1-\rho^2}z_{1-\gamma}\right)=\gamma
\end{align*}
with
$$\tilde{Z}_\rho=Y_{\beta}-\rho Y_{\delta}-(\beta-\rho\delta)\overset{d}\sim \mathcal{N}(0,{1-\rho^2}),$$
where the inequality uses the fact that $\rho\delta \geq0$.

Since $\widehat{CI}_u\left(Y_{\beta},\rho Y_{\delta};z_{1-\alpha+\gamma},\sqrt{1-\rho^2}z_{1-\gamma}\right)$ makes use of a multiplicity correction based upon the Bonferroni bound, for similar reasons used to motivate the adjusted Bonferroni critical values of \cite{McC17}, it is possible to decrease the excess length of $\widehat{CI}_u\left(Y_{\beta},\rho Y_{\delta};z_{1-\alpha+\gamma},\sqrt{1-\rho^2}z_{1-\gamma}\right)$ while retaining uniform control of coverage probability.  In particular, fix $\gamma\in(0,\alpha)$ and find the constant $c^*\in [0,\sqrt{1-\rho^2}z_{1-\gamma}]$ that solves
\begin{equation}
P\left(Z_1>\min\left\{z_{1-\alpha+\gamma},\rho Z_2+c\right\}\right)=\alpha \label{MC-prob}
\end{equation}
in $c$, where 
\begin{equation*}
\left(\begin{array}{c}
Z_1 \\
Z_2
\end{array}\right)\sim \mathcal{N}
\left(\left(\begin{array}{c}
0 \\
0
\end{array}\right), 
\left(\begin{array}{cc}
1 & \rho \\
\rho & 1
\end{array}\right)\right).
\end{equation*}
The CI $\widehat{CI}_u(Y_{\beta},\rho Y_{\delta};z_{1-\alpha+\gamma},c^*)$ is contained in $\widehat{CI}_u\left(Y_{\beta},\rho Y_{\delta};z_{1-\alpha+\gamma},\sqrt{1-\rho^2}z_{1-\gamma}\right)$ and maintains correct coverage probability over the parameter space:
\begin{align*}
P\left(\beta\in\widehat{CI}_u(Y_{\beta},\rho Y_{\delta};z_{1-\alpha+\gamma},c^*)\right)&=P\left(\beta\geq Y_{\beta}-\min\left\{z_{1-\alpha+\gamma},\rho Y_{\delta}+c^*\right\}\right) \\
&=P\left(Z_1\leq \min\left\{z_{1-\alpha+\gamma},\rho\delta+\rho Z_2+c^*\right\}\right) \\
&\geq P\left(Z_1\leq \min\left\{z_{1-\alpha+\gamma},\rho Z_2+c^*\right\}\right) \\
&=1-P\left(Z_1> \min\left\{z_{1-\alpha+\gamma},\rho Z_2+c^*\right\}\right)=1-\alpha
\end{align*}
for all $(\beta,\delta)\in\mathbb{R}\times\mathbb{R}_+$, where the second equality follows from the fact that $(Y_{\beta},Y_{\delta})\overset{d}\sim(\beta,\delta)+(Z_1,Z_2)$ and the inequality again uses the fact that $\rho\delta\geq 0$. The problem \eqref{MC-prob} is computationally straightforward and can, for example, be solved by means of Monte Carlo simulations.  A heuristic approach to choosing the tuning parameter $\gamma$ makes use of similar reasoning to that used to compute the adjusted Bonferroni critical values in \cite{McC17}:~a ``small'' $\gamma$ such as $\gamma=\alpha/10$ yields only slightly higher expected excess length when $\delta$ is ``large'' but significantly lower expected excess length when $\delta$ is ``small''.  

Finally, it is interesting to note that $\widehat{CI}_u(Y_{\beta},\rho Y_{\delta};z_{1-\alpha+\gamma},c^*)$ can be viewed as a CI that results from a model selection procedure \emph{designed for inference}.  In the context of the regression model example, we can view the model selection procedure as follows:
\begin{enumerate}
\item If $Y_{\delta}>(z_{1-\alpha+\gamma}-c^*)/\rho$, construct the CI for $\beta$ from the ``full'' regression using the critical value $z_{1-\alpha+\gamma}$.
\item If $Y_{\delta}\leq(z_{1-\alpha+\gamma}-c^*)/\rho$, construct the CI for $\beta$ from the ``short'' regression using the critical value $c^*$.
\end{enumerate}
The model selection pretest rule $Y_{\delta}>(z_{1-\alpha+\gamma}-c^*)/\rho$ is analogous to using a $t$-test as a pretest but with a nonstandard critical value that incorporates both the two-step nature of the inference procedure as well as the dependence between $Y_{\beta}$ and $Y_{\delta}$.  Note that as $\rho\rightarrow 1$, this nonstandard pretest approaches a standard $t$-test pretest.  Unlike standard model selection procedures, this procedure is designed for inference in the sense that (i) it uniformly controls coverage probability by directly incorporating the model selection uncertainty in its construction and (ii) it is designed to yield low excess length rather than a different notion of risk (such as mean-squared error).\footnote{Though some recent post-selection inference procedures (e.g.,~\citealp{BCH14,McC17}) uniformly control coverage probability/size, the selection procedures used  in their construction are not designed to yield CIs with desirable length properties.}

\subsection{One-Sided Confidence Intervals}

In this section, we focus on forming analogous adaptive one-sided CIs but now allowing $\delta\geq 0$ to be multidimensional so that the large sample problem corresponds to \eqref{Limit Exp}, where  
\begin{equation}
\Omega=\left(\begin{array}{cc}
1 & \Omega_{\beta\delta} \\
\Omega_{\delta\beta} & \Omega_{\delta\delta}
\end{array}\right). \label{block correlation mtx}
\end{equation}
Without loss of generality, we focus on upper one-sided CIs for $\beta$ since lower one-sided CIs may be attained analogously upon multiplying $Y_\beta$ by negative one.
The optimal $(1-\alpha)$-level upper one-sided CI for $\beta$ when $\delta=0$ is equal to
\[
	\left[Y_\beta - \Omega_{\beta\delta} \Omega_{\delta\delta}^{-1} Y_\delta - z_{1-\alpha}\sqrt{1-\Omega_{\beta\delta} \Omega_{\delta\delta}^{-1}\Omega_{\delta\beta} },\infty \right).
\]
The CI that intersects this CI with the standard CI for $\beta$ that ignores the information in $Y_\delta$ will not maintain coverage in general.  More specifically, the argument in \eqref{adaptive cov} for showing correct coverage only generalizes when all of the elements of $\Omega_{\beta\delta} \Omega_{\delta\delta}^{-1}$ are non-negative.  In the case that this condition does not hold, we can still find adaptive CIs with potential length improvements by ``dropping'' elements of $Y_{\delta}$ from consideration.  The following algorithm is designed to do just that while maintaining particularly low excess length when $\delta$ is equal or close to zero.  

For $\gamma\in(0,\alpha)$, consider the function $c:[0,1)\rightarrow [0,z_{1-\gamma}]$ such that 
\begin{equation}
P(Z_1>\min\{z_{1-\alpha+\gamma},\tilde Z_2+c(\omega)\})=\alpha, \label{c-function def}
\end{equation}
where
\begin{equation*}
\left(\begin{array}{c}
Z_1 \\
\tilde Z_2
\end{array}\right)\sim \mathcal{N}
\left(0, 
\left(\begin{array}{cc}
1 & \omega \\
\omega & \omega
\end{array}\right)\right).
\end{equation*}
The following result ensures that $c:[0,1)\rightarrow [0,z_{1-\gamma}]$ is well-defined and continuous.

\begin{proposition} \label{prop:c existence}
For $\alpha\in(0,1/2)$, $c:[0,1)\rightarrow[0,z_{1-\gamma}]$ as defined in \eqref{c-function def} exists and is continuous.
\end{proposition}

Note that $c(0)=z_{1-\alpha}$.  Let $Y_{\delta}^{(s)}$ denote an arbitrary subvector of $Y_{\delta}$, including the empty one, with
\begin{equation*}
\left(\begin{array}{c}
Y_{\beta} \\
Y_{\delta}^{(s)}
\end{array}\right)\sim \mathcal{N}
\left(\left(\begin{array}{c}
\beta \\
\delta^{(s)}
\end{array}\right), 
\left(\begin{array}{cc}
1 & \Omega_{\beta\delta^{(s)}} \\
\Omega_{\delta^{(s)}\beta} & \Omega_{\delta^{(s)}\delta^{(s)}}
\end{array}\right)\right), 
\end{equation*} 
where by convention $\delta^{(s)}$, $\Omega_{\beta\delta^{(s)}}$ and $\Omega_{\delta^{(s)}\delta^{(s)}}$ (as well as $\Omega_{\beta\delta^{(s)}} \Omega_{\delta^{(s)}\delta^{(s)}}^{-1}$ and $\Omega_{\beta\delta^{(s)}} \Omega_{\delta^{(s)}\delta^{(s)}}^{-1}\Omega_{\delta^{(s)}\beta}$) are set equal to zero when $Y_{\delta}^{(s)} = \emptyset$.

\begin{one sided}
\noindent Amongst all subvectors of $Y_{\delta}$ (including the empty one) such that the elements of $\Omega_{\beta\delta^{(s)}} \Omega_{\delta^{(s)}\delta^{(s)}}^{-1}$ are non-negative, find the subvector $Y_{\delta}^{(s^*)}$ such that the expected excess length of 
\[
 \widehat{CI}_u(Y_{\beta},\Omega_{\beta\delta^{(s)}}\Omega_{\delta^{(s)}\delta^{(s)}}^{-1}Y_{\delta}^{(s)};z_{1-\alpha+\gamma},c\left(\Omega_{\beta\delta^{(s)}} \Omega_{\delta^{(s)}\delta^{(s)}}^{-1}\Omega_{\delta^{(s)}\beta}\right))
\]
at $\delta = 0$ is minimized at $s=s^*$. Then, construct
\begin{equation}
\widehat{CI}_u(Y_{\beta},\Omega_{\beta\delta^{(s^*)}}\Omega_{\delta^{(s^*)}\delta^{(s^*)}}^{-1}Y_{\delta}^{(s^*)};z_{1-\alpha+\gamma},c\left(\Omega_{\beta\delta^{(s^*)}} \Omega_{\delta^{(s^*)}\delta^{(s^*)}}^{-1}\Omega_{\delta^{(s^*)}\beta}\right)). \tag*{$\blacksquare$}
\end{equation}
\end{one sided}

The goal of this algorithm is to generate short CIs when the user is agnostic about which elements of $\delta$ are more likely to be (close to) zero.  Figure \ref{figure one-sided excess length} shows the expected excess length of $\widehat{CI}_u(Z_1,\tilde Z_2,z_{1-\alpha+\gamma},c(\omega))$ as a function of $\omega$, for $\alpha \in \{0.1, 0.05, 0.01\}$ and our recommended value of $\gamma = \alpha/10$.\footnote{Expected excess length is obtained numerically on the following grid of values: $\omega\in\{0,0.001,0.002,\dots,0.999\}$. See the online supplemental appendix for details.} This expected excess length is strictly decreasing in $\omega$ (at least for the considered choices of $\alpha$ and $\gamma$).  Since it does not depend upon $\beta$, this implies that the expected excess length of $\widehat{CI}_u(Y_{\beta},\Omega_{\beta\delta^{(s)}}\Omega_{\delta^{(s)}\delta^{(s)}}^{-1}Y_{\delta}^{(s)};z_{1-\alpha+\gamma},c\left(\Omega_{\beta\delta^{(s)}} \Omega_{\delta^{(s)}\delta^{(s)}}^{-1}\Omega_{\delta^{(s)}\beta}\right))$ evaluated at $\delta=0$ is smallest for the subvector $Y_\delta^{(s)}$ that maximizes $\Omega_{\beta\delta^{(s)}} \Omega_{\delta^{(s)}\delta^{(s)}}^{-1}\Omega_{\delta^{(s)}\beta}$, leading us to the following simplified algorithm.

\begin{figure}[h] 
  \begin{center}
    \includegraphics[width=70mm]{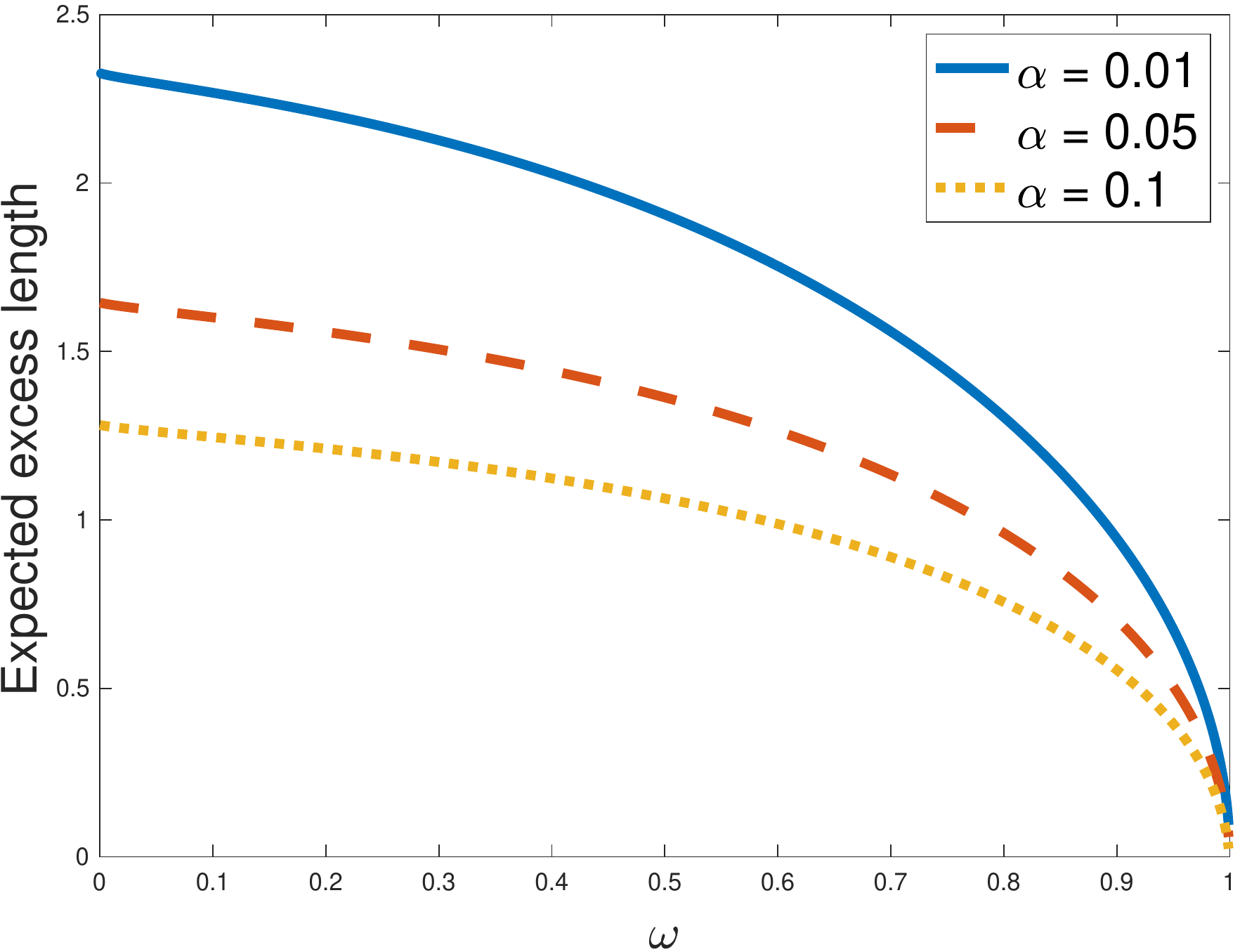}
     \caption{Expected excess length of $\widehat{CI}_u(Z_1,\tilde Z_2,z_{1-\alpha+\gamma},c(\omega))$ as a function of $\omega$, for $\alpha \in \{0.01,0.05,0.1\}$ and $\gamma = \alpha/10$.}
         \label{figure one-sided excess length}
    \end{center}
\end{figure}

\begin{one sided star}
\noindent Amongst all subvectors of $Y_{\delta}$ such that the elements of $\Omega_{\beta\delta^{(s)}} \Omega_{\delta^{(s)}\delta^{(s)}}^{-1}$ are non-negative, find the subvector $Y_{\delta}^{(s^*)}$ such that $\Omega_{\beta\delta^{(s)}} \Omega_{\delta^{(s)}\delta^{(s)}}^{-1}\Omega_{\delta^{(s)}\beta}$ is maximized at $s=s^*$. Then, construct
\begin{equation}
 \widehat{CI}_u^*(Y_{\beta},Y_\delta,\Omega)\equiv \widehat{CI}_u(Y_{\beta},\Omega_{\beta\delta^{(s^*)}}\Omega_{\delta^{(s^*)}\delta^{(s^*)}}^{-1}Y_{\delta}^{(s^*)};z_{1-\alpha+\gamma},c\left(\Omega_{\beta\delta^{(s^*)}} \Omega_{\delta^{(s^*)}\delta^{(s^*)}}^{-1}\Omega_{\delta^{(s^*)}\beta}\right)). \tag*{$\blacksquare$}
\end{equation}
\end{one sided star}

It is worth noting that (i) $c(\omega)$ for $\omega\in(0,1)$ is very simple to compute via Monte Carlo simulation, while $c(0) = z_{1-\alpha}$, and (ii)  Algorithm One-Sided* only requires one to evaluate the function $c(\cdot)$ at the single point $\Omega_{\beta\delta^{(s^*)}} \Omega_{\delta^{(s^*)}\delta^{(s^*)}}^{-1}\Omega_{\delta^{(s^*)}\beta}$.  Therefore, the algorithm carries very low computational cost.  We also note that $c\left(\Omega_{\beta\delta^{(s^*)}} \Omega_{\delta^{(s^*)}\delta^{(s^*)}}^{-1}\Omega_{\delta^{(s^*)}\beta}\right)$ can always be replaced by $\sqrt{1-\Omega_{\beta\delta^{(s^*)}} \Omega_{\delta^{(s^*)}\delta^{(s^*)}}^{-1}\Omega_{\delta^{(s^*)}\beta}}z_{1-\gamma}$ in the algorithm to yield a CI with correct coverage but worse excess length (in analogy with the CI using the Bonferroni correction in the previous section).

In addition to the numerical justification for the selection of the subvector $Y_{\delta}^{(s^*)}$ to form $ \widehat{CI}_u^*(Y_{\beta},Y_\delta,\Omega)$, Algorithm One-Sided* is also justified on theoretical grounds as it entails the intersection of CIs that are optimal over two different subclasses of DGPs like those of e.g.,~\cite{CL04}, \cite{KK20} and \cite{KK20b}.  More specifically, the following proposition applies a general result of \cite{AK18} to the current inference setting to formalize what we mean by ``optimal'' here.

\begin{proposition} \label{prop:one-sided optimality}
For inference on $\beta$ in \eqref{Limit Exp}, the following statements hold for $\alpha\in (0,1)$: 

(i) among all upper one-sided CIs with coverage of at least $(1-\alpha)$ for all $\delta\geq 0$, the CI that minimizes all maximum excess length quantiles over the $\delta\geq 0$ parameter space at quantile levels greater than $\alpha$  is equal to
\[[Y_\beta-z_{1-\alpha},\infty),\]

(ii) among all upper one-sided CIs with coverage of at least $(1-\alpha)$ for all $\delta\geq 0$, the CI that minimizes all excess length quantiles at the point $\delta=0$ and quantile levels greater than $\alpha$  is equal to 
\[
	\left[Y_\beta - \Omega_{\beta\delta^{(s^*)}} \Omega_{\delta^{(s^*)}\delta^{(s^*)}}^{-1} Y_{\delta^{(s^*)}} - z_{1-\alpha}\sqrt{1-\Omega_{\beta\delta^{(s^*)}} \Omega_{\delta^{(s^*)}\delta^{(s^*)}}^{-1}\Omega_{\delta^{(s^*)}\beta} },\infty \right).
\]
\end{proposition}

In particular, this proposition implies that for $\alpha < 1/2$ the CIs in (i) and (ii) above minimize maximum \emph{median} excess length across $\delta\geq 0$ and \emph{median} excess length at $\delta=0$, respectively. Since median and expected excess length coincide for CIs that are affine in the data ($Y$), we also have that the CIs in (i) and (ii) minimize maximum \emph{expected} excess length across $\delta\geq 0$ and \emph{expected} excess length at $\delta=0$ among all affine CIs, respectively, if $\alpha < 1/2$. Algorithm One-Sided* computes a CI that intersects two optimal CIs of these forms while making a non-conservative multiplicity correction that improves upon the conservative Bonferroni adjustement.

As can be seen from Figure \ref{figure one-sided excess length}, extreme values of $\Omega_{\beta\delta^{(s^*)}} \Omega_{\delta^{(s^*)}\delta^{(s^*)}}^{-1}\Omega_{\delta^{(s^*)}\beta}$ such as 0.999 can lead to expected excess lengths of our one-sided CI close to zero, entailing expected excess length reductions of nearly 100\% relative to the expected excess length of the standard CI.  At more empirically-relevant values of $\Omega_{\beta\delta^{(s^*)}} \Omega_{\delta^{(s^*)}\delta^{(s^*)}}^{-1}\Omega_{\delta^{(s^*)}\beta}$, such as say 0.7, Figure \ref{figure one-sided excess length} still implies expected excess length reductions of more than 30\% for $\alpha = 0.05$.  On the other hand, the expected excess length of our one-sided CI is bounded above by $z_{1-\alpha+\gamma}=1.695$ for $\alpha=0.05$, $\gamma=\alpha/10$ and \emph{any} value of $\Omega_{\beta\delta^{(s^*)}} \Omega_{\delta^{(s^*)}\delta^{(s^*)}}^{-1}\Omega_{\delta^{(s^*)}\beta}$.  This implies that the expected excess length of our recommended CI relative to the standard one-sided CI is bounded above by $z_{1-\alpha+\gamma}/z_{1-\alpha}=1.695/1.645\approx 1.03$ for $\alpha=0.05$.  That is the expected excess length increase of our recommended CI relative to the standard CI is bounded above by roughly 3\% for $\alpha=0.05$ and any value of $\Omega_{\beta\delta^{(s^*)}} \Omega_{\delta^{(s^*)}\delta^{(s^*)}}^{-1}\Omega_{\delta^{(s^*)}\beta}$.\footnote{For $\alpha$ equal to 0.01 and 0.1 (and $\gamma=\alpha/10$), the expected excess length of our one-sided CI is bounded above by 2.366 and 1.341 and the expected excess length of the standard one-sided CI is equal to 2.326 and 1.282, respectively. This implies that the expected excess length increase of our CI relative to the standard CI is bounded above by 1.7\% and 4.6\%, for $\alpha=0.01$ and $\alpha=0.1$ respectively.}

\begin{figure}[h] 
  \begin{center}
    \includegraphics[width=70mm]{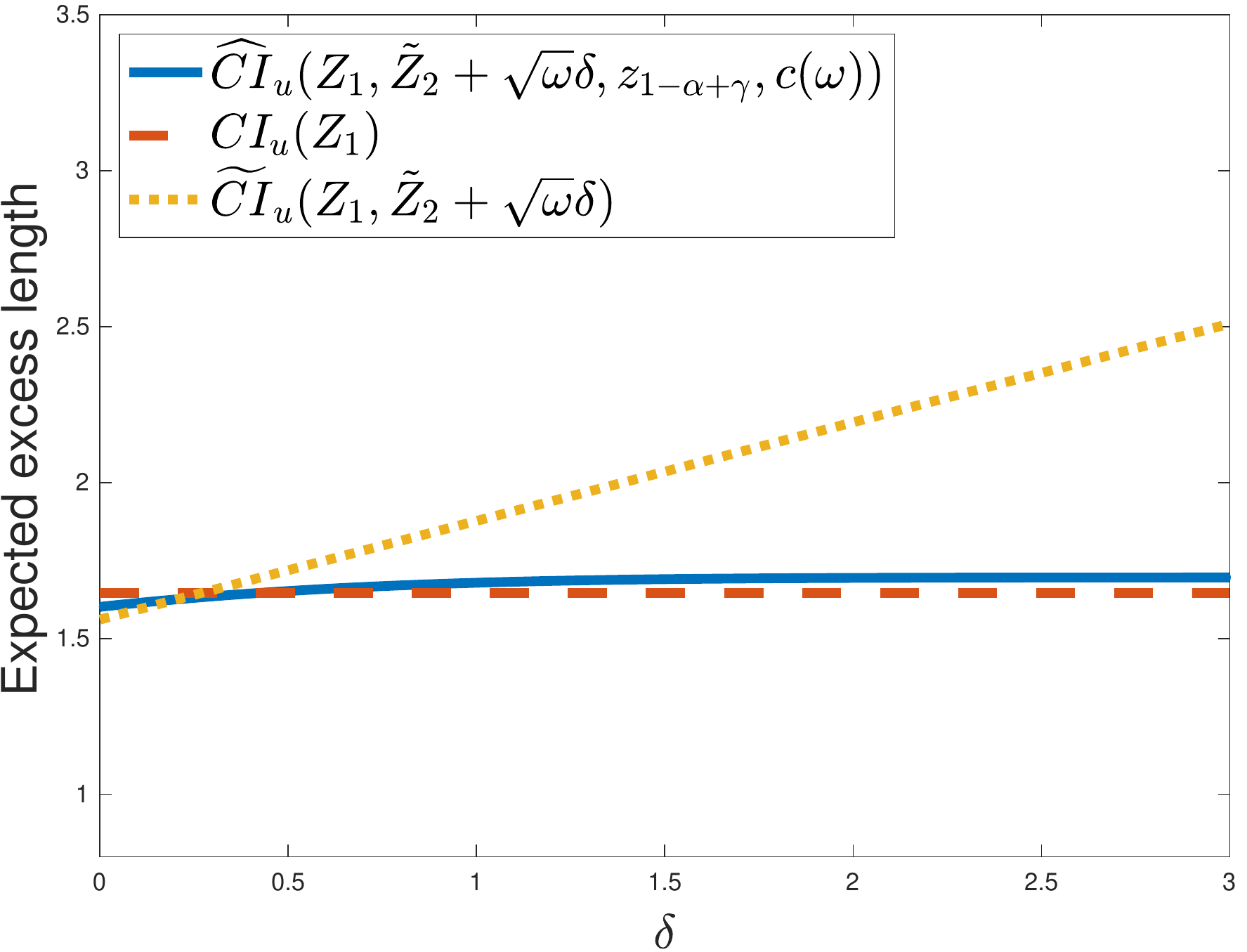}
    \includegraphics[width=70mm]{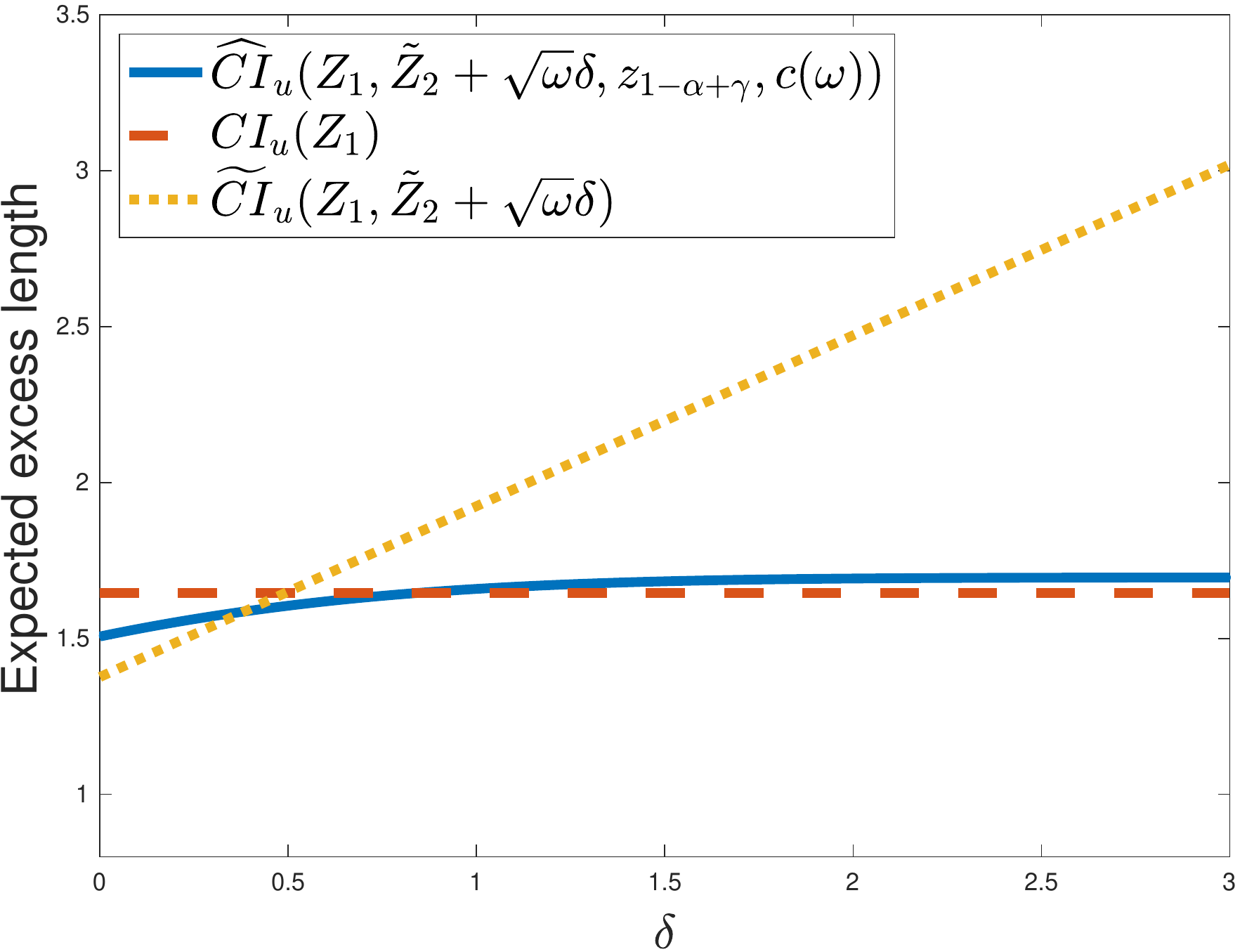}\\
    \includegraphics[width=70mm]{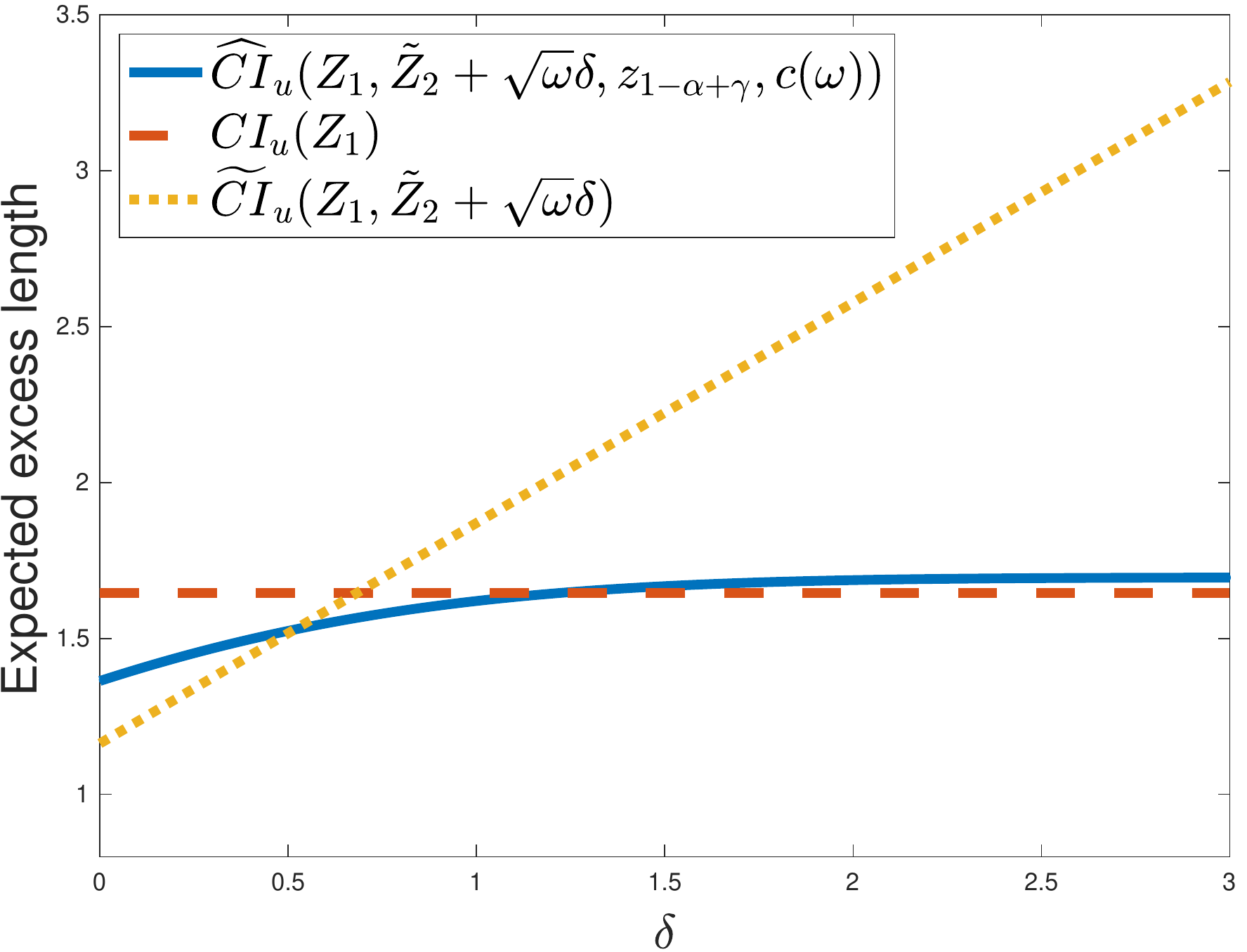}
    \includegraphics[width=70mm]{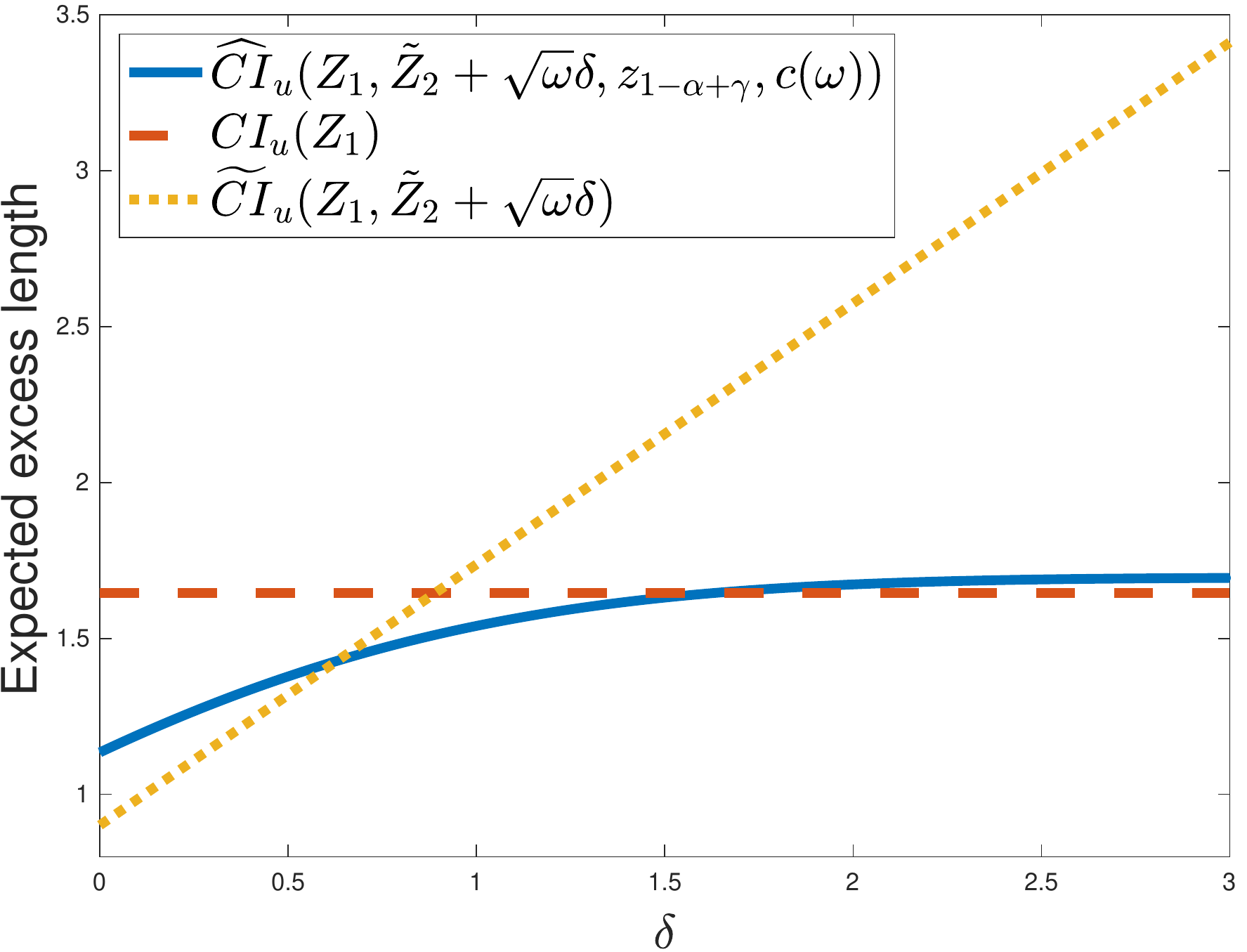}\\
     \caption{Expected excess length of $\widehat{CI}_u(\cdot)$, ${CI}_u(\cdot)$, and $\widetilde{CI}_u(\cdot)$ as a function of $\delta$ for $\omega$ equal to 0.1, 0.3, 0.5, and 0.7 from left to right and for $\alpha = 0.05$ and $\gamma = \alpha/10$.}
         \label{figure one-sided excess length comparison}
    \end{center}
\end{figure}

Figure \ref{figure one-sided excess length comparison} plots the expected excess length of our one-sided CI ($\widehat{CI}_u(\cdot)$), the minimax standard one-sided CI (${CI}_u(\cdot)$) and the CI that is optimal when $\delta=0$ ($\widetilde{CI}_u(\cdot)$) as a function of $\delta$, for $\alpha = 0.05$ (and $\gamma = \alpha/10$) and several values of $\omega=\Omega_{\beta\delta^{(s^*)}} \Omega_{\delta^{(s^*)}\delta^{(s^*)}}^{-1}\Omega_{\delta^{(s^*)}\beta}$. Similarly to Figure \ref{figure one-sided excess length}, Figure \ref{figure one-sided excess length comparison} shows that the gains in expected excess length of our one-sided CI compared to the standard one-sided CI are more pronounced for larger values of $\omega$. Furthermore, Figure \ref{figure one-sided excess length comparison} illustrates the adaptive nature of our one-sided CI: at the endpoints of the parameter space, $\delta = 0$ and $\delta = \infty$, its expected excess length approaches those of the optimal CI at $\delta=0$ and the minimax standard one-sided CI.  

A different choice of $\gamma$ from $\alpha/10$ would entail different tradeoffs for our CI over the $\delta\geq 0$ parameter space.  For example, a larger choice of $\gamma$ would yield lower expected excess length in a neighborhood of $\delta=0$ by bringing it ``closer'' to the optimal CI at $\delta=0$.  Conversely, such a choice for $\gamma$ would yield higher expected excess length at large values of $\delta$.  In fact, the user of our CI could choose $\gamma$ according to how much of an increase in expected excess length they are willing to tolerate relative to the standard CI at large values of $\delta$ since the ratio of the expected excess length of our CI relative to the standard CI is bounded above by $z_{1-\alpha+\gamma}/z_{1-\alpha}$.  For example, the choice of $\gamma=\alpha/2$ would entail an expected excess length increase relative to the standard CI bounded above by about 19\% for $\alpha=0.05$ since $z_{1-\alpha+\gamma}/z_{1-\alpha}=1.96/1.645\approx 1.19$.  Analogous implications for the choice of $\gamma$ apply to our two-sided CI constructions below.

\begin{table}[h!]											
\begin{center}									
\caption{Coefficients for 6$^\text{th}$ order polynomial approximations of $c(\omega)$ for $\gamma = \alpha/10$}			
\label{Coefficients_1s_table}									
\begin{tabular}{c|ccccccc}								
\hline								
\hline			
$\alpha$ & 1 & $\omega$ & $\omega^2$ & $\omega^3$ & $\omega^4$ & $\omega^5$ & $\omega^6$  \\ 
\hline
0.01 & 2.3241 & 2.5073 & -19.6229 & 65.0489 & -122.0242 & 112.9814 & -40.9895 \\ 
0.05 &1.6385 & 2.4813 & -16.1007 & 52.6998 & -98.9348 & 91.7646 & -33.3628 \\
0.1 & 1.2726 & 2.4250 & -14.1041 & 46.0326 & -86.7946 & 80.8189 & -29.4840 \\ 
\hline								
\end{tabular}
\end{center}						
\end{table}

In order to make practical implementation computationally trivial for the user, requiring no Monte Carlo simulation, we approximate $c(\omega)$ via a polynomial response surface regression. Table \ref{Coefficients_1s_table} provides the estimated coefficients for a 6$^\text{th}$ order polynomial approximation of $c(\omega)$ for $\alpha \in \{0.01,0.05,0.1\}$ and $\gamma = \alpha/10$. For each value of $\alpha$, the $R^2$ is greater than 0.999.\footnote{The regression is performed on the same grid that underlies Figure \ref{figure one-sided excess length}, i.e., $\omega \in \{0,0.001,0.002,\dots,0.999\}$.} Nevertheless, using the response surface approximation of $c(\omega)$ in the construction of our one-sided CI can lead to small (asymptotic) coverage distortions: A grid search over $\{0,0.001,0.002,\dots,0.999\}$ reveals a minimum coverage probability of 94.80\% when $\alpha = 0.05$.\footnote{The corresponding values for $\alpha = 0.01$ and $\alpha = 0.1$ are 98.94\% and 89.68\%, respectively.} However in finite-sample practical applications, the coverage distortions resulting from using the response surface approximation tend to be smaller than those induced by using the large-sample normal approximation to conduct inference.  Evidence of this can be found in the simulation results of Section \ref{Sims} for which the finite-sample coverage distortions of our CIs that are implemented using the response surface approximation are very close to those of the standard CI that does not rely on such an approximation.  Thus from a practical perspective, the small size distortions arising from the response surface approximation are insignificant.

\subsection{Two-Sided Confidence Intervals}

In this section, we focus on forming analogous adaptive two-sided CIs allowing $\delta\geq 0$ to be multidimensional in the large sample problem characterized by \eqref{Limit Exp} and \eqref{block correlation mtx}.  These two-sided CIs use the same basic logic as the one-sided CIs of the previous section but work to shorten each side of the CI separately while maintaining correct coverage.  Unlike our one-sided CIs that use a single critical value $c\left(\Omega_{\beta\delta^{(s^*)}} \Omega_{\delta^{(s^*)}\delta^{(s^*)}}^{-1}\Omega_{\delta^{(s^*)}\beta}\right)$, two-sided CIs are formed using two, one corresponding to the upper bound of the interval and the other corresponding to the lower bound.  We address this additional degree of freedom by choosing these upper and lower critical values to minimize expected length at $\delta=0$, subject to the constraint that the resulting two-sided CI maintains correct coverage across the $\delta\geq 0$ parameter space.  The following algorithms provide the details.

Let 
\begin{equation}
\left(\begin{array}{c}
Z_1 \\
\tilde Z_2 \\
\tilde Z_3
\end{array}\right)\sim \mathcal{N}
\left(0, 
\left(\begin{array}{ccc}
1 & \omega_{12} & \omega_{13} \\
\omega_{12} & \omega_{12} & \omega_{23} \\
 \omega_{13}  & \omega_{23} & \omega_{13} 
\end{array}\right)\right) \label{Z-dist}
\end{equation} 
and $\widetilde{\mathcal{C}}=\{(c_u,\tilde\omega)\in\mathbb{R}_{\infty}\times \bar{\mathcal{S}}:c_u\in [\underline{c_u}(\tilde\omega),\infty]\}$, where $\bar{\mathcal{S}} = {\mathcal{S}} \cup \{ (x,y,z)\in \mathbb{R}^3: x \in [0,1), y = z = 0\} \cup \{ (x,y,z)\in \mathbb{R}^3: y \in (0,1), x = z = 0\}$ with $\mathcal{S}=\{(x,y,z)\in \mathbb{R}^3:x,y\in (0,1),-z^2+2xyz + xy - x^2y - xy^2 > 0\}$,\footnote{In terms of arguments $(x,y,z)$, the definition of $\mathcal{S}$ is equivalent to the positive definiteness of the matrix 
$\left(\begin{array}{ccc}
1 & x & y \\
x & x & z \\
 y  & z & y 
\end{array}\right).$} where $\underline{c_u}:\bar{\mathcal{S}}\rightarrow \mathbb{R}$ is implicitly defined by
\begin{equation}
P(- \min \{z_{1-(\alpha-\gamma)/2},-\tilde  Z_3+\underline{c_u}(\tilde\omega)\} \leq Z_1 \leq z_{1-(\alpha-\gamma)/2})=1-\alpha.\label{lower bar c_u}
\end{equation}
For $\gamma\in(0,\alpha)$, consider the function $ \tilde c:\widetilde{\mathcal{C}}\rightarrow \mathbb{R}_\infty$
implicitly defined by
\begin{equation} \label{cl_cc}
	P(- \min \{z_{1-(\alpha-\gamma)/2},-\tilde  Z_3+c_u\} \leq Z_1 \leq \min\{z_{1-(\alpha-\gamma)/2},\tilde  Z_2+\tilde c(c_u,\tilde \omega)\})=1-\alpha
\end{equation}
at points $(c_u,\tilde\omega)\in\widetilde{\mathcal{C}}$ for which $\omega_{12},\omega_{13}\neq 0$.  The domain $\widetilde{\mathcal{C}}$ of $\tilde c(\cdot)$ is defined in terms of the lower bound $\underline{c_u}(\tilde\omega)$ on $c_u$ in \eqref{lower bar c_u} so that for any given $\tilde\omega$, the solution to \eqref{cl_cc} exists.  More specifically, the lower bound $\underline{c_u}(\tilde\omega)$ rules out $c_u$ values that are too small to admit a solution to \eqref{cl_cc}.  Next, for $(c_u,\tilde\omega)\in\widetilde{\mathcal{C}}$ with $\omega_{12}= 0$, define $\tilde{c}(c_u,\tilde\omega)=\lim_{\bar\omega_{12}\rightarrow 0}\tilde{c}(c_u,\bar\omega_{12},\omega_{13},\omega_{23})$ and for $(c_u,\tilde\omega)\in\widetilde{\mathcal{C}}$ with $\omega_{13}= 0$, define $\tilde{c}(c_u,\tilde\omega)=\lim_{\bar\omega_{13}\rightarrow 0}\tilde{c}(c_u,\omega_{12},\bar\omega_{13},\omega_{23})$.\footnote{The limits in these definitions exist by the continuity of $\tilde{c}(c_u,\tilde\omega)$ at all $(c_u,\tilde\omega)\in \widetilde{\mathcal{C}}$ with $\omega_{12},\omega_{13}\neq 0$.  See Lemma \ref{lem:existence and uniqueness of c-tilde} in the Appendix.  We define $\tilde{c}(c_u,\tilde\omega)$ at $\omega_{12}=0$ and $\omega_{13}=0$ in terms of limits because multiple values of $\tilde{c}(c_u,\tilde\omega)$ satisfy \eqref{cl_cc} when $\omega_{12}=0$ and we wish to treat $\omega_{12}$ and $\omega_{13}$ symmetrically in light of Proposition \ref{prop:c_ell} below.}  Finally, define the correspondence $ \tilde c_u:\bar{\mathcal{S}}\rightrightarrows  \mathbb{R}$ as
\begin{equation}
	 \tilde c_u(\tilde \omega) =\argmin_{c_u\in [\underline{c_u}(\tilde\omega),\infty]} E[\max\{\min\{z_{1-(\alpha-\gamma)/2},\tilde  Z_2+ \tilde c(c_u,\tilde \omega)\}+\min \{z_{1-(\alpha-\gamma)/2},-\tilde  Z_3+c_u\},0\}]. \label{c_u def}
\end{equation}
Note that $ \tilde c_u(0)=z_{1-\alpha/2}$.  The following proposition ensures that $ \tilde c_u:\bar{\mathcal{S}}\rightrightarrows  \mathbb{R}$ is well-defined and possesses some desirable properties.\footnote{In our numerical work, we have found the solution to \eqref{c_u def} to be a singleton and $\tilde c_u$ to be a continuous function when $\omega_{12},\omega_{13}\neq 0$.}

\begin{proposition} \label{prop:c_u existence}
For any $\tilde\omega\in\bar{\mathcal{S}}$, $\tilde c_{u}(\tilde\omega)\subset \mathbb{R}_{\infty}$ defined in \eqref{c_u def} is non-empty and compact and $\tilde c_{u}:\bar{\mathcal{S}}\rightrightarrows\mathbb{R}_\infty$ is upper hemicontinuous.
\end{proposition}

Now, let $Y_{\delta}^{(s_1)}$ and $Y_{\delta}^{(s_2)}$ denote two arbitrary (possibly empty) subvectors of $Y_{\delta}$ with
\begin{equation*}
\left(\begin{array}{c}
Y_{\beta} \\
Y_{\delta}^{(s_1)} \\
Y_{\delta}^{(s_2)}
\end{array}\right)\sim \mathcal{N}
\left(\left(\begin{array}{c}
\beta \\
 \delta^{(s_1)} \\
 \delta^{(s_2)}
\end{array}\right), 
\left(\begin{array}{ccc}
1 & \Omega_{\beta \delta^{(s_1)}} & \Omega_{\beta \delta^{(s_2)}} \\
\Omega_{ \delta^{(s_1)}\beta} & \Omega_{ \delta^{(s_1)} \delta^{(s_1)}} & \Omega_{ \delta^{(s_1)} \delta^{(s_2)}} \\
\Omega_{ \delta^{(s_2)}\beta} & \Omega_{ \delta^{(s_2)} \delta^{(s_1)}} & \Omega_{ \delta^{(s_2)} \delta^{(s_2)}}
\end{array}\right)\right),
\end{equation*}
where by convention, $\delta^{(s_1)}$ ($\delta^{(s_2)}$), $\Omega_{\beta\delta^{(s_1)}}$ ($\Omega_{\beta\delta^{(s_2)}}$), $\Omega_{\delta^{(s_1)}\delta^{(s_1)}}$ ($\Omega_{\delta^{(s_2)}\delta^{(s_2)}}$), and $\Omega_{\delta^{(s_1)}\delta^{(s_2)}}$, as well as $\Omega_{\beta\delta^{(s_1)}} \Omega_{\delta^{(s_1)}\delta^{(s_1)}}^{-1}$ ($\Omega_{\beta\delta^{(s_2)}} \Omega_{\delta^{(s_2)}\delta^{(s_2)}}^{-1}$), $\Omega_{\beta\delta^{(s_1)}} \Omega_{\delta^{(s_1)}\delta^{(s_1)}}^{-1}\Omega_{\delta^{(s_1)}\beta}$ ($\Omega_{\beta\delta^{(s_2)}} \Omega_{\delta^{(s_2)}\delta^{(s_2)}}^{-1}\Omega_{\delta^{(s_2)}\beta}$), and $\Omega_{\beta \delta^{(s_1)}} \Omega_{ \delta^{(s_1)} \delta^{(s_1)}}^{-1} \Omega_{ \delta^{(s_1)} \delta^{(s_2)}} \Omega_{ \delta^{(s_2)} \delta^{(s_2)}}^{-1}\Omega_{ \delta^{(s_2)}\beta}$, are set equal to zero when $Y_{\delta}^{(s_1)} = \emptyset$ ($Y_{\delta}^{(s_2)} = \emptyset$).  Define
\begin{align*}
&\widehat{CI}_t(Y_{\beta},\Omega_{\beta\delta^{(s_1)}} \Omega_{\delta^{(s_1)}\delta^{(s_1)}}^{-1} Y_\delta^{(s_1)},\Omega_{\beta\delta^{(s_2)}} \Omega_{\delta^{(s_2)}\delta^{(s_2)}}^{-1} Y_\delta^{(s_2)};z_{1-(\alpha-\gamma)/2},c_\ell(\tilde\Omega^{(s_1,s_2)}),c_u(\tilde\Omega^{(s_1,s_2)}))\\ 
&=\left[  Y_{\beta}-\min\left\{z_{1-(\alpha-\gamma)/2},\Omega_{\beta\delta^{(s_1)}} \Omega_{\delta^{(s_1)}\delta^{(s_1)}}^{-1} Y_\delta^{(s_1)} + c_\ell(\tilde\Omega^{(s_1,s_2)})\right\}, \right. \\
& \quad \left.  Y_{\beta}+\min\left\{z_{1-(\alpha-\gamma)/2},-\Omega_{\beta\delta^{(s_2)}} \Omega_{\delta^{(s_2)}\delta^{(s_2)}}^{-1} Y_\delta^{(s_2)} + c_u(\tilde\Omega^{(s_1,s_2)})\right\} \right],
\end{align*}
where $c_u(\tilde\omega)\in \tilde c_u(\tilde\omega)$, $c_\ell(\tilde\omega)= \tilde c(c_u(\tilde\omega),\tilde\omega)$ and
{\small\[\tilde\Omega^{(s_1,s_2)}=({\Omega_{\beta \delta^{(s_1)}} \Omega_{ \delta^{(s_1)} \delta^{(s_1)}}^{-1}\Omega_{ \delta^{(s_1)}\beta}},{\Omega_{\beta \delta^{(s_2)}} \Omega_{ \delta^{(s_2)} \delta^{(s_2)}}^{-1}\Omega_{ \delta^{(s_2)}\beta}},\Omega_{\beta \delta^{(s_1)}} \Omega_{ \delta^{(s_1)} \delta^{(s_1)}}^{-1} \Omega_{ \delta^{(s_1)} \delta^{(s_2)}} \Omega_{ \delta^{(s_2)} \delta^{(s_2)}}^{-1}\Omega_{ \delta^{(s_2)}\beta}).\]
}  
For any given pair of subvectors $Y_\delta^{(s_1)}$ and $Y_\delta^{(s_2)}$ such that all elements of $\Omega_{\beta \delta^{(s_1)}} \Omega_{ \delta^{(s_1)} \delta^{(s_1)}}^{-1}$ are non-negative and all elements of $\Omega_{\beta \delta^{(s_2)}} \Omega_{ \delta^{(s_2)} \delta^{(s_2)}}^{-1}$ are non-positive, $c_\ell(\tilde\Omega^{(s_1,s_2)})$ and $c_u(\tilde\Omega^{(s_1,s_2)})$ minimize the expected length at $\delta=0$ of CIs of the form 
\[\widehat{CI}_t(Y_{\beta},\Omega_{\beta\delta^{(s_1)}} \Omega_{\delta^{(s_1)}\delta^{(s_1)}}^{-1} Y_\delta^{(s_1)},\Omega_{\beta\delta^{(s_2)}} \Omega_{\delta^{(s_2)}\delta^{(s_2)}}^{-1} Y_\delta^{(s_2)};z_{1-(\alpha-\gamma)/2},c_\ell,c_u)\] 
amongst all $(c_\ell,c_u)$ values that have valid coverage for all $\delta\geq 0$. To impart intuition, note that the two subvectors $Y_\delta^{(s_1)}$ and $Y_\delta^{(s_2)}$ play separate roles in our two-sided CI construction:~the lower (upper) bound of the CI is a function of $Y_\delta^{(s_1)}$ ($Y_\delta^{(s_2)}$) so that small values of $Y_\delta^{(s_1)}$ ($Y_\delta^{(s_2)}$) serve to increase (decrease) the lower (upper) bound of $\widehat{CI}_t(\cdot)$.

\begin{two sided}
\noindent Amongst all pairs of subvectors of $Y_{\delta}$ (including the empty ones) such that all elements of $\Omega_{\beta \delta^{(s_1)}} \Omega_{ \delta^{(s_1)} \delta^{(s_1)}}^{-1}$ are non-negative and all elements of $\Omega_{\beta \delta^{(s_2)}} \Omega_{ \delta^{(s_2)} \delta^{(s_2)}}^{-1}$ are non-positive, find the subvector pair $Y_{\delta}^{(s_1^*)}$ and $Y_{\delta}^{(s_2^*)}$ such that the expected length of 
\[
\widehat{CI}_t(Y_{\beta},\Omega_{\beta\delta^{(s_1)}} \Omega_{\delta^{(s_1)}\delta^{(s_1)}}^{-1} Y_\delta^{(s_1)},\Omega_{\beta\delta^{(s_2)}} \Omega_{\delta^{(s_2)}\delta^{(s_2)}}^{-1} Y_\delta^{(s_2)};z_{1-(\alpha-\gamma)/2},c_\ell(\tilde\Omega^{(s_1,s_2)}),c_u(\tilde\Omega^{(s_1,s_2)}))
\]
at $\delta = 0$ is minimized at $s_1=s_1^*$ and $s_2=s_2^*$. Then, construct
\begin{equation}
\widehat{CI}_t(Y_{\beta},\Omega_{\beta\delta^{(s^*_1)}} \Omega_{\delta^{(s^*_1)}\delta^{(s^*_1)}}^{-1} Y_\delta^{(s^*_1)},\Omega_{\beta\delta^{(s^*_2)}} \Omega_{\delta^{(s^*_2)}\delta^{(s^*_2)}}^{-1} Y_\delta^{(s^*_2)};z_{1-(\alpha-\gamma)/2},c_\ell(\tilde\Omega^{(s_1^*,s_2^*)}),c_u(\tilde\Omega^{(s^*_1,s^*_2)})). \tag*{$\blacksquare$}
\end{equation}
\end{two sided}


\begin{figure}[h] 
  \begin{center}
    \includegraphics[width=70mm]{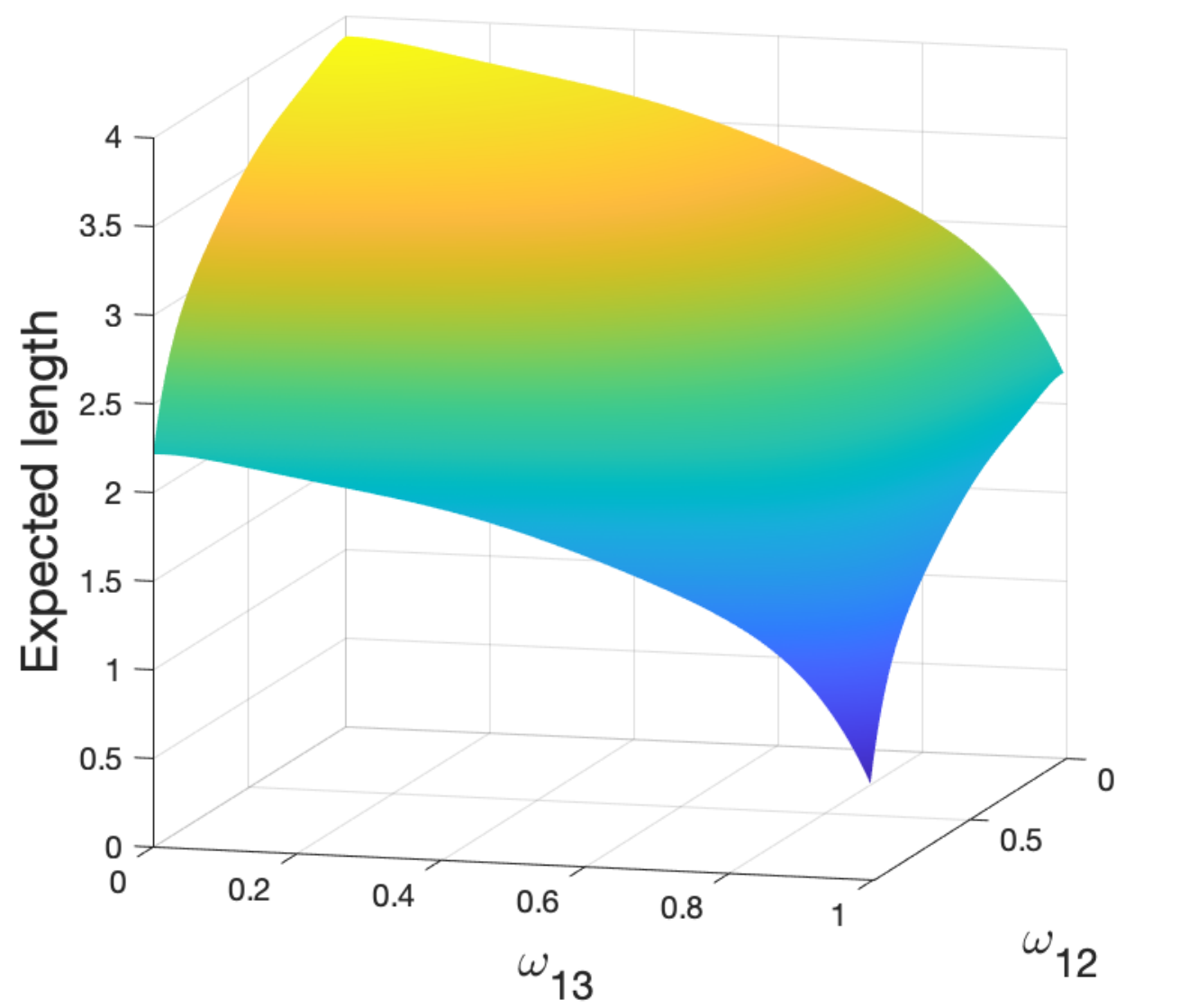}
     \caption{Expected length as a function of $\omega_{12}$ and $\omega_{13}$, for $\alpha = 0.05$ and $\gamma = \alpha/10$.}
         \label{el}
    \end{center}
\end{figure}

Figure \ref{el} shows the fitted surface of a 6$^\text{th}$ order polynomial regression of the expected length of $\widehat{CI}_t(Z_1,\tilde Z_2, \tilde Z_3;z_{1-(\alpha-\gamma)/2},c_\ell(\tilde\omega),c_u(\tilde \omega))$ on $\omega_{12}$ and $\omega_{13}$ alone, for $\alpha = 0.05$ and $\gamma = \alpha/10$. The values of $\tilde \omega$ on which this regression is based are given by $\bar{\Omega} = \bar{\mathcal{S}} \cap \mathcal{G}^2 \times -\mathcal{G} \cup \mathcal{G} \cup \{-0.99,-0.98,\dots,0.99\}$, where 
\begin{equation} \label{grid}
{\mathcal{G}} = \{0, 0.005, 0.01, 0.02, \dots, 0.1, 0.15, \dots, 0.9, 0.91,\dots,0.99,0.995\} 
\end{equation}
and $-\mathcal{G} = \{g : -g \in \mathcal{G}\}$.
The corresponding $R^2$ is greater than 0.999, implying that the expected length is nearly invariant to $\omega_{23}$. Similarly, the maximum difference between the largest and the smallest expected length over the set $\bar{\Omega}$ for any given $(\omega_{12}, \omega_{13})$ is equal to 0.0289.  Note also that the fitted expected length in Figure \ref{el} is strictly decreasing in $\omega_{12}$ and $\omega_{13}$.  Since the expected length does not depend upon $\beta$, this implies that the expected length of $\widehat{CI}_t(Y_{\beta},\Omega_{\beta\delta^{(s_1)}} \Omega_{\delta^{(s_1)}\delta^{(s_1)}}^{-1} Y_\delta^{(s_1)},\Omega_{\beta\delta^{(s_2)}} \Omega_{\delta^{(s_2)}\delta^{(s_2)}}^{-1} Y_\delta^{(s_2)};z_{1-(\alpha-\gamma)/2},c_\ell(\tilde\Omega^{(s_1^*,s_2^*)}),c_u(\tilde\Omega^{(s_1,s_2)}))$ evaluated at $\delta=0$ is approximately smallest for the subvectors $Y_\delta^{(s_1)}$ and $Y_\delta^{(s_2)}$ that maximize $\Omega_{\beta\delta^{(s_1)}} \Omega_{\delta^{(s_1)}\delta^{(s_1)}}^{-1}\Omega_{\delta^{(s)}\beta}$ and $\Omega_{\beta\delta^{(s_2)}} \Omega_{\delta^{(s)}\delta^{(s_2)}}^{-1}\Omega_{\delta^{(s_2)}\beta}$, motivating the following simplified algorithm.

\begin{two sided star}
\noindent Amongst all pairs of subvectors of $Y_{\delta}$ (including the empty ones) such that all elements of $\Omega_{\beta \delta^{(s_1)}} \Omega_{ \delta^{(s_1)} \delta^{(s_1)}}^{-1}$ are non-negative and all elements of $\Omega_{\beta \delta^{(s_2)}} \Omega_{ \delta^{(s_2)} \delta^{(s_2)}}^{-1}$ are non-positive, find the subvector pair $Y_{\delta}^{(s_1^*)}$ and $Y_{\delta}^{(s_2^*)}$ such that ${\Omega_{\beta \delta^{(s_1)}} \Omega_{ \delta^{(s_1)} \delta^{(s_1)}}^{-1}\Omega_{ \delta^{(s_1)}\beta}}$ and ${\Omega_{\beta \delta^{(s_2)}} \Omega_{ \delta^{(s_2)} \delta^{(s_2)}}^{-1}\Omega_{ \delta^{(s_2)}\beta}}$ are maximized at $s_1=s_1^*$ and $s_2=s_2^*$. Then, construct
\begin{gather}
\widehat{CI}_t^*(Y_{\beta},Y_\delta,\Omega) \nonumber \\
\equiv\widehat{CI}_t(Y_{\beta},\Omega_{\beta\delta^{(s^*_1)}} \Omega_{\delta^{(s^*_1)}\delta^{(s^*_1)}}^{-1} Y_\delta^{(s^*_1)},\Omega_{\beta\delta^{(s^*_2)}} \Omega_{\delta^{(s^*_2)}\delta^{(s^*_2)}}^{-1} Y_\delta^{(s^*_2)};z_{1-(\alpha-\gamma)/2},c_\ell(\tilde\Omega^{(s_1^*,s_2^*)}),c_u(\tilde\Omega^{(s^*_1,s^*_2)})). \tag*{$\blacksquare$}
\end{gather}
\end{two sided star}

Similarly to the one-sided CI case above, Figure \ref{el} shows that the expected length of our two-sided CI can be very small for extreme values of $\Omega_{\beta\delta^{(s_1^*)}} \Omega_{\delta^{(s_1^*)}\delta^{(s_1^*)}}^{-1}\Omega_{\delta^{(s_1^*)}\beta}$ and $\Omega_{\beta\delta^{(s_2^*)}} \Omega_{\delta^{(s_2^*)}\delta^{(s_2^*)}}^{-1}\Omega_{\delta^{(s_2^*)}\beta}$. At the same time, for \emph{any} realization of the data, the realized length of our two-sided CI cannot exceed $2\times z_{1-(\alpha-\gamma)/2}=4.009$ for $\alpha=0.05$ and $\gamma=\alpha/10$.  This implies that the length increase of our recommended CI cannot exceed 2.28\% relative to the fixed length $2\times z_{1-\alpha/2}=3.92$ of the standard two-sided CI.\footnote{For $\alpha$ equal to 0.01 and 0.1, the  length of our two-sided CI cannot exceed 5.224 and 3.391 and the  length of the standard two-sided CI is equal to 5.152 and 3.290, respectively. This implies a maximum length increase of 1.41\% and 3.07\%, respectively.} 

Table \ref{Coefficients_2s_table_5} provides the estimated coefficients for a 6$^\text{th}$ order polynomial approximation of $c_u(\tilde \omega)$ for $\alpha = 0.05$ and $\gamma = \alpha/10$ in terms of $\omega_{12}$ and $\omega_{13}$ alone. Tables \ref{Coefficients_2s_table_1} and \ref{Coefficients_2s_table_10} in Appendix \ref{AT} provide the corresponding coefficients for $\alpha = 0.01$ and $\alpha = 0.1$ (and $\gamma = \alpha/10$), respectively. For all three values of $\alpha$, the corresponding $R^2$ is greater than 0.999.\footnote{The regressions are performed on the same grid that underlies Figure \ref{el}, i.e., $\omega \in \bar{\Omega}$.} This is remarkable, as it implies that $c_u(\cdot)$ and $c_\ell(\cdot)$ are nearly invariant to $\omega_{23}$. In fact, the maximum (asymptotic) size distortions that result from relying on the polynomial approximation of $c_u(\cdot)$ and $c_\ell(\cdot)$ in the construction of our two-sided CI are very similar to those that we found for relying on the polynomial approximation in the construction of our one-sided CI.\footnote{For $\alpha$ equal to 0.01, 0.05, and 0.1, a grid search over $\bar{\Omega}$ reveals minimum coverage probabilities of 98.96\%, 94.78\%, and 89.60\%, respectively.} We therefore also expect minimal size distortions from employing the polynomial approximations of $c_u(\cdot)$ and $c_\ell(\cdot)$ in practice, which is again corroborated by the finite sample coverage probabilities found in Section \ref{Sims}.

\begin{table}[h!]											
\begin{center}									
\caption{Coefficients for $6^\text{th}$ order polynomial approximation of $c_u(\omega)$ for $\alpha = 0.05$ and $\gamma = \alpha/10$}									
\label{Coefficients_2s_table_5}									
\begin{tabular}{c|rrrrrrr}								
\hline								
\hline			
 & \multicolumn{1}{c}{1} & \multicolumn{1}{c}{$\omega_{12}$} & \multicolumn{1}{c}{$\omega_{12}^2$} & \multicolumn{1}{c}{$\omega_{12}^3$} & \multicolumn{1}{c}{$\omega_{12}^4$} & \multicolumn{1}{c}{$\omega_{12}^5$} & \multicolumn{1}{c}{$\omega_{12}^6$}  \\ 
\hline
1 & $1.9540$ & $1.3388$ & $-4.5110$ & $11.7294$ & $-18.8756$ & $15.5342$ & $-5.2786$ \\ 
$\omega_{13}$ & $1.1289$ & $-0.8006$ & $1.1262$ & $-1.1742$ & $2.1281$ & $-0.5511$ & $$ \\ 
$\omega_{13}^2$ & $-12.2929$ & $0.0090$ & $0.9084$ & $-3.2329$ & $0.1723$ & $$ & $$ \\ 
$\omega_{13}^3$ & $45.6505$ & $0.5939$ & $0.8153$ & $1.7625$ & $$ & $$ & $$ \\ 
$\omega_{13}^4$ & $-92.3587$ & $-1.0048$ & $-0.9854$ & $$ & $$ & $$ & $$ \\ 
$\omega_{13}^5$ & $89.5045$ & $0.2851$ & $$ & $$ & $$ & $$ & $$ \\ 
$\omega_{13}^6$ & $-33.3683$ & $$ & $$ & $$ & $$ & $$ & $$ \\ 
\hline								
\end{tabular}
\end{center}						
\end{table}

The following proposition also enables one to directly compute an approximation to $c_{\ell}(\tilde\omega)$ from Table \ref{Coefficients_2s_table_5} (and Tables \ref{Coefficients_2s_table_1} and \ref{Coefficients_2s_table_10}) by simply reversing the roles of $\omega_{12}$ and $\omega_{13}$ in the computation of $c_u(\tilde\omega)$.

\begin{proposition} \label{prop:c_ell} 
$c_\ell(\omega_{13},\omega_{12},\omega_{23})=c_u(\omega_{12},\omega_{13},\omega_{23})$.
\end{proposition}



\section{Finite Sample Problem of Restricted Nuisance Parameters} \label{FS}

Consider inference on a scalar parameter of interest $b\in\mathbb{R}$ in a well-behaved model with a vector nuisance parameter $d\in \mathbb{R}_+^k$ for some $k\geq 1$ that is known to have all elements greater than or equal to zero.  For a standard parameter estimator $(\hat b,\hat d^{\prime})^{\prime}$, as the number of observations $n$ in the sample grows, standard assumptions imply\footnote{For simplicity of notation, we suppress the dependence of certain finite-sample quantities on the sample size $n$ until Section 3.2.}
\begin{equation}
\sqrt{n}\left(\begin{array}{c}
\widehat b-b \\
\widehat d-d
\end{array}\right)\overset{d}\longrightarrow \mathcal{N}(0,\Sigma) \quad \text{with} \quad \Sigma=\left(\begin{array}{cc}
\Sigma_{bb} & \Sigma_{bd} \\
\Sigma_{db} & \Sigma_{dd}
\end{array}\right), \label{anorm}
\end{equation}
where $\Sigma$ is a consistently estimable covariance matrix.   Note that this setting accommodates regression models, instrumental variables models, maximum likelihood models and models estimated by the generalized method of moments under standard assumptions when some nuisance parameters are known to be greater or less than a given bound via simple reparameterization of the nuisance parameters.  For example, say that the researcher knows from economic theory that the nuisance parameter $\tilde d$ is less than or equal to $\tilde c$ for some known constant $\tilde c\in\mathbb{R}$.  Then $d=-(\tilde d-\tilde c)$ is the simple reparameterization that fits this setting.

\begin{example}
One of the most common examples that fits this setting is inference on a regression coefficient of interest $b$ in the standard linear regression model for observations $i=1,\ldots,n$
\[y_i=bz_i+x_i^{\prime}d+w_i^{\prime}c+\varepsilon_i,\]
where $y_i$ is the dependent variable, $z_i$ is the scalar regressor of interest, $x_i\in\mathbb{R}^{\mathcal D_x}$ are control variables with \emph{known positive partial effects} $d\geq 0$ on $y_i$, $w_i\in\mathbb{R}^{\mathcal D_w}$ are control variables with unrestricted partial effects $c$ and $\varepsilon_i$ is the error term.  The ordinary least squares estimator $(\hat b,\hat d^{\prime})^{\prime}$ satisfies \eqref{anorm} under standard assumptions on the linear regression model. 
\end{example}

\subsection{Implementation}

For a consistent covariance matrix estimator $\widehat\Sigma$, \eqref{anorm} suggests the following large-sample distributional approximation consistent with \eqref{Limit Exp} and \eqref{block correlation mtx}:
\begin{equation}
\diag\left(\widehat\Sigma\right)^{-1/2}\sqrt{n}\left(\begin{array}{c}
\widehat b \\
\widehat d
\end{array}\right)\overset{a}\sim \mathcal{N}\left(\left(\begin{array}{c}
\beta \\
\delta
\end{array}\right), 
\Omega\right), \label{finite sample dist approx}
\end{equation}
where $\beta=\sqrt{n}b/\sqrt{\Sigma_{bb}}$, $\delta=\diag(\Sigma_{dd})^{-1/2}\sqrt{n}d$ and $\Omega=\diag(\Sigma)^{-1/2}\Sigma\diag(\Sigma)^{-1/2}$.  Note, however, that $\Omega$ is not typically known in practice but can be consistently estimated by $\widehat\Omega=\diag(\widehat\Sigma)^{-1/2}\widehat\Sigma\diag(\widehat\Sigma)^{-1/2}$.  Let $\hat{\delta}^{(s)}=\diag(\widehat\Sigma_{dd}^{(s)})^{-1/2}\sqrt{n}\hat{d}^{(s)}$, $\hat s^*$ denote the subset of the set of indices $\{1,\ldots,k\}$ that maximizes $\widehat\Omega_{bd^{(s)}}\widehat\Omega_{d^{(s)}d^{(s)}}^{-1}\widehat\Omega_{d^{(s)}b}$ amongst all subsets of indices $s\subseteq\{1,\ldots,k\}$ such that the elements of $\widehat\Omega_{bd^{(s)}}\widehat\Omega_{d^{(s)}d^{(s)}}^{-1}$ are non-negative and $(\hat s_1^*,\hat s_2^*)$ denote the subsets of the set of indices $\{1,\ldots,k\}$ that maximize ${\widehat\Omega_{\beta \delta^{(s_1)}} \widehat\Omega_{ \delta^{(s_1)} \delta^{(s_1)}}^{-1}\widehat\Omega_{ \delta^{(s_1)}\beta}}$ and ${\widehat\Omega_{\beta \delta^{(s_2)}} \widehat\Omega_{ \delta^{(s_2)} \delta^{(s_2)}}^{-1}\widehat\Omega_{ \delta^{(s_2)}\beta}}$ amongst all subsets of indices $s_1,s_2\subseteq\{1,\ldots,k\}$ such that the elements of $\widehat\Omega_{\beta \delta^{(s_1)}} \widehat\Omega_{ \delta^{(s_1)} \delta^{(s_1)}}^{-1}$ are non-negative and the elements of $\widehat\Omega_{\beta \delta^{(s_2)}} \widehat\Omega_{ \delta^{(s_2)} \delta^{(s_2)}}^{-1}$ are non-positive.

The distributional approximation in \eqref{finite sample dist approx} and the availability of the consistent estimator $\widehat\Omega$ suggest that we can use
\begin{gather}
{CI}_{u,n}\left(\hat b,\hat d;\widehat\Sigma\right)=\frac{\sqrt{\widehat\Sigma_{bb}}}{\sqrt{n}}\widehat{CI}_u^*\left(\frac{\sqrt{n}\hat{b}}{\sqrt{\widehat \Sigma_{bb}}},\diag(\widehat\Sigma_{dd})^{-1/2}\sqrt{n}\hat{d};\widehat\Omega\right) \notag \\
=\left[\hat b-\frac{\widehat\Sigma_{bb}}{\sqrt{n}}\min\left\{z_{1-\alpha+\gamma},\widehat\Omega_{bd^{(\hat s^*)}}\widehat\Omega_{d^{(\hat s^*)}d^{(\hat s^*)}}^{-1}\hat{\delta}^{(\hat s^*)}+c(\widehat\Omega_{bd^{(\hat s^*)}}\widehat\Omega_{d^{(\hat s^*)}d^{(\hat s^*)}}^{-1}\widehat\Omega_{d^{(\hat s^*)b}})\right\},\infty\right) \label{1-sided finite sample CI}
\end{gather}
and
\begin{align}
&CI_{t,n}\left(\widehat b,\widehat d;\widehat\Sigma\right)=\frac{\sqrt{\widehat\Sigma_{bb}}}{\sqrt{n}}\widehat{CI}_t^*\left(\frac{\sqrt{n}\widehat{b}}{\sqrt{\widehat \Sigma_{bb}}},\diag(\widehat\Sigma_{dd})^{-1/2}\sqrt{n}\widehat{d};\widehat\Omega\right) \notag \\
&=\left[\hat b-\frac{\widehat\Sigma_{bb}}{\sqrt{n}}\min\left\{z_{1-\frac{\alpha-\gamma}{2}},\widehat\Omega_{bd^{(\hat s_1^*)}}\widehat\Omega_{d^{(\hat s_1^*)}d^{(\hat s_1^*)}}^{-1}\hat{\delta}^{(\hat s_1^*)}+c_\ell\left(\widehat{\tilde\Omega}^{(\hat s_1^*,\hat s_2^*)}\right)\right\},\right. \notag\\
&\qquad \left. \hat b+\frac{\widehat\Sigma_{bb}}{\sqrt{n}}\min\left\{z_{1-\frac{\alpha-\gamma}{2}},-\widehat\Omega_{bd^{(\hat s_2^*)}}\widehat\Omega_{d^{(\hat s_2^*)}d^{(\hat s_2^*)}}^{-1}\hat{\delta}^{(\hat s_2^*)}+c_u\left(\widehat{\tilde\Omega}^{(\hat s_1^*,\hat s_2^*)}\right)\right\}\right] \label{2-sided finite sample CI}
\end{align}
as upper one-sided and two-sided CIs for the parameter $b$, where $\widehat{CI}_u^*(\cdot)$ and $\widehat{CI}_t^*(\cdot)$ are defined in Algorithms One-Sided* and Two-Sided*.  The theoretical results of the following section formally confirm that these CIs attain uniformly correct asymptotic coverage under weak conditions.

\subsection{Asymptotic Properties} \label{sec:asymptotics}

We now present theoretical results ensuring the uniformly correct asymptotic coverage of both the one- and two-sided finite-sample CIs defined in \eqref{1-sided finite sample CI} and \eqref{2-sided finite sample CI}, as well as a uniform upper bound on their asymptotic coverage, under a set of widely-applicable sufficient conditions on the parameter space.  In particular, let the parameter $\lambda$ index the true distribution of the observations used to construct the CIs and decompose $\lambda$ as follows:~$\lambda=(b,d,\Sigma,F)$, where $b$ is the scalar parameter of interest, $d$ is the nuisance parameter known to have all elements greater than zero, $\Sigma$ is the asymptotic variance corresponding to the parameter estimator $(\hat b_n,\hat d_n^{\prime})^{\prime}$ used by the researcher and $F$ is a (potentially) infinite-dimensional parameter that, along with $(b,d)$, determines the distribution of the observed data.  We assume that we have a consistent estimator $\widehat\Sigma_n$ of $\Sigma$ at our disposal.

The parameter space $\Lambda$ for $\lambda$ is defined to include parameters $\lambda=(b,d,\Sigma,F)$ such that for some finite $\kappa>0$, the following conditions hold:

(i) $b\in\mathbb{R}$ and $d\in\mathbb{R}_+^k$ for some positive integer $k$;

(ii) $\Sigma\in \Phi$, $\lambda_{\min}(\Sigma)\geq \kappa$ and $\lambda_{\max}(\Sigma)\leq \kappa^{-1}$, where $\Phi$ denotes the set of all positive definite covariance matrices.

In addition, under any sequence of parameters $\{\lambda_{n,\mathfrak{b},\mathfrak{d},\Sigma^*}=(b_{n,\mathfrak{b}},d_{n,\mathfrak{d}},\Sigma_{n,\Sigma^*},F_{n,\mathfrak{b},\mathfrak{d},\Sigma^*}):n\geq 1\}$ in $\Lambda$ such that
\begin{gather}
\sqrt{n}(b_{n,\mathfrak{b}},d_{n,\mathfrak{d}})\rightarrow (\mathfrak{b},\mathfrak{d}), \label{bd drifting seq} \\
\Sigma_{n,\Sigma^*}\rightarrow \Sigma^* \label{covariance drifting seq}
\end{gather}
for some $(\mathfrak{b},\mathfrak{d},\Sigma^*)\in \mathbb{R}_{\infty}\times \mathbb{R}_{+,\infty}^k\times \Phi$, the following remaining conditions hold:

(iii) $\widehat\Sigma_n$ exists and $\lambda_{\min}(\widehat\Sigma_n)>0$ with probability 1 for all $n\geq 1$ and $\widehat\Sigma_n\overset{p}\longrightarrow \Sigma^*$;

(iv) $\sqrt{n}(\hat b_n-b_{n,\mathfrak{b}},\hat d_n^{\prime}-d_{n,\mathfrak{d}}^{\prime})^{\prime}\overset{d}\longrightarrow \mathcal{N}(0,\Sigma^*)$;

(v) for any sequence $\{\lambda_{n,\mathfrak{b},\mathfrak{d},\Sigma^*}\}$ in $\Lambda$ and any subsequence $\{s_n:n\geq 1\}$ of $\{n:n\geq 1\}$ for which \eqref{bd drifting seq}--\eqref{covariance drifting seq} hold along the subsequence, conditions (iii)--(iv) also hold along the subsequence.

In conjunction with a particular model, parameter estimator $(\hat b_n,\hat d_n^{\prime})^{\prime}$ and covariance matrix estimator $\widehat\Sigma_n$, this definition of the parameter space $\Lambda$ effectively serves as a set of high-level assumptions on the underlying DGP.  More specifically, (i)--(ii) are standard parameter space assumptions while (iii)--(v) can typically be verified under standard dependence and moment conditions on the underlying data via laws of large numbers and central limit theorems.  We refer the interested reader to Appendix \ref{sec:reg_space} for details in the context of the standard linear regression model.

With the relevant parameter space defined, we may now state the main theoretical result of this paper that establishes lower and upper bounds on the uniform asymptotic coverage probability of the CIs we propose.

\begin{theorem} \label{thm:uniform coverage}
For $\alpha\in (0,1/2)$ and $\gamma \in (0,\alpha)$,
\[\liminf_{n\rightarrow\infty}\text{ }\inf_{\lambda\in\Lambda}P_\lambda\left(b\in CI_{\cdot,n}(\hat b_n,\hat d_n;\widehat\Sigma_n)\right)\geq 1-\alpha\]
and
\[\limsup_{n\rightarrow\infty} \text{ }\sup_{\lambda\in\Lambda}P_\lambda\left(b\in CI_{\cdot,n}(\hat b_n,\hat d_n;\widehat\Sigma_n)\right)\leq 1-\alpha+\gamma,\]
where $CI_{\cdot,n}(\cdot)$ is equal to either $CI_{u,n}(\cdot)$ or $CI_{t,n}(\cdot)$. 
\end{theorem}

These results show not only that our proposed CIs have correct asymptotic coverage in a strong sense but also that by choosing $\gamma$ to be ``small'' reduces how conservative the CIs can be.  However, there is a tradeoff in the choice of $\gamma$:~although a smaller $\gamma$ leads to CIs that are closer to being similar across the parameter space, it also allows for less length gains when the elements of $d$ are close or equal to zero.

\section{Empirical Application of Sign-Restricted Regression} \label{EA}
For our proposed CIs to to be able to improve upon the length of standard CIs in the standard linear regression context, the researcher must know the sign of at least one of the control variables' coefficients and the estimator of the coefficient of interest must be (asymptotically) correlated with the estimator of the sign-restricted control variables'  coefficients. Both conditions are often satisfied in the context of treatment effect regressions for cross-cutting/factorial designs in field experiments. Take, for example, the 2$\times$2 factorial design:
\begin{equation} \label{interaction}
	Y = \alpha_0 + \alpha_1 T_1 + \alpha_2 T_2 + \alpha_3 T_1 \times T_2 + u,
\end{equation}
where $E[u|T_1,T_2] = 0$ and $T_1$ and $T_2$ denote two independent, randomly assigned treatments with $T_i \in \{0,1\}$ for $i \in \{1,2\}$. Here, $\alpha_1$ and $\alpha_2$ are the treatment effects of $T_1$ and $T_2$ ``relative to a business-as-usual counterfactual'' \citep*{Muralidharan19} and $\alpha_3$ is the ``interaction effect'', i.e., the treatment effect of jointly providing both treatments minus the sum of the treatment effects of $T_1$ and $T_2$.\footnote{Using the potential outcomes notation, where $Y_{t_1,t_2}$ is the potential outcome of $Y$ when $T_1 = t_1$ and $T_2 = t_2$, the three treatment effects can be written as $\alpha_1 = E[Y_{1,0} - Y_{0,0}]$, $\alpha_2 = E[Y_{0,1} - Y_{0,0}]$, and $\alpha_3 =  E[Y_{1,1} - Y_{0,0}] - (E[Y_{1,0} - Y_{0,0}] + E[Y_{0,1} - Y_{0,0}])$.}  If $Y$ is a ``positive'' outcome, it is often reasonable to assume that $\alpha_1 \geq 0$ and $\alpha_2 \geq 0$.  For example, a research ethics committee is unlikely to clear an experimental design if this is not the case. Furthermore, the OLS estimators of the three treatment effects are likely to be highly correlated in this setting.  For example, if each treatment is assigned with probability 1/2 and the error term $u$ is conditionally homoskedastic, then the asymptotic correlation matrix of $\sqrt{n}(\hat{\alpha}_1,\hat{\alpha}_2,\hat{\alpha}_3)'$ is given by
\[
	\left[ \begin{array}{ccc} 1 & 1/2 & -1/\sqrt{2} \\ 1/2 & 1 & -1/\sqrt{2}  \\ -1/\sqrt{2}  & -1/\sqrt{2}  & 1  \end{array} \right].
\]

Any of the three treatment effects may be of interest and, under the assumption that $\alpha_1 \geq 0$ and $\alpha_2 \geq 0$, it is reasonable to be interested in upper one-sided CIs for $\alpha_1$ and $\alpha_2$ and a two-sided CI for $\alpha_3$. The above correlation structure implies that our upper one-sided CIs for ${\alpha}_1$ and ${\alpha}_2$ have the potential to improve upon the length of standard upper one-sided CIs. Similarly, our two-sided CI for $\alpha_3$ has the potential to improve upon the length of the standard two-sided CI through a smaller upper bound.

Sometimes researchers are interested in the following alternative specification of the above regression:
\begin{equation} \label{both}
	Y = \alpha_0 + \alpha_1 (T_1 - T_1 \times T_2 ) + \alpha_2 (T_2 - T_1 \times T_2 ) + \alpha^*_3 T_1 \times T_2 + u,
\end{equation}
where $\alpha^*_3 = \alpha_3 - \alpha_1 - \alpha_2$ is the effect of ``both'' treatments provided jointly, relative to a business-as-usual counterfactual.\footnote{That is $\alpha_3^* = E[Y_{1,1} - Y_{0,0}]$.} This regression again results in high correlation between OLS estimators:~under the same conditions as in the example above, the asymptotic correlation matrix of $\sqrt{n}(\hat{\alpha}_1,\hat{\alpha}_2,\hat{\alpha}^*_3)'$ is given by
\[
	\left[ \begin{array}{ccc} 1 & 1/2 & 1/{2} \\ 1/2 & 1 & 1/{2}  \\ 1/{2}  & 1/{2}  & 1  \end{array} \right].
\]
In this case, our two-sided CI for $\alpha_3^*$ has the potential to improve upon the length of the standard two-sided CI through a larger lower bound.

To illustrate the usefulness of our proposed CIs, we apply them in the context of a field experiment where a 2$\times$2 factorial design was used. In particular, we revisit \cite*{blattman2017} (BJS) who recruited 999 poor young men in Liberia who exhibited ``high rates of violence, crime, and other antisocial behaviors'' to participate in an experiment. The two treatments are ``therapy'', an eight-week program of group cognitive behavior therapy, and ``cash'', a \$200 grant corresponding to roughly three months' wages. In simple terms, the main research question is whether ``therapy'' and ``cash'' can help reduce violent, criminal, and other antisocial behaviors. The hypothesized channels are improved noncognitive skills such as self-control (``therapy'') and an increase in legal work (``cash''). BJS conducted two follow-up surveys, the first 2--5 weeks and the second 12--13 months after the intervention to elicit ``short-term'' and ``long-term'' impacts, respectively. 

\begin{table}[h!]											
\begin{center}									
\caption{Empirical results}			
\label{ES}									
\begin{tabular}{c|rcccccr}								
\hline								
\hline			
&  \multicolumn{1}{c}{$ \hat{b} $}  & SE & SSCI & (E)L & SCI & (E)L  &  \multicolumn{1}{c}{Ratio}  \\ 
\hline
T    & $0.0829$ & $0.0929$ & $[-0.0149,\infty)$ & $0.0978$ & $[-0.0700,\infty)$ & $0.1529$ & $0.6395$  \\ 
C    & $-0.1316$ & $0.0969$ & $[-0.2959,\infty)$ & $0.1643$ & $[-0.2910,\infty)$ & $0.1594$ & $1.0307$   \\ 
B    & $0.2468$ & $0.0883$ & $[0.0988,0.4238]$ & $0.3250$ & $[0.0737,0.4198]$ & $0.3462$ & $0.9390$ \\ 
I    & $0.2955$ & $0.1255$ & $[0.0439,0.4101]$ & $0.3662$ & $[0.0495,0.5415]$ & $0.4920$ & $0.7443$  \\ 
\hline								
\end{tabular}
\end{center}						
\end{table}

Table \ref{ES} reproduces the results concerning the treatments' long-term impact on a summary index of antisocial behaviors (times minus one) (cf.\ the first row of Panel B of Table 2 in BJS). The table includes one of the main findings of BJS:~while the two treatments do not have statistically significant long-term effects in isolation, they do have a \emph{joint} positive long-term effect on the index of antisocial behaviors. Column 1 ($\hat{\beta}$) shows the OLS point estimates for ``therapy'' (T), ``cash'' (C), ``both'' (B), and ``interaction'' (I) as defined above and column 2 (SE) reports the corresponding (heteroskedasticity-robust) standard errors. Note that BJS only consider the specification given in equation \eqref{both}, i.e., they only estimate the effect of ``both'' treatments and not the ``interaction'' effect.\footnote{\label{controls}In fact, BJS consider the specification given in equation \eqref{both} augmented by a set of additional controls. For the purpose of this analysis, we take the signs of these additional controls as unknown. See BJS for more information on the additional controls.} Column 3 (SSCI\textemdash Simple and Short Confidence Interval) shows our proposed CIs for $\alpha = 0.05$, which are upper one-sided for T and C and two-sided for B and I, when assuming that the treatment effects of ``therapy'' and ``cash'' are \textit{a priori} known to be nonnegative. They are constructed using Algorithms One-Sided$^*$ and Two-Sided$^*$ in combination with the response surface approximations.\footnote{The estimated (asymptotic) correlation matrices for the estimator of the effects of i) T, C, and B and ii) T, C, and I are given by
\[
\left[ \begin{array}{rrr} 1.0000 &   0.5238  &  0.6104 \\
    0.5238   & 1.0000 &   0.5543\\
    0.6104   &  0.5543  &   1.0000 \end{array} \right] \text{ and }
  \left[  \begin{array}{rrr} 1.0000 &   0.5238 &  -0.7154 \\
    0.5238 &   1.0000  & -0.7699\\
   -0.7154&   -0.7699 &   1.0000 \end{array} \right], 
\]
respectively. We augmented the corresponding regressions by the same set of controls as BJS, cf.\ footnote \ref{controls}.} Column 5 (SCI\textemdash Standard Confidence Interval) shows the corresponding standard CIs. Columns 4 and 6 (both (E)L) give the (``excess'') lengths of SSCI and SCI, where the ``excess'' length of one-sided CIs here is computed as the difference between $\hat b$ and the CI's lower bound.\footnote{We write ``excess'' in quotes to emphasize the fact that this is not equal to the true excess length that cannot be computed here in the absence of knowledge of the true value of the regression coefficients.}  Column 7 (Ratio) computes the ratio of the (``excess'') length of SSCI relative to SCI.  We find that, while our proposed CI is marginally longer than the standard CI for C\textemdash its ``excess'' length reaching the bound on expected excess length increase of $\sim 3\%$, it is much shorter for T, B, and I.

\section{Calibrated Simulations for Sign-Restricted Regression} \label{Sims}

To illustrate the finite-sample properties of our proposed CIs, we perform a Monte Carlo study calibrated to the BJS factorial design regression of the previous section. In particular, we create 10,000 bootstrap samples by drawing with replacement from the sample of $n = 947$ men underlying the regression results in Table \ref{ES}. In each bootstrap sample, we estimate the regressions \eqref{interaction} and \eqref{both}. Since the expected value of the treatment effect of ``cash'' under the empirical distribution is equal to the point estimate in the original sample, -0.1316, it is outside of the sign-restricted  parameter space $\alpha_2\geq 0$.  We therefore recenter the estimates of the treatment effect of ``cash'' over the bootstrap samples to have mean zero (by adding 0.1316). For each bootstrap sample, we construct our proposed CIs, using Algorithms One-Sided$^*$ and Two-Sided$^*$ in combination with the response surface approximations, standard CIs, the (excess) length of each CI and whether they cover the true parameter value, i.e., the corresponding (re-centered) point estimate in the original sample.  All CIs are constructed using standard heteroskedasticity-robust variance-covariance matrix estimators computed within each bootstrap sample.  Since the empirical distribution from which the bootstrap samples are drawn is not normally distributed, this simulation exercise captures the effect on CI coverage of departures from the large sample normal means problem of Section \ref{NMLSP}.

\begin{table}[h!]											
\begin{center}									
\caption{Monte Carlo results}			
\label{SS}									
\begin{tabular}{cc|rrrrrr}								
\hline								
\hline			
& & \multicolumn{1}{c}{T} & \multicolumn{1}{c}{C} & \multicolumn{1}{c}{B}  & \multicolumn{1}{c}{I}  & \multicolumn{1}{c}{B0} & \multicolumn{1}{c}{I0}  \\ 
\hline
\multirow{2}{*}{CP}  &  SSCI  &    94.26  &  94.48  &  95.34  &  94.80 & 95.07  &  93.97 \\
				& SCI &    94.30 &   93.96 &   94.78 &   94.49 & 94.78   & 94.49 \\
				\hline
\multirow{3}{*}{E(E)L}  & SSCI & 0.14  &  0.16  &  0.35 &  0.45 & 0.33  &  0.41 \\
   				   &  SCI & 0.15   & 0.16 &   0.36  &  0.50 & 0.36   & 0.50\\
   		  	          & Ratio  & 0.9303  &  1.0071   & 0.9743 &   0.8920 & 0.9230  &  0.8241 \\
\hline								
\end{tabular}
\end{center}						
\end{table}

Table \ref{SS} reports the coverage probability (CP) computed across bootstrap realizations of our proposed CIs and of the standard CIs for all four treatment effects, T, C, B, and I. Table \ref{SS} also reports the expected (excess) length (E(E)L) of these CIs across the bootstrap realizations. In addition to the above DGP, we also consider a modification where the true value of the treatment effect of ``therapy'' is set equal to zero (by subtracting the point estimate in the original sample, 0.0829, from the corresponding estimates in the bootstrap iterations). The corresponding results for the effect of ``both'' treatments and the ``interaction'' effect are given in the last two columns, B0 and I0. 

We observe that our proposed CIs have good finite sample coverage, comparable to that of the standard CIs, with little coverage distortion despite the non-normally distributed data.  In terms of expected (excess) length, most of our proposed CIs offer sizeable improvements over standard CIs, with expected (excess) length improvements of up to nearly 18\% for this particular data calibration.

\newpage

\appendix

\section{Technical Appendix} \label{sec:proofs}

\begin{proof}[Proof of Proposition \ref{prop:c existence}] Consider the function $f:[0,1)\times [0,z_{1-\gamma}]$ such that for $(Z_1,\tilde Z_2)$ defined in \eqref{c-function def},
\[f(\omega,c)=P(Z_1>\min\{z_{1-\alpha+\gamma},\tilde Z_2+c\})-\alpha.\]
For $\omega\in (0,1)$ and $c\in [0,z_{1-\gamma}]$,
\begin{gather*}
f(\omega,c)=\int_{-\infty}^\infty P(Z_1>\min\{z_{1-\alpha+\gamma},\tilde Z_2+c\}|\tilde Z_2=\tilde z_2)\frac{1}{\sqrt{\omega}}\phi(\tilde z_2/\sqrt{\omega})d\tilde z_2-\alpha \\
=\int_{-\infty}^\infty \Phi\left(\frac{\tilde z_2-\min\{z_{1-\alpha+\gamma},\tilde z_2+c\}}{\sqrt{1-\omega^2}}\right)\frac{1}{\sqrt{\omega}}\phi(\tilde z_2/\sqrt{\omega})d\tilde z_2-\alpha \\
=\int_{-\infty}^{z_{1-\alpha+\gamma}-c} \Phi\left(- \frac{c}{\sqrt{1-\omega^2}}\right)\frac{1}{\sqrt{\omega}}\phi(\tilde z_2/\sqrt{\omega})d\tilde z_2 +\int_{z_{1-\alpha+\gamma}-c}^\infty \Phi\left(\frac{\tilde z_2-z_{1-\alpha+\gamma}}{\sqrt{1-\omega^2}}\right)\frac{1}{\sqrt{\omega}}\phi(\tilde z_2/\sqrt{\omega})d\tilde z_2 -\alpha \\
=\Phi\left( -\frac{c}{\sqrt{1-\omega^2}}\right) \Phi \left( \frac{z_{1-\alpha+\gamma}-c}{\sqrt{\omega}}\right) +\int_{z_{1-\alpha+\gamma}-c}^\infty \Phi\left(\frac{\tilde z_2-z_{1-\alpha+\gamma}}{\sqrt{1-\omega^2}}\right)\frac{1}{\sqrt{\omega}}\phi(\tilde z_2/\sqrt{\omega})d\tilde z_2 -\alpha.
\end{gather*}
Clearly, $f(\omega,c)$ is continuously differentiable for all $\omega\in (0,1)$ and $c\in [0,z_{1-\gamma}]$.  In addition,
\[
	\frac{\partial f(\omega,c)}{\partial c} = -\frac{1}{\sqrt{1-\omega^2}} \phi\left(- \frac{c}{\sqrt{1-\omega^2}}\right) \Phi \left( \frac{z_{1-\alpha+\gamma}-c}{\sqrt{\omega}}\right) < 0
\]
for all $\omega\in (0,1)$ and $c\in [0,z_{1-\gamma}]$ since $\gamma\in(0,\alpha)$.

Finally, note that for any $\omega\in(0,1)$, there exists $c\in[0,z_{1-\gamma}]$  such that $f(\omega,c)=0$ since $f(\omega,\cdot)$ is continuously strictly decreasing,
\[f(\omega,0)=P(Z_1>\min\{z_{1-\alpha+\gamma},\tilde Z_2\})-\alpha >P(Z_1-\tilde Z_2>0)-\alpha=1/2-\alpha>0\] 
and
\[f(\omega,z_{1-\gamma})=P(Z_1>\min\{z_{1-\alpha+\gamma},\tilde Z_2+z_{1-\gamma}\})-\alpha\leq P(Z_1>z_{1-\alpha+\gamma}) - \alpha =-\gamma<0.\]
In conjunction with the fact that $c(0)=z_{1-\alpha}=\lim_{\omega\rightarrow 0}c(\omega)$, the statement of the proposition then follows from the implicit function theorem. 
\end{proof}

The next lemmata are used to prove Proposition \ref{prop:one-sided optimality}.

\begin{lemma} \label{inverse 3x3 matrix}
For conformable matrices $E$, $F$, $G$, $H$, $J$, and $K$, let 
\[
	X = \left[ \begin{array}{ccc} E & F & G \\ F' & H & J \\ G' & J' & K \end{array} \right].
\]
Then, assuming the relevant inverse matrices exist, we have
\[
	X^{-1} =  \left[ \begin{array}{ccc} E^{-1} + E^{-1}[FA^{-1}F' + U S^{-1}U']E^{-1} & -E^{-1}[F-US^{-1}B']A^{-1} & -E^{-1}US^{-1} \\ -A^{-1}[F'-BS^{-1}U']E^{-1} & A^{-1} + A^{-1} B S^{-1} B' A^{-1}& -A^{-1} B S^{-1} \\ 
	-S^{-1}U'E^{-1} &  -S^{-1} B' A^{-1}  & S^{-1} \end{array} \right],
\]
where $A = H - F'E^{-1}F$, $B = J - F'E^{-1}G$, $D = K-G'E^{-1}G$, $S = D - B'A^{-1}B$, and $U = G - FA^{-1}B$.
\end{lemma}

\begin{proof} The proof follows from repeated application of the formula for blockwise inversion of a matrix.
\end{proof}

\begin{lemma} \label{Woodbury}
For conformable matrices $Y$ and $Z$, assuming the relevant inverse matrices exist, we have
\[
	(Y+Z)^{-1} = Y^{-1} - Y^{-1}Z(Y+Z)^{-1}.
\]
\end{lemma}

\begin{proof} The proof follows directly from the Woodbury identity. 
\end{proof}

\begin{lemma} \label{submatrix correspondence} Let $\delta^{(s)}$ and $\delta^{(-s)}$ be two arbitrary subvectors of $\delta$ such that $\delta = (\delta^{(s)},\delta^{(-s)})$, where the order of the elements is without loss of generality. Furthermore, let
\[
	\left[ \begin{array}{ccc} 1 & \Omega_{\beta \delta^{(s)}} & \Omega_{\beta \delta^{(-s)}} \\ \Omega_{\delta^{(s)} \beta} & \Omega_{\delta^{(s)} \delta^{(s)}} & \Omega_{\delta^{(s)} \delta^{(-s)}} \\ \Omega_{\delta^{(-s)} \beta} & \Omega_{\delta^{(-s)} \delta^{(s)}} & \Omega_{\delta^{(-s)} \delta^{(-s)}} \end{array} \right] \text{ and } 
	\left[ \begin{array}{ccc} \Omega^{\beta \beta} & \Omega^{\beta \delta^{(s)}} & \Omega^{\beta \delta^{(-s)}} \\ \Omega^{\delta^{(s)} \beta} & \Omega^{\delta^{(s)} \delta^{(s)}} & \Omega^{\delta^{(s)} \delta^{(-s)}} \\ \Omega^{\delta^{(-s)} \beta} & \Omega^{\delta^{(-s)} \delta^{(s)}} & \Omega^{\delta^{(-s)} \delta^{(-s)}} \end{array} \right]
\]
be conformable partitions of $\Omega$ and $\Omega^{-1}$, respectively. Then, we have
\begin{enumerate}
\item[(i)] $\Omega^{\beta \delta^{(s)}} - \Omega^{\beta \delta^{(-s)}} (\Omega^{\delta^{(-s)} \delta^{(-s)}})^{-1} \Omega^{\delta^{(-s)} \delta^{(s)}} = - (\Omega^{\beta \beta} - \Omega^{\beta \delta^{(-s)}} (\Omega^{\delta^{(-s)} \delta^{(-s)}})^{-1} \Omega^{\delta^{(-s)} \beta}) \Omega_{\beta \delta^{(s)}} \Omega_{\delta^{(s)} \delta^{(s)}}^{-1}$ and
\item[(ii)] $1 = (\Omega^{\beta \beta} - \Omega^{\beta \delta^{(-s)}} (\Omega^{\delta^{(-s)} \delta^{(-s)}})^{-1} \Omega^{\delta^{(-s)} \beta}) (1 - \Omega_{\beta \delta^{(s)}} \Omega_{\delta^{(s)} \delta^{(s)}}^{-1} \Omega_{\delta^{(s)} \beta})$.
\end{enumerate}
\end{lemma}
\begin{proof}

 \textit{(i)} Using Lemma \ref{inverse 3x3 matrix}, we show the equivalent result that
\[
	X^{12} - X^{13}(X^{33})^{-1}X^{32} + (X^{11} - X^{13}(X^{33})^{-1}X^{31})X_{12}X_{22}^{-1} = 0,
\]
where we use the same notational convention concerning sub- and superscripts as for $\Omega$ and $\Omega^{-1}$. We have
\begin{align*}
& X^{12} - X^{13}(X^{33})^{-1}X^{32} + (X^{11} - X^{13}(X^{33})^{-1}X^{31})X_{12}X_{22}^{-1} \\
=&  -E^{-1}[F-US^{-1}B']A^{-1} -E^{-1}US^{-1}B'A^{-1}\\
 +& \left[ E^{-1} + E^{-1}[FA^{-1}F' + U S^{-1}U']E^{-1} - E^{-1}US^{-1}U'E^{-1}\right] F H^{-1}\\
 =&  E^{-1}FH^{-1}-E^{-1}FA^{-1} + E^{-1}FA^{-1}F'E^{-1}FH^{-1}\\
  =& E^{-1}F[H^{-1} - A^{-1} + A^{-1}F'E^{-1}FH^{-1}] = 0,
\end{align*}
where the last equality follows from the fact that $H^{-1} - A^{-1} + A^{-1}F'E^{-1}FH^{-1} = 0$ which, in turn, follows from Lemma \ref{Woodbury} using $Y = A = H - F'E^{-1}F$ and $Z = F'E^{-1}F$.\\
\textit{(ii)} Using Lemma \ref{inverse 3x3 matrix}, we show the equivalent result that
\[
	(X^{11} - X^{13}(X^{33})^{-1}X^{31}) (X_{11} - X_{12} X_{22}^{-1} X_{21}) = I,
\]
where we again use the same notational convention as for $\Omega$ and $\Omega^{-1}$. We have
\begin{align*}
&(X^{11} - X^{13}(X^{33})^{-1}X^{31}) (X_{11} - X_{12} X_{22}^{-1} X_{21})\\
	=&[E^{-1} + E^{-1}FA^{-1}F'E^{-1}][E - FH^{-1}F'] \\
	=& I - E^{-1}FH^{-1}F'+ E^{-1}FA^{-1}F' - E^{-1} FA^{-1}F'E^{-1}FH^{-1} F'\\
	=& I -  E^{-1}F[H^{-1} - A^{-1} + A^{-1}F'E^{-1}FH^{-1}]F' = I,
\end{align*}
where we used again that $H^{-1} - A^{-1} + A^{-1}F'E^{-1}FH^{-1} = 0$.
\end{proof}

\begin{proof}[Proof of Proposition \ref{prop:one-sided optimality}]
Note that the problem of forming a CI for $\beta$ when we observe $Y\sim \mathcal{N}(\theta,\Omega)$ with $\theta=(\beta,\delta^\prime)^\prime$, $\delta\geq 0$ and known $\Omega$ is equivalent to forming a CI for $\beta$ in the setting of \cite{AK18} (AK henceforth):
\[\widetilde Y=X\theta+\varepsilon,\quad \varepsilon\sim\mathcal{N}(0,I_{k+1})\]
where $\widetilde Y=\Omega^{-1/2}Y$, $X=\Omega^{-1/2}$ and $k$ is the dimension of $\delta$. In what follows, we appeal to Theorem 3.1 of AK. \\
(i) In order to form the CI in this theorem, we must first form the affine estimator $\hat L_{\tilde\delta,\mathcal{F},\mathcal{G}}$ in (23) of AK, for $\mathcal{F}=\mathcal{G}=\{(\beta,\delta^\prime)^\prime\in\mathbb{R}^{k+1}:\delta\geq 0\}$.  
The modulus of continuity defined on p.~667 of AK specialized to our setting is
\begin{gather*}
\omega(\tilde\delta;\mathcal{F},\mathcal{G})=\sup_{\gamma_1,\theta}\{\gamma_1-\beta\} \\
\text{s.t. } (\gamma_1-\beta,\gamma_{1}'-\delta^{\prime})\Omega^{-1}(\gamma_1-\beta,\gamma_{-1}'-\delta^{\prime})^\prime\leq \tilde\delta^2, \quad \gamma_{-1}\geq 0, \quad \delta\geq 0,
\end{gather*}
where we use the partition $\gamma\equiv(\gamma_1,\gamma_{-1}^\prime)^\prime$. Let
\[\Omega^{-1}=\left(\begin{array}{cc}
\Omega^{\beta\beta} & \Omega^{\beta\delta} \\
\Omega^{\delta\beta} & \Omega^{\delta\delta}
\end{array}\right)\]
so that the constraints for the modulus problem can be written as
\[ (\gamma_1-\beta)^2\Omega^{\beta\beta}+2(\gamma_1-\beta)\Omega^{\beta\delta}(\gamma_{-1}-\delta)+(\gamma_{-1}-\delta)^\prime\Omega^{\delta\delta}(\gamma_{-1}-\delta)\leq \tilde\delta^2, \quad \gamma_{-1}\geq 0, \quad \delta\geq 0.
\]
 For any $\theta$ and $\gamma$ that solve this optimization problem in the absence of the final two constraints, we may simply add $(c,\dots,c)'$ for a large constant $c$ to both $\theta$ and $\gamma$ and obtain the same value without imposing the final two constraints on $\gamma_{-1}$ and $\delta$.  Thus, these final two constraints do not affect the optimal procedure and we may instead focus on the modulus problem that drops them with the understanding that the solutions in $\gamma_{-1}$ and $\delta$ must be large and positive.

After dropping these constraints, the first order condition wrt $\delta$ in the modulus problem is
\[-2\lambda[(\gamma_1-\beta)\Omega^{\beta\delta}+(\gamma_{-1}-\delta)^\prime\Omega^{\delta\delta}]=0,\]
where $\lambda>0$ is the KKT multiplier associated with the remaining constraint.  Using the formula for blockwise inversion of a matrix, the optimal solution to the modulus problem must therefore satisfy
\[(\gamma_{-1}-\delta)^\prime=-(\gamma_1-\beta)\Omega^{\beta\delta}(\Omega_{\delta\delta}-\Omega_{\delta\beta}\Omega_{\beta\delta}).\]
The modulus problem thus simplifies to
\begin{gather*}
\omega(\tilde\delta;\mathcal{F},\mathcal{G})=\sup_{\gamma_1,\beta}\{\gamma_1-\beta\} \\
\text{s.t. } (\gamma_1-\beta)^2\Omega^{\beta\beta}-(\gamma_1-\beta)^2\Omega^{\beta\delta}(\Omega_{\delta\delta}-\Omega_{\delta\beta}\Omega_{\beta\delta})\Omega^{\delta\beta}\leq \tilde\delta^2,
\end{gather*}
where the constraint further simplifies to
\[
\text{s.t. } (\gamma_1-\beta)^2\leq \tilde\delta^2,
\]
by the formula for blockwise inversion of a matrix. Thus, we have
\[\omega(\tilde\delta;\mathcal{F},\mathcal{G})=\tilde\delta\]
with a solution given by
\[\gamma_{\tilde\delta,\mathcal{F},\mathcal{G}}^*=\left(\begin{array}{c}
\tilde\delta/2 \\
\delta^*+\tilde\delta\Omega_{\delta\beta}
\end{array}\right), \quad \theta_{\tilde\delta,\mathcal{F},\mathcal{G}}^*=\left(\begin{array}{c}
-\tilde\delta/2 \\
\delta^*
\end{array}\right)\]
and midpoint
\[\theta_{M,\tilde\delta,\mathcal{F},\mathcal{G}}^*=(\theta_{\tilde\delta,\mathcal{F},\mathcal{G}}^*+\gamma_{\tilde\delta,\mathcal{F},\mathcal{G}}^*)/2=\left(\begin{array}{c}
0 \\
\delta^*+\tilde\delta\Omega_{\delta\beta}/2
\end{array}\right),\]
for some large and positive $\delta^*$, where we use the fact that $\Omega^{\beta\delta}(\Omega_{\delta\delta}-\Omega_{\delta\beta}\Omega_{\beta\delta})=-\Omega_{\beta\delta}$ by the formula for blockwise inversion of a matrix.  Formula  (23) of AK thus yields
\begin{align*}
\hat L_{\tilde\delta,\mathcal{F},\mathcal{G}}&=\tilde\delta^{-1}(\gamma_{\tilde\delta,\mathcal{F},\mathcal{G}}^*-\theta_{\tilde\delta,\mathcal{F},\mathcal{G}}^*)^\prime\Omega^{-1}(Y-\theta_{M,\tilde\delta,\mathcal{F},\mathcal{G}}^*) \\
&=(1,\Omega_{\beta\delta})\left(\begin{array}{cc}
\Omega^{\beta\beta} & \Omega^{\beta\delta} \\
\Omega^{\delta\beta} & \Omega^{\delta\delta}
\end{array}\right)\left(\begin{array}{c}
Y_\beta \\
Y_\delta-\delta^*-\tilde\delta\Omega_{\delta\beta}/2
\end{array}\right) \\
&=(\Omega^{\beta\beta}+\Omega_{\beta\delta}\Omega^{\delta\beta})Y_\beta+(\Omega^{\beta\delta}+\Omega_{\beta\delta}\Omega^{\delta\delta})Y_\delta-(\Omega^{\beta\delta}+\Omega_{\beta\delta}\Omega^{\delta\delta})(\delta^*+\tilde\delta/2)\Omega_{\delta\beta} \\
&=Y_\beta,
\end{align*}
where the final equality follows from the facts $\Omega^{\beta\beta}+\Omega_{\beta\delta}\Omega^{\delta\beta}=1$ and $\Omega^{\beta\delta}+\Omega_{\beta\delta}\Omega^{\delta\delta}=0$ by the formula for blockwise inversion of a matrix. Theorem 3.1 of AK then provides that among all upper one-sided CIs with coverage of at least $(1-\alpha)$ for all $\delta\geq 0$
\[[Y_\beta-z_{1-\alpha},\infty)\]
minimizes all maximum excess length quantiles over the $\delta\geq 0$ parameter space at quantile levels greater than $\alpha$.

(ii) We first form the affine estimator $\hat L_{\tilde \delta,\mathcal{F},\mathcal{G}}$ in (23) of AK, for $\mathcal{F}=\{(\beta,\delta^\prime)^\prime\in\mathbb{R}^{k+1}:\delta\geq 0\}$ and $\mathcal{G}=\{(\beta,\delta^\prime)^\prime\in\mathbb{R}^{k+1}:\delta= 0\}$. The modulus of continuity in this setting is
\begin{gather}
\omega(\tilde\delta;\mathcal{F},\mathcal{G})=\sup_{\gamma_1,\theta}\{\gamma_1-\beta\} \nonumber \\
\text{s.t. } (\gamma_1-\beta)^2\Omega^{\beta\beta}-2(\gamma_1-\beta)\Omega^{\beta\delta}\delta+\delta^\prime\Omega^{\delta\delta}\delta\leq \tilde{\delta}^2,\quad \delta\geq 0. \label{constraint}
\end{gather}
Here, the first order condition wrt $\delta_i$ is
\[-2\lambda[(\gamma_1-\beta)\Omega_i^{\beta\delta}-\Omega_{i,\cdot}^{\delta\delta}\delta] - \mu_i=0,\]
where $\lambda>0$ is the KKT multiplier associated with the first constraint, $\mu_i\geq 0$ is the KKT multiplier associated with the constraint $\delta_i\geq 0$ that satisfies the complementary slackness condition $\mu_i\delta_i=0$, and $\Omega_{i,\cdot}^{\delta\delta}$ denotes the $i^{th}$ row of $\Omega^{\delta\delta}$. The solution to the modulus problem must therefore satisfy 
\begin{equation}
\Omega_{i,\cdot}^{\delta\delta}\delta=(\gamma_1-\beta)\Omega_{i}^{\delta\beta} + \tilde\mu_i \label{delta constraint}
\end{equation}
for $i=1,\ldots,k$ and some constants $\tilde\mu_i\geq 0$ such that $\tilde\mu_i\delta_i=0$.  The solution to the modulus problem thus maximizes $(\gamma_1-\beta)$ amongst all $\gamma_1,\theta$ values that satisfy \eqref{constraint} and \eqref{delta constraint} for $i=1,\ldots,k$ and some constants $\tilde\mu_i\geq 0$ such that $\tilde\mu_i\delta_i=0$.  

Next, we consider the candidate solutions to the modulus problem. 
 Let $\delta^{(s)}$ ($\delta^{(-s)}$) denote a (possibly empty) subvector of $\delta$ that satisfies $\delta^{(s)} = 0$ ($\delta^{(-s)} \geq 0$) with $\tilde \mu^{(s)} > 0$ ($\tilde \mu^{(-s)} = 0$). Then, using the notational conventions introduced in Lemma \ref{submatrix correspondence} in what follows, the set of equations given in \eqref{delta constraint} implies
\begin{equation} \label{delta minus s}
	\delta^{(-s)} = (\gamma_1 - \beta) ( \Omega^{\delta^{(-s)}\delta^{(-s)}} )^{-1} \Omega^{\delta^{(-s)}\beta},
\end{equation}
where we use the convention that $( \Omega^{\delta^{(-s)}\delta^{(-s)}} )^{-1} \Omega^{\delta^{(-s)}\beta}=0$ for $\delta^{(s)}=\delta$, and the modulus problem simplifies to 
\begin{gather*}
\omega(\tilde\delta;\mathcal{F},\mathcal{G})=\sup_{\gamma_1,\beta}\{\gamma_1-\beta\} \\
\text{s.t. } (\gamma_1-\beta)^2(\Omega^{\beta\beta}-\Omega^{\beta\delta^{(-s)}}(\Omega^{\delta^{(-s)}\delta^{(-s)}})^{-1}\Omega^{\delta^{(-s)}\beta})\leq \tilde\delta^2.
\end{gather*}
Recall that, given the definition of $\delta^{(s)}$ and $\delta^{(-s)}$, the constraint $\delta \geq 0$ is satisfied. Thus, we have
\begin{align*}
\omega(\tilde\delta;\mathcal{F},\mathcal{G})=\tilde\delta/\sqrt{\Omega^{\beta\beta}-\Omega^{\beta\delta^{(-s^{**})}}(\Omega^{\delta{(-s^{**})}\delta{(-s^{**})}})^{-1}\Omega^{\delta{(-s^{**})}\beta}},
\end{align*}
where $s^{**}$ is such that $\delta^{(-s^{**})}$ maximizes $\Omega^{\beta\delta^{(-s)}}(\Omega^{\delta^{(-s)}\delta^{(-s)}})^{-1}\Omega^{\delta^{(-s)}\beta}$ (subject to $\delta \geq 0$),  with a solution given by
\begin{gather*}
\gamma_{\tilde\delta,\mathcal{F},\mathcal{G}}^*=\left(\begin{array}{c}
\tilde\delta/(2\sqrt{\Omega^{\beta\beta}-\Omega^{\beta\delta^{(-s^{**})}}(\Omega^{\delta^{(-s^{**})}\delta^{(-s^{**})}})^{-1}\Omega^{\delta^{(-s^{**})}\beta}}) \\
0_{k\times 1}
\end{array}\right), \\ 
\theta_{\tilde\delta,\mathcal{F},\mathcal{G}}^*=\left(\begin{array}{c}
-\tilde\delta/(2\sqrt{\Omega^{\beta\beta}-\Omega^{\beta\delta^{(-s^{**})}}(\Omega^{\delta^{(-s^{**})}\delta^{(-s^{**})}})^{-1}\Omega^{\delta^{(-s^{**})}\beta}}) \\
\delta^{**}
\end{array}\right)
\end{gather*}
and midpoint
\[\theta_{M,\tilde\delta,\mathcal{F},\mathcal{G}}^*=(\theta_{\tilde\delta,\mathcal{F},\mathcal{G}}^*+\gamma_{\tilde\delta,\mathcal{F},\mathcal{G}}^*)/2=(0,\delta^{**\prime}/2)^\prime,\]
where $\delta^{**}$ has elements $\delta^{(s^{**})}$ and $\delta^{(-s^{**})}$. Formula (23) of AK thus yields
\begin{align*}
\hat L_{\tilde\delta,\mathcal{F},\mathcal{G}}&=\left(\tilde\delta\sqrt{\Omega^{\beta\beta}-\Omega^{\beta\delta^{(-s^{**})}}(\Omega^{\delta^{(-s^{**})}\delta^{(-s^{**})}})^{-1}\Omega^{\delta^{(-s^{**})\beta}}}\right)^{-1}(\gamma_{\tilde\delta,\mathcal{F},\mathcal{G}}^*-\theta_{\tilde\delta,\mathcal{F},\mathcal{G}}^*)^\prime\Omega^{-1}(Y-\theta_{M,\tilde\delta,\mathcal{F},\mathcal{G}}^*) \\
&=\left(\Omega^{\beta\beta}-\Omega^{\beta\delta^{(-s^{**})}}(\Omega^{\delta^{(-s^{**})}\delta^{(-s^{**})}})^{-1}\Omega^{\delta^{(-s^{**})}\beta}\right)^{-1}(1,\ddot{\delta}^{\prime})\left(\begin{array}{cc}
\Omega^{\beta\beta} & \Omega^{\beta\delta} \\
\Omega^{\delta\beta} & \Omega^{\delta\delta}
\end{array}\right)\left(\begin{array}{c}
Y_\beta \\
Y_\delta-\delta^{**}/2
\end{array}\right) \\
&=Y_\beta+\frac{\Omega^{\beta\delta^{(s^{**})}}-\Omega^{\beta\delta^{(-s^{**})}}(\Omega^{\delta^{(-s^{**})}\delta^{(-s^{**})}})^{-1}\Omega^{\delta^{(-s^{**})}\delta^{(-s^{**})}}}{\Omega^{\beta\beta}-\Omega^{\beta\delta^{(-s^{**})}}(\Omega^{\delta^{(-s^{**})}\delta^{(-s^{**})}})^{-1}\Omega^{\delta^{(-s^{**})}\beta}}Y_{\delta}^{(s^{**})} \\
&=Y_\beta-\Omega_{\beta \delta^{(s^{**})}} \Omega_{\delta^{(s^{**})} \delta^{(s^{**})}}^{-1} Y_{\delta}^{(s^{**})},
\end{align*}
where $\ddot{\delta}=\sqrt{\Omega^{\beta\beta}-\Omega^{\beta\delta^{(-s^{**})}}(\Omega^{\delta^{(-s^{**})}\delta^{(-s^{**})}})^{-1}\Omega^{\delta^{(-s^{**})}\beta}}\delta^{**}/\tilde\delta$ and the last equality follows from Lemma \ref{submatrix correspondence}. Similarly, Lemma \ref{submatrix correspondence} implies that
\[
\omega(\tilde\delta;\mathcal{F},\mathcal{G}) = \tilde \delta \sqrt{1-\Omega_{\beta \delta^{(s^{**})}} \Omega_{\delta^{(s^{**})} \delta^{(s^{**})}}^{-1} \Omega_{\delta^{(s^{**})}\beta  }}.
\]
Theorem 3.1 of AK then provides that among all upper one-sided CIs with coverage of at least $(1-\alpha)$ for all $\delta\geq 0$
\begin{equation} \label{form optimal} 
[Y_\beta-\Omega_{\beta\delta^{(s^{**})}}\Omega_{\delta^{(s^{**})}\delta^{(s^{**})}}^{-1}Y_\delta-z_{1-\alpha}\sqrt{1-\Omega_{\beta\delta^{(s^{**})}}\Omega_{\delta^{(s^{**})}\delta^{(s^{**})}}^{-1}\Omega_{\delta^{(s^{**})}\beta}},\infty)
\end{equation}
minimizes all excess length quantiles at $\delta=0$ and quantile levels greater than $\alpha$.

Next, we show that $s^{**} = s^*$. First, note that the excess length of any CI of the form given in \eqref{form optimal}\textemdash for some $\delta^{(s^{**})}$\textemdash at $\delta=0$ and any quantile greater than $\alpha$ is equal to $c \sqrt{1-\Omega_{\beta \delta^{(s^{**})}} \Omega_{\delta^{(s^{**})} \delta^{(s^{**})}}^{-1} \Omega_{\delta^{(s^{**})}\beta } }$ for some $c > 0$. Furthermore, recall (from the discussion in the main text) that for any CI of this form to have coverage of at least $(1-\alpha)$ for all $\delta\geq 0$ we need $\Omega_{\delta^{(s^{**})}\delta^{(s^{**})}}^{-1} \Omega_{\delta^{(s^{**})} \beta} \geq 0$.\footnote{Note that the condition $\Omega_{\delta^{(s)}\delta^{(s)}}^{-1} \Omega_{\delta^{(s)} \beta} \geq 0$ can also be derived from the modulus problem. To see this, note that the set of equations in \eqref{delta constraint} implies
\[
	\Omega^{\delta^{(s)} \delta^{(-s)}} \delta^{(-s)} = (\gamma_1 - \beta) \Omega^{\delta^{(s)} \beta} + \tilde \mu^{(s)}.
\]	
Plugging in the formula for $\delta^{(-s)}$ in equation \eqref{delta minus s}, we get
\begin{eqnarray*}
	& (\gamma_1 - \beta) \Omega^{\delta^{(s)} \delta^{(-s)}} ( \Omega^{\delta^{(-s)}\delta^{(-s)}} )^{-1} \Omega^{\delta^{(-s)}\beta} &= (\gamma_1 - \beta) \Omega^{\delta^{(s)} \beta} + \tilde \mu^{(s)} \\
	\Leftrightarrow & - \frac{\tilde \mu^{(s)}}{\gamma_1 - \beta} &=\Omega^{\delta^{(s)} \beta} - \Omega^{\delta^{(s)} \delta^{(-s)}} ( \Omega^{\delta^{(-s)}\delta^{(-s)}} )^{-1} \Omega^{\delta^{(-s)}\beta}\\
	\Leftrightarrow & \frac{\tilde \mu^{(s)}}{(\gamma_1 - \beta)(\Omega^{\beta \beta} - \Omega^{\beta \delta^{(-s)}} ( \Omega^{\delta^{(-s)}\delta^{(-s)}} )^{-1} \Omega^{\delta^{(-s)}\beta})}&=\Omega_{\delta^{(s)}\delta^{(s)}}^{-1} \Omega_{\delta^{(s)} \beta},
\end{eqnarray*}
where the last step uses Lemma \ref{submatrix correspondence} (with sub- and superscripts interchanged). As the left hand side is non-negative, we conclude that $\Omega_{\delta^{(s)}\delta^{(s)}}^{-1} \Omega_{\delta^{(s)} \beta} \geq 0$.} Therefore, the CI that minimizes excess length at $\delta = 0$ and all quantiles greater than $\alpha$, among all upper one-sided CIs with coverage of at least $(1-\alpha)$ for all $\delta\geq 0$, must also be equal to 
\[[Y_\beta-\Omega_{\beta\delta^{(s^*)}}\Omega_{\delta^{(s^*)}\delta^{(s^*)}}^{-1} Y_{\delta}^{(s^*)}-\sqrt{1-\Omega_{\beta\delta^{(s^*)}}\Omega_{\delta^{(s^*)}\delta^{(s^*)}}^{-1}\Omega_{\delta^{(s^*)}\beta}}z_{1-\alpha},\infty),\]
where $\delta^{(s^*)}$ is the subvector of $\delta$ that maximizes $\Omega_{\beta\delta^{(s)}}\Omega_{\delta^{(s)}\delta^{(s)}}^{-1}\Omega_{\delta^{(s)}\beta}$ amongst all subvectors for which $\Omega_{\delta^{(s)}\delta^{(s)}}^{-1}\Omega_{\delta^{(s)}\beta}\geq 0$.
\end{proof}


The following lemmata are used in the proofs of Proposition \ref{prop:c_u existence} and Theorem \ref{thm:uniform coverage}.

\begin{lemma} \label{lem:existence and uniqueness of c-tilde}
The function $\tilde c:\widetilde{\mathcal{C}}\rightarrow\mathbb{R}_\infty$ exists and is continuous.
\end{lemma}

\begin{proof} Consider the function $f:\mathbb{R}_\infty\times\widetilde{\mathcal{C}}\rightarrow [\alpha-1,\alpha]$ such that for $(Z_1,\tilde Z_2,\tilde Z_3)$ defined in \eqref{Z-dist},
\[f(\tilde c,c_u,\tilde\omega)=P(-\min\{z_{1-\frac{\alpha-\gamma}{2}},-\tilde Z_3+c_u\}\leq Z_1\leq \min\{z_{1-\frac{\alpha-\gamma}{2}},\tilde{Z}_2+\tilde c\})-(1-\alpha).\]  For $(\tilde c,c_u,\tilde\omega)\in\mathbb{R}_\infty\times\widetilde{\mathcal{C}}$ with $\omega_{12},\omega_{13}\neq 0$,
\begin{align*}
&f(\tilde c,c_u,\tilde\omega) \\
&=\int_{-\infty}^\infty\int_{-\infty}^\infty P(-\min\{z_{1-\frac{\alpha-\gamma}{2}},-\tilde Z_3+c_u\}\leq Z_1\leq \min\{z_{1-\frac{\alpha-\gamma}{2}},\tilde{Z}_2+\tilde c\}|\tilde Z_2=\tilde z_2,\tilde Z_3=\tilde z_3)g(\tilde z_2,\tilde z_3)d\tilde z_2d \tilde z_3-(1-\alpha) \\
&=\int_{-\infty}^\infty\int_{-\infty}^\infty \left[\Phi\left(\frac{\min\{z_{1-\frac{\alpha-\gamma}{2}},\tilde{z}_2+\tilde c\}-\mu(\tilde z_2,\tilde z_3)}{\sigma(\tilde\omega)}\right)-\Phi\left(\frac{-\min\{z_{1-\frac{\alpha-\gamma}{2}},-\tilde{z}_3+ c_u\}-\mu(\tilde z_2,\tilde z_3)}{\sigma(\tilde\omega)}\right)\right]\\
&\qquad\qquad\quad\times\mathbf{1}(\min\{z_{1-\frac{\alpha-\gamma}{2}},\tilde{z}_2+\tilde c\}\geq-\min\{z_{1-\frac{\alpha-\gamma}{2}},-\tilde{z}_3+ c_u\})g(\tilde z_2,\tilde z_3)d\tilde z_2d\tilde z_3-(1-\alpha) \\
&=\int_{-\infty}^\infty\left[\int_{-\infty}^{z_{1-\frac{\alpha-\gamma}{2}}-\tilde c}\Phi\left(\frac{\tilde{z}_2+\tilde c-\mu(\tilde z_2,\tilde z_3)}{\sigma(\tilde\omega)}\right)\mathbf{1}(\tilde{z}_2+\tilde c\geq-\min\{z_{1-\frac{\alpha-\gamma}{2}},-\tilde{z}_3+ c_u\})\right. \\
&\qquad\qquad +\left.\int_{z_{1-\frac{\alpha-\gamma}{2}}-\tilde c}^\infty\Phi\left(\frac{z_{1-\frac{\alpha-\gamma}{2}}-\mu(\tilde z_2,\tilde z_3)}{\sigma(\tilde\omega)}\right)\mathbf{1}(z_{1-\frac{\alpha-\gamma}{2}}\geq-\min\{z_{1-\frac{\alpha-\gamma}{2}},-\tilde{z}_3+ c_u\})\right]g(\tilde z_2,\tilde z_3)d\tilde z_2d\tilde z_3 \\
&\quad-\int_{-\infty}^\infty\left[\int_{-\infty}^{c_u-z_{1-\frac{\alpha-\gamma}{2}}}\Phi\left(\frac{-z_{1-\frac{\alpha-\gamma}{2}}-\mu(\tilde z_2,\tilde z_3)}{\sigma(\tilde\omega)}\right)\mathbf{1}(\min\{z_{1-\frac{\alpha-\gamma}{2}},\tilde{z}_2+\tilde c\}\geq-z_{1-\frac{\alpha-\gamma}{2}})\right. \\
&\qquad\qquad\quad +\left.\int_{c_u-z_{1-\frac{\alpha-\gamma}{2}}}^\infty\Phi\left(\frac{\tilde{z}_3-c_u-\mu(\tilde z_2,\tilde z_3)}{\sigma(\tilde\omega)}\right)\mathbf{1}(\min\{z_{1-\frac{\alpha-\gamma}{2}},\tilde{z}_2+\tilde c\}\geq\tilde z_3-c_u)\right]g(\tilde z_2,\tilde z_3)d\tilde z_3d\tilde z_2-(1-\alpha) \\
&=\int_{-\infty}^{c_u-z_{1-\frac{\alpha-\gamma}{2}}}\int_{-z_{1-\frac{\alpha-\gamma}{2}}-\tilde c}^{z_{1-\frac{\alpha-\gamma}{2}}-\tilde c}\Phi\left(\frac{\tilde z_2+\tilde c-\mu(\tilde z_2,\tilde z_3)}{\sigma(\tilde\omega)}\right)g(\tilde z_2,\tilde z_3)d\tilde z_2\tilde z_3 \\
&\quad+\int_{c_u-z_{1-\frac{\alpha-\gamma}{2}}}^{\infty}\int_{\tilde z_3-c_u-\tilde c}^{z_{1-\frac{\alpha-\gamma}{2}}-\tilde c}\Phi\left(\frac{\tilde z_2+\tilde c-\mu(\tilde z_2,\tilde z_3)}{\sigma(\tilde\omega)}\right)g(\tilde z_2,\tilde z_3)d\tilde z_2\tilde z_3 \\
&\quad + \int_{-\infty}^{c_u+z_{1-\frac{\alpha-\gamma}{2}}}\int_{z_{1-\frac{\alpha-\gamma}{2}}-\tilde c}^{\infty}\Phi\left(\frac{z_{1-\frac{\alpha-\gamma}{2}}-\mu(\tilde z_2,\tilde z_3)}{\sigma(\tilde\omega)}\right)g(\tilde z_2,\tilde z_3)d\tilde z_2\tilde z_3 \\
&\quad - \int_{-\infty}^{c_u-z_{1-\frac{\alpha-\gamma}{2}}}\int_{-z_{1-\frac{\alpha-\gamma}{2}}-\tilde c}^{\infty}\Phi\left(\frac{-z_{1-\frac{\alpha-\gamma}{2}}-\mu(\tilde z_2,\tilde z_3)}{\sigma(\tilde\omega)}\right)g(\tilde z_2,\tilde z_3)d\tilde z_2\tilde z_3 \\
&\quad - \int_{c_u-z_{1-\frac{\alpha-\gamma}{2}}}^{c_u+z_{1-\frac{\alpha-\gamma}{2}}}\int_{z_{1-\frac{\alpha-\gamma}{2}}-\tilde c}^{\infty}\Phi\left(\frac{\tilde z_3-c_u-\mu(\tilde z_2,\tilde z_3)}{\sigma(\tilde\omega)}\right)g(\tilde z_2,\tilde z_3)d\tilde z_2\tilde z_3 \\
&\quad - \int_{c_u-z_{1-\frac{\alpha-\gamma}{2}}}^{\infty}\int_{\tilde z_3-c_u-\tilde c}^{z_{1-\frac{\alpha-\gamma}{2}}-\tilde c}\Phi\left(\frac{\tilde z_3-c_u-\mu(\tilde z_2,\tilde z_3)}{\sigma(\tilde\omega)}\right)g(\tilde z_2,\tilde z_3)d\tilde z_2\tilde z_3-(1-\alpha),
\end{align*}
where $g(\cdot)$ denotes the probability density function of $(\tilde Z_2,\tilde Z_3)$, $\mu(\tilde z_2,\tilde z_3)=(\omega_{12},\omega_{13})\Sigma_{22}^{-1}(\tilde z_2,\tilde z_3)^{\prime}$ and $\sigma(\tilde\omega)=\sqrt{1-(\omega_{12},\omega_{13})\Sigma_{22}^{-1}( \omega_{12}, \omega_{13})^{\prime}}$ with
\[\Sigma_{22}=\left(\begin{array}{cc}
\omega_{12} & \omega_{23} \\
\omega_{23} & \omega_{13}
\end{array}\right).\]
This function is clearly continuously differentiable.  In addition,
\begin{align*}
\frac{\partial f(\tilde c,c_u,\tilde\omega)}{\partial\tilde c}&=\int_{c_u+z_{1-\frac{\alpha-\gamma}{2}}}^\infty\left[\Phi\left(\frac{\tilde z_3-c_u-\mu(z_{1-\frac{\alpha-\gamma}{2}}-\tilde c,\tilde z_3)}{\sigma(\tilde \omega)}\right)\right. \\
&\qquad\qquad\qquad \left.-\Phi\left(\frac{z_{1-\frac{\alpha-\gamma}{2}}-\mu(z_{1-\frac{\alpha-\gamma}{2}}-\tilde c,\tilde z_3)}{\sigma(\tilde \omega)}\right)\right]g(z_{1-\frac{\alpha-\gamma}{2}}-\tilde c,\tilde z_3)d\tilde z_3 \\
&\quad +\int_{-\infty}^{c_u-z_{1-\frac{\alpha-\gamma}{2}}}\int_{-z_{1-\frac{\alpha-\gamma}{2}}-\tilde c}^{z_{1-\frac{\alpha-\gamma}{2}}-\tilde c}\frac{1}{\sigma(\tilde \omega)}\phi\left(\frac{\tilde z_2+\tilde c-\mu(\tilde z_2,\tilde z_3)}{\sigma(\tilde\omega)}\right)g(\tilde z_2,\tilde z_3)d\tilde z_2d\tilde z_3 \\
&\quad +\int_{c_u-z_{1-\frac{\alpha-\gamma}{2}}}^{\infty}\int_{\tilde z_3-c_u-\tilde c}^{z_{1-\frac{\alpha-\gamma}{2}}-\tilde c}\frac{1}{\sigma(\tilde \omega)}\phi\left(\frac{\tilde z_2+\tilde c-\mu(\tilde z_2,\tilde z_3)}{\sigma(\tilde\omega)}\right)g(\tilde z_2,\tilde z_3)d\tilde z_2d\tilde z_3>0
\end{align*}
for all $(\tilde c,c_u,\tilde\omega)\in\mathbb{R}\times\widetilde{\mathcal{C}}$ with $\omega_{12},\omega_{13}\neq 0$ since all three integrals are stricly positive.  

Next, note that for any $(c_u,\tilde\omega)\in\widetilde{\mathcal{C}}$ with $\omega_{12}, \omega_{13}\neq 0$, there exists $\tilde c\in\mathbb{R}_{\infty}$ such that $f(\tilde c,c_u,\tilde\omega)=0$ since $f(\cdot,c_u,\tilde\omega)$ is continuously strictly increasing,
\[\lim_{\tilde c\rightarrow -\infty}f(\tilde c,c_u,\tilde\omega)=-(1-\alpha)<0\]
and
\[\lim_{\tilde c\rightarrow \infty}f(\tilde c,c_u,\tilde\omega)=P(-\min\{z_{1-\frac{\alpha-\gamma}{2}},-\tilde Z_3+c_u\}\leq Z_1\leq z_{1-\frac{\alpha-\gamma}{2}})-(1-\alpha)\geq 0\]
by \eqref{lower bar c_u} and the fact that $P(-\min\{z_{1-\frac{\alpha-\gamma}{2}},-\tilde Z_3+c_u\}\leq Z_1\leq z_{1-\frac{\alpha-\gamma}{2}})$ is increasing in $c_u$.

Thus, the implicit function theorem implies that $\tilde{c}(c_u,\tilde\omega)$ is continuous at all $(c_u,\tilde\omega)\in \widetilde{\mathcal{C}}$ with $\omega_{12},\omega_{13}\neq 0$ and is therefore continuous at all $(c_u,\tilde\omega)\in \widetilde{\mathcal{C}}$ by the definition of $\tilde{c}(c_u,\tilde\omega)$ at $(c_u,\tilde\omega)\in \widetilde{\mathcal{C}}$ with $\omega_{12}=0$ or $\omega_{13}=0$. 
\end{proof}

\begin{lemma} \label{lem:zero var size}
For any $(c_u,\tilde\omega)\in \widetilde{\mathcal{C}}$ with $\omega_{12}=0$ or $\omega_{13}= 0$, 
\[P(- \min \{z_{1-(\alpha-\gamma)/2},-\tilde  Z_3+c_u\} \leq Z_1 \leq \min\{z_{1-(\alpha-\gamma)/2},\tilde  Z_2+\tilde c(c_u,\tilde \omega)\})=1-\alpha.\]
\end{lemma}

\begin{proof} By \eqref{cl_cc},
\[P(- \min \{z_{1-(\alpha-\gamma)/2},-\tilde  Z_3+c_u\} \leq Z_1 \leq \min\{z_{1-(\alpha-\gamma)/2},\tilde  Z_2+\tilde c(c_u,\tilde \omega)\})=1-\alpha\]
for all $(c_u,\tilde\omega)\in \widetilde{\mathcal{C}}$ with $\omega_{12},\omega_{13}\neq 0$.  Since the probability on the left hand side of this equality is continuous in $\tilde\omega$ by Lemma \ref{lem:existence and uniqueness of c-tilde} and the continuity of the density function of $(Z_1,\widetilde Z_2,\widetilde Z_3)$ in $\tilde\omega$, the result immediately follows. 
\end{proof}

\begin{proof}[Proof of Proposition \ref{prop:c_u existence}] 
Very similar arguments to those given in the proof of Lemma 1 provide that $\underline{c_u}:\bar{\mathcal{S}}\rightarrow \mathbb{R}$ exists and is continuous.  Thus, $[\underline{c_u}(\cdot),\infty]$ is nonempty, compact-valued and continuous when treated as a correspondence from $\bar{\mathcal{S}}$ into $\mathbb{R}_{\infty}$ (see above).

Next, note that the minimand in \eqref{c_u def} is \begin{gather*}
E[\max\{\min\{z_{1-(\alpha-\gamma)/2},\tilde  Z_2+ \tilde c(c_u,\tilde \omega)\}+\min \{z_{1-(\alpha-\gamma)/2},-\tilde  Z_3+c_u\},0\}] \\
=\int_{-\infty}^\infty \int_{-\infty}^\infty \max\{\min\{z_{1-(\alpha-\gamma)/2},\tilde z_2+\tilde c(c_u,\tilde \omega)\}+\min\{z_{1-(\alpha-\gamma)/2},-\tilde z_3+c_u\},0\}g(\tilde z_2,\tilde z_3)d\tilde z_2d\tilde z_3 \\
=\int_{-\infty}^{c_u-z_{1-(\alpha-\gamma)/2}}\int_{-z_{1-(\alpha-\gamma)/2}-\tilde c(c_u,\tilde \omega)}^{z_{1-(\alpha-\gamma)/2}-\tilde c(c_u,\tilde \omega)}(\tilde z_2+\tilde c(c_u,\tilde \omega)+z_{1-(\alpha-\gamma)/2})g(\tilde z_2,\tilde z_3)d\tilde z_2d\tilde z_3 \\
+\int_{c_u-z_{1-(\alpha-\gamma)/2}}^{\infty}\int_{\tilde z_3-\tilde c(c_u,\tilde \omega)-c_u}^{z_{1-(\alpha-\gamma)/2}-\tilde c(c_u,\tilde \omega)}(\tilde z_2-\tilde z_3+\tilde c(c_u,\tilde \omega)+c_u)g(\tilde z_2,\tilde z_3)d\tilde z_2d\tilde z_3 \\
+\int_{-\infty}^{c_u-z_{1-(\alpha-\gamma)/2}}\int_{z_{1-(\alpha-\gamma)/2}-\tilde c(c_u,\tilde \omega)}^{\infty}2z_{1-(\alpha-\gamma)/2}g(\tilde z_2,\tilde z_3)d\tilde z_2d\tilde z_3 \\
+\int_{c_u-z_{1-(\alpha-\gamma)/2}}^{c_u+z_{1-(\alpha-\gamma)/2}}\int_{z_{1-(\alpha-\gamma)/2}-\tilde c(c_u,\tilde \omega)}^{\infty}(z_{1-(\alpha-\gamma)/2}-\tilde z_3+c_u)g(\tilde z_2,\tilde z_3)d\tilde z_2d\tilde z_3,
\end{gather*}
where $g(\cdot)$ denotes the probability density function of $(\tilde Z_2,\tilde Z_3)$.  When treated as a function from $\widetilde{\mathcal{C}}$ into $\mathbb{R}_+$, this expression is clearly continuous in $(c_u,\tilde\omega)$ since $\tilde c:\widetilde{\mathcal{C}}\rightarrow\mathbb{R}$ is continuous by Lemma \ref{lem:existence and uniqueness of c-tilde}.  The maximum theorem then implies the statement of the proposition. \end{proof}

\begin{proof}[Proof of Proposition \ref{prop:c_ell}]
Let $(Z_1^*,\tilde Z_2^*, \tilde Z_3^*)' = (- Z_1,- \tilde Z_3,- \tilde Z_2)'$ and note that 
\[
	P(- \min \{z_{1-(\alpha-\gamma)/2},-\tilde  Z_3+c_u(\tilde\omega_{12},\tilde\omega_{13},\tilde\omega_{23})\} \leq Z_1 \leq \min\{z_{1-(\alpha-\gamma)/2},\tilde  Z_2+ c_\ell(\tilde\omega_{12},\tilde\omega_{13},\tilde\omega_{23})\})=1-\alpha
	\]
	and 
\[
E[\max\{\min\{z_{1-(\alpha-\gamma)/2},\tilde  Z_2+  c_\ell(\tilde\omega_{12},\tilde\omega_{13},\tilde\omega_{23})\}+\min \{z_{1-(\alpha-\gamma)/2},-\tilde  Z_3+c_u(\tilde\omega_{12},\tilde\omega_{13},\tilde\omega_{23})\},0\}]
\]
	are equivalent to
\[
	P(- \min \{z_{1-(\alpha-\gamma)/2},-\tilde  Z_3^*+c_\ell(\tilde\omega_{12},\tilde\omega_{13},\tilde\omega_{23})\} \leq Z_1^* \leq \min\{z_{1-(\alpha-\gamma)/2},\tilde  Z_2^*+ c_u(\tilde\omega_{12},\tilde\omega_{13},\tilde\omega_{23})\})=1-\alpha
\]
and
\[
E[\max\{\min\{z_{1-(\alpha-\gamma)/2},\tilde  Z_2^*+  c_u(\tilde\omega_{12},\tilde\omega_{13},\tilde\omega_{23})\}+\min \{z_{1-(\alpha-\gamma)/2},-\tilde  Z_3^*+c_\ell(\tilde\omega_{12},\tilde\omega_{13},\tilde\omega_{23})\},0\}].
\]
The result then follows by noting that $(Z_1^*,\tilde Z_2^*, \tilde Z_3^*)'  \sim (Z_1,\tilde Z_3,\tilde Z_2)'$. 
\end{proof}

\begin{proof}[Proof of Theorem \ref{thm:uniform coverage}] First, we provide the proof of the statement of the theorem for $CI_{u,n}(\cdot)$. Under (i)--(ii), standard subsequencing arguments in the uniform inference literature (see e.g.,~\citealp{AG10,McC17,ACG20}) provide that 
\begin{equation}
\liminf_{n\rightarrow\infty}\text{ }\inf_{\lambda\in\Lambda}P_\lambda\left(b\in CI_{u,n}(\hat b_n,\hat d_n;\widehat\Sigma_n)\right)=\lim_{n\rightarrow\infty}P_{\lambda_{k_n,\mathfrak{b},\mathfrak{d},\Sigma^*}}\left(b_{k_n,\mathfrak{b}}\in CI_{u,k_n}(\hat b_{k_n},\hat d_{k_n};\widehat\Sigma_{k_n})\right) \label{1-sided subseq prob}
\end{equation}
for a subsequence $\{k_n:n\geq 1\}$ of $\{n:n\geq 1\}$ such that $\lambda_{k_n,\mathfrak{b},\mathfrak{d},\Sigma^*}\in\Lambda$ for all $n\geq 1$, $\sqrt{k_n}(b_{k_n,\mathfrak{b}},d_{k_n,\mathfrak{d}})\rightarrow (\mathfrak{b},\mathfrak{d}) $ and $\Sigma_{k_n,\Sigma^*}\rightarrow \Sigma^*$ for some $(\mathfrak{b},\mathfrak{d},\Sigma^*)\in \mathbb{R}_\infty\times \mathbb{R}_{+,\infty}^k\times \Phi$ with $\lambda_{\min}(\Sigma^*)\geq \kappa$ and $\lambda_{\max}(\Sigma^*)\leq \kappa^{-1}$.  

Let $\Omega^*=\diag(\Sigma^*)^{-1/2}\Sigma^*\diag(\Sigma^*)^{-1/2}$ and $s^*(\Omega^*)$ denote the subset of the set of indices $\{1,\ldots,k\}$ that maximizes $\Omega_{bd^{(s)}}^*\Omega_{d^{(s)}d^{(s)}}^{*-1}\Omega_{d^{(s)}b}^*$ amongst all subsets of indices $s\subset\{1,\ldots,k\}$ such that the elements of $\Omega_{bd^{(s)}}^*\Omega_{d^{(s)}d^{(s)}}^{*-1}$ are non-negative.  Since $\widehat\Omega_{k_n}\overset{p}\longrightarrow \Omega^*$ under $\{\lambda_{k_n,\mathfrak{b},\mathfrak{d},\Sigma^*}:n\geq 1\}$ by (ii), (iii) and (v), note that $s^*(\widehat\Omega_{k_n})\overset{p}\longrightarrow s^*(\Omega^*)$ under $\{\lambda_{k_n,\mathfrak{b},\mathfrak{d},\Sigma^*}:n\geq 1\}$.  Thus, \eqref{1-sided subseq prob} implies
\begin{align}
&\liminf_{n\rightarrow\infty}\text{ }\inf_{\lambda\in\Lambda}P_\lambda\left(b\in CI_{u,n}(\hat b_n,\hat d_n;\widehat\Sigma_n)\right)=\lim_{n\rightarrow\infty}P_{\lambda_{k_n,\mathfrak{b},\mathfrak{d},\Sigma^*}}\left( b_{k_n,\mathfrak{b}}\geq \hat b_{k_n}-\frac{\sqrt{\widehat\Sigma_{k_n,bb}}}{\sqrt{k_n}}\min\left\{z_{1-\alpha+\gamma},\right.\right. \notag  \\
&\quad\left.\left.\widehat\Omega_{k_n,bd^{(s^*)}}\widehat\Omega_{k_n,d^{(s^*)}d^{(s^*)}}^{-1}\diag(\widehat\Sigma_{k_n,d^{(s^*)}d^{(s^*)}})^{-1/2}\sqrt{k_n}\hat d_{k_n}^{(s^*)}+c\left(\widehat\Omega_{k_n,bd^{(s^*)}}\widehat\Omega_{k_n,d^{(s^*)}d^{(s^*)}}^{-1}\widehat\Omega_{k_n,d^{(s^*)}b}\right)\right\}\right) \notag \\
&=\lim_{n\rightarrow\infty}P_{\lambda_{k_n,\mathfrak{b},\mathfrak{d},\Sigma^*}}\left(\frac{\sqrt{k_n}(\hat b_{k_n}-b_{k_n,\mathfrak{b}})}{\sqrt{\widehat\Sigma_{k_n,bb}}}\leq \min\left\{z_{1-\alpha+\gamma},\right.\right. \notag \\
&\quad\left.\left.\widehat\Omega_{k_n,bd^{(s^*)}}\widehat\Omega_{k_n,d^{(s^*)}d^{(s^*)}}^{-1}\diag(\widehat\Sigma_{k_n,d^{(s^*)}d^{(s^*)}})^{-1/2}\sqrt{k_n}\hat d_{k_n}^{(s^*)}+c\left(\widehat\Omega_{k_n,bd^{(s^*)}}\widehat\Omega_{k_n,d^{(s^*)}d^{(s^*)}}^{-1}\widehat\Omega_{k_n,d^{(s^*)}b}\right)\right\}\right) \notag \\
&=\begin{cases}
P\left(Z_1\leq \min\left\{z_{1-\alpha+\gamma},\Omega_{bd^{(s^*)}}^*\Omega_{d^{(s^*)}d^{(s^*)}}^{-1}Y_\delta^{(s^*)}+c\left(\Omega_{bd^{(s^*)}}^*\Omega_{d^{(s^*)}d^{(s^*)}}^{*-1}\Omega_{d^{(s^*)}b}^*\right)\right\}\right) & \text{if } \|\mathfrak{d}^{(s^*)}\|<\infty, \\
P(Z_1\leq z_{1-\alpha+\gamma}) & \text{if } \|\mathfrak{d}^{(s^*)}\|=\infty
\end{cases} \label{one-sided case 2}
 \end{align}
 by (ii)--(v) and Proposition \ref{prop:c existence}, where we use $s^*$ as shorthand for $s^*(\Omega^*)$ and
\begin{equation*}
\left(\begin{array}{c}
Z_1 \\
Y_{\delta}^{(s^*)}
\end{array}\right)\sim \mathcal{N}
\left(\left(\begin{array}{c}
0 \\
\delta^{(s^*)}
\end{array}\right), 
\left(\begin{array}{cc}
1 & \Omega_{bd^{(s^*)}} \\
\Omega_{d^{(s^*)}b} & \Omega_{d^{(s^*)}d^{(s^*)}}
\end{array}\right)\right)
\end{equation*} 
with $\delta^{(s^*)}=\diag(\Sigma_{d^{(s^*)}d^{(s^*)}}^*)^{-1/2}\mathfrak{d}^{(s^*)}$.  Now for the $\|\mathfrak{d}^{(s^*)}\|<\infty$ case, since $\Omega_{bd^{(s^*)}}^*\Omega_{d^{(s^*)}d^{(s^*)}}^{-1}\delta^{(s^*)}\geq 0$ by (i), \eqref{bd drifting seq} and (v),
\begin{gather}
P\left(Z_1\leq \min\left\{z_{1-\alpha+\gamma},\Omega_{bd^{(s^*)}}^*\Omega_{d^{(s^*)}d^{(s^*)}}^{-1}Y_\delta^{(s^*)}+c\left(\Omega_{bd^{(s^*)}}^*\Omega_{d^{(s^*)}d^{(s^*)}}^{*-1}\Omega_{d^{(s^*)}b}^*\right)\right\}\right) \notag \\
=P\left(Z_1\leq \min\left\{z_{1-\alpha+\gamma},\Omega_{bd^{(s^*)}}^*\Omega_{d^{(s^*)}d^{(s^*)}}^{-1}\delta^{(s^*)}+\tilde Z_2+c\left(\Omega_{bd^{(s^*)}}^*\Omega_{d^{(s^*)}d^{(s^*)}}^{*-1}\Omega_{d^{(s^*)}b}^*\right)\right\}\right) \notag \\
\geq P\left(Z_1\leq \min\left\{z_{1-\alpha+\gamma},\tilde Z_2+c\left(\Omega_{bd^{(s^*)}}^*\Omega_{d^{(s^*)}d^{(s^*)}}^{*-1}\Omega_{d^{(s^*)}b}^*\right)\right\}\right)=1-\alpha \label{finite d case}
\end{gather}
 by the definition of $c(\cdot)$ in \eqref{c-function def}, where 
\begin{equation*}
\left(\begin{array}{c}
Z_1 \\
\tilde Z_2
\end{array}\right)\sim \mathcal{N}
\left(\left(\begin{array}{c}
0 \\
0
\end{array}\right), 
\left(\begin{array}{cc}
1 & \Omega_{bd^{(s^*)}}^*\Omega_{d^{(s^*)}d^{(s^*)}}^{*-1}\Omega_{d^{(s^*)}b}^* \\
\Omega_{bd^{(s^*)}}^*\Omega_{d^{(s^*)}d^{(s^*)}}^{*-1}\Omega_{d^{(s^*)}b}^* & \Omega_{bd^{(s^*)}}^*\Omega_{d^{(s^*)}d^{(s^*)}}^{*-1}\Omega_{d^{(s^*)}b}^*
\end{array}\right)\right).
\end{equation*} 
On the other hand, for the $\|\mathfrak{d}^{(s^*)}\|=\infty$ case,
\begin{equation}
P(Z_1\leq z_{1-\alpha+\gamma})=1-\alpha+\gamma>1-\alpha. \label{infinite d case}
\end{equation}
Together, \eqref{one-sided case 2}--\eqref{infinite d case} yield the lower bound in the statement of the theorem for $CI_{u,n}(\cdot)$. 

To prove the upper bound, note that by nearly identical arguments to those used to establish \eqref{1-sided subseq prob},
\begin{equation*}
\limsup_{n\rightarrow\infty}\text{ }\sup_{\lambda\in\Lambda}P_\lambda\left(b\in CI_{u,n}(\hat b_n,\hat d_n;\widehat\Sigma_n)\right)=\lim_{n\rightarrow\infty}P_{\lambda_{m_n,\mathfrak{b},\mathfrak{d},\Sigma^*}}\left(b_{m_n,\mathfrak{b}}\in CI_{u,m_n}(\hat b_{m_n},\hat d_{m_n};\widehat\Sigma_{m_n})\right) 
\end{equation*}
for a subsequence $\{m_n:n\geq 1\}$ of $\{n:n\geq 1\}$ such that $\lambda_{m_n,\mathfrak{b},\mathfrak{d},\Sigma^*}\in\Lambda$ for all $n\geq 1$, $\sqrt{m_n}(b_{m_n,\mathfrak{b}},d_{m_n,\mathfrak{d}})\rightarrow (\mathfrak{b},\mathfrak{d}) $ and $\Sigma_{m_n,\Sigma^*}\rightarrow \Sigma^*$ for some $(\mathfrak{b},\mathfrak{d},\Sigma^*)\in \mathbb{R}_\infty\times \mathbb{R}_{+,\infty}^k\times \Phi$ with $\lambda_{\min}(\Sigma^*)\geq \kappa$ and $\lambda_{\max}(\Sigma^*)\leq \kappa^{-1}$.  Note that for the probability to the left of the inequality in \eqref{finite d case},
\begin{gather*}
P\left(Z_1\leq \min\left\{z_{1-\alpha+\gamma},\Omega_{bd^{(s^*)}}^*\Omega_{d^{(s^*)}d^{(s^*)}}^{-1}\delta^{(s^*)}+\tilde Z_2+c\left(\Omega_{bd^{(s^*)}}^*\Omega_{d^{(s^*)}d^{(s^*)}}^{*-1}\Omega_{d^{(s^*)}b}^*\right)\right\}\right) \\
\leq P\left(Z_1\leq z_{1-\alpha+\gamma}\right)=1-\alpha+\gamma.
\end{gather*}
Then, nearly identical reasoning used to establish \eqref{one-sided case 2}--\eqref{infinite d case}, replacing ``$\liminf_{n\rightarrow\infty}\inf_{\lambda\in\Lambda}$'' with ``$\limsup_{n\rightarrow\infty}\sup_{\lambda\in\Lambda}$'' and the subsequences $\{k_n:n\geq 1\}$ and $\{\lambda_{k_n,\mathfrak{b},\mathfrak{d},\Sigma^*}\in\Lambda:n\geq 1\}$ with $\{\lambda_{m_n,\mathfrak{b},\mathfrak{d},\Sigma^*}\in\Lambda:n\geq 1\}$, yields the upper bound in the statement of the theorem for $CI_{u,n}(\cdot)$. 

Next, we provide the proof of the statement of the theorem for $CI_{t,n}(\cdot)$. The same arguments used to establish \eqref{1-sided subseq prob} also apply to $CI_{t,n}(\hat b_n,\hat d_n;\widehat\Sigma_n)$ so that it suffices to consider $\lim_{n\rightarrow\infty}P_{\lambda_{k_n,\mathfrak{b},\mathfrak{d},\Sigma^*}}(b_{k_n,\mathfrak{b}}\in CI_{T,k_n}(\hat b_{k_n},\hat d_{k_n};\widehat\Sigma_{k_n}))$ under the same subsequences $\{k_n:n\geq 1\}$ of $\{n:n\geq 1\}$ and sequences of parameters $\{\lambda_{k_n,\mathfrak{b},\mathfrak{d},\Sigma^*}:n\geq 1\}$ described in the proof of the statement of the theorem for $CI_{u,n}(\cdot)$.

Using analogous notation and reasoning to the proof above, 
\begin{align}
&\liminf_{n\rightarrow\infty}\text{ }\inf_{\lambda\in\Lambda}P_\lambda\left(b\in CI_{t,n}(\hat b_n,\hat d_n;\widehat\Sigma_n)\right)=\lim_{n\rightarrow\infty}P_{\lambda_{k_n,\mathfrak{b},\mathfrak{d},\Sigma^*}}\left( \hat b_{k_n}-\frac{\sqrt{\widehat\Sigma_{k_n,bb}}}{\sqrt{k_n}}\min\left\{z_{1-\frac{\alpha-\gamma}{2}},\right.\right. \notag  \\
&\quad\left.\widehat\Omega_{k_n,bd^{(s_1^*)}}\widehat\Omega_{k_n,d^{(s_1^*)}d^{(s_1^*)}}^{-1}\diag(\widehat\Sigma_{k_n,d^{(s_1^*)}d^{(s_1^*)}})^{-1/2}\sqrt{k_n}\hat d_{k_n}^{(s_1^*)}+c_\ell\left(\widehat{\tilde \Omega}_{k_n}^{(s_1^*,s_2^*)}\right)\right\}\leq b_{k_n,\mathfrak{b}}\leq \hat b_{k_n} \notag\\
& \left.+\frac{\sqrt{\widehat\Sigma_{k_n,bb}}}{\sqrt{k_n}}\min\left\{z_{1-\frac{\alpha-\gamma}{2}},-\widehat\Omega_{k_n,bd^{(s_2^*)}}\widehat\Omega_{k_n,d^{(s_2^*)}d^{(s_2^*)}}^{-1}\diag(\widehat\Sigma_{k_n,d^{(s_2^*)}d^{(s_2^*)}})^{-1/2}\sqrt{k_n}\hat d_{k_n}^{(s_2^*)}+c_u\left(\widehat{\tilde \Omega}_{k_n}^{(s_1^*,s_2^*)}\right)\right\} \right) \notag \\
&=\lim_{n\rightarrow\infty}P_{\lambda_{k_n,\mathfrak{b},\mathfrak{d},\Sigma^*}}\left(-\min\left\{z_{1-\frac{\alpha-\gamma}{2}},-\widehat\Omega_{k_n,bd^{(s_2^*)}}\widehat\Omega_{k_n,d^{(s_2^*)}d^{(s_2^*)}}^{-1}\diag(\widehat\Sigma_{k_n,d^{(s_2^*)}d^{(s_2^*)}})^{-1/2}\sqrt{k_n}\hat d_{k_n}^{(s_2^*)}+c_u\left(\widehat{\tilde \Omega}_{k_n}^{(s_1^*,s_2^*)}\right)\right\} \right. \notag \\
&\quad  \leq\frac{\sqrt{k_n}(\hat b_{k_n}-b_{k_n,\mathfrak{b}})}{\sqrt{\widehat\Sigma_{k_n,bb}}} \notag \\
&\quad \left.\leq \min\left\{z_{1-\frac{\alpha-\gamma}{2}},\widehat\Omega_{k_n,bd^{(s_1^*)}}\widehat\Omega_{k_n,d^{(s_1^*)}d^{(s_1^*)}}^{-1}\diag(\widehat\Sigma_{k_n,d^{(s_1^*)}d^{(s_1^*)}})^{-1/2}\sqrt{k_n}\hat d_{k_n}^{(s_1^*)}+\tilde c \left(c_u\left(\widehat{\tilde \Omega}_{k_n}^{(s_1^*,s_2^*)}\right),\widehat{\tilde \Omega}_{k_n}^{(s_1^*,s_2^*)}\right)\right\}\right). \label{two-sided case 4}
\end{align}
Since $\widehat{\tilde \Omega}_{k_n}^{(s_1^*,s_2^*)}\overset{p}\longrightarrow {\tilde \Omega}^{*(s_1^*,s_2^*)}$ under $\{\lambda_{k_n,\mathfrak{b},\mathfrak{d},\Sigma^*}:n\geq 1\}$ as $k_n\rightarrow\infty$ by (ii), (iii) and (v), there exists a subsequence $\{l_n:n\geq 1\}$ of $\{k_n:n\geq 1\}$ such that $\widehat{\tilde \Omega}_{l_n}^{(s_1^*,s_2^*)}\overset{a.s.}\longrightarrow {\tilde \Omega}^{*(s_1^*,s_2^*)}$ under $\{\lambda_{l_n,\mathfrak{b},\mathfrak{d},\Sigma^*}:n\geq 1\}$ as $l_n\rightarrow\infty$.  Next, by the properties of $\tilde c_u:\bar{\mathcal{S}}\rightrightarrows \mathbb{R}$ given in Proposition \ref{prop:c_u existence}, there exists a subsequence $\{h_n:n\geq 1\}$ of $\{l_n:n\geq 1\}$ for which the subsequence $\left\{c_u\left(\widehat{\tilde \Omega}_{h_n}^{(s_1^*,s_2^*)}\right):n\geq 1\right\}$ of $\left\{c_u\left(\widehat{\tilde \Omega}_{l_n}^{(s_1^*,s_2^*)}\right):n\geq 1\right\}$ is such that $c_u\left(\widehat{\tilde \Omega}_{h_n}^{(s_1^*,s_2^*)}\right)\in \tilde c_u\left(\widehat{\tilde \Omega}_{h_n}^{(s_1^*,s_2^*)}\right)$ for all $n\geq 1$ and $c_u\left(\widehat{\tilde \Omega}_{h_n}^{(s_1^*,s_2^*)}\right)\overset{a.s.}\longrightarrow c_u^*\left({\tilde \Omega}^{*(s_1^*,s_2^*)}\right)$ for some $c_u^*\left({\tilde \Omega}^{*(s_1^*,s_2^*)}\right)\in \tilde c_u\left({\tilde \Omega}^{*(s_1^*,s_2^*)}\right)$ as $n\rightarrow\infty$.  In conjunction with Lemma \ref{lem:existence and uniqueness of c-tilde}, this implies that \eqref{two-sided case 4} is equal to
\begin{align}
&\lim_{n\rightarrow\infty}P_{\lambda_{h_n,\mathfrak{b},\mathfrak{d},\Sigma^*}}\left(-\min\left\{z_{1-\frac{\alpha-\gamma}{2}},-\widehat\Omega_{h_n,bd^{(s_2^*)}}\widehat\Omega_{h_n,d^{(s_2^*)}d^{(s_2^*)}}^{-1}\diag(\widehat\Sigma_{h_n,d^{(s_2^*)}d^{(s_2^*)}})^{-1/2}\sqrt{h_n}\hat d_{h_n}^{(s_2^*)}+c_u^*\left({\tilde \Omega}^{*(s_1^*,s_2^*)}\right)\right\} \right. \notag \\
&\quad  \leq\frac{\sqrt{h_n}(\hat b_{h_n}-b_{h_n,\mathfrak{b}})}{\sqrt{\widehat\Sigma_{h_n,bb}}} \notag \\
&\quad \left.\leq \min\left\{z_{1-\frac{\alpha-\gamma}{2}},\widehat\Omega_{h_n,bd^{(s_1^*)}}\widehat\Omega_{h_n,d^{(s_1^*)}d^{(s_1^*)}}^{-1}\diag(\widehat\Sigma_{h_n,d^{(s_1^*)}d^{(s_1^*)}})^{-1/2}\sqrt{h_n}\hat d_{h_n}^{(s_1^*)}+\tilde c \left(c_u^*\left({\tilde \Omega}^{*(s_1^*,s_2^*)}\right),{\tilde \Omega}^{*(s_1^*,s_2^*)}\right)\right\}\right). \label{two-sided subsequence}
\end{align}

If $\|\mathfrak{d}^{(s_1^*)}\|,\|\mathfrak{d}^{(s_2^*)}\|<\infty$, since $\Omega_{bd^{(s_1^*)}}^*\Omega_{d^{(s_1^*)}d^{(s_1^*)}}^{*-1}\delta^{(s_1^*)}\geq 0$ and $\Omega_{bd^{(s_2^*)}}^*\Omega_{d^{(s_2^*)}d^{(s_2^*)}}^{*-1}\delta^{(s_2^*)}\leq 0$ by (i)--(v) and \eqref{bd drifting seq}, \eqref{two-sided subsequence} is equal to
\begin{align}
&P\left(-\min\left\{z_{1-\frac{\alpha-\gamma}{2}},-\Omega_{bd^{(s_2^*)}}^*\Omega_{d^{(s_2^*)}d^{(s_2^*)}}^{*-1}Y_\delta^{(s_2^*)}+c_u^*\left({\tilde \Omega}^{*(s_1^*,s_2^*)}\right)\right\}\right. \notag \\
&\quad\left.\leq Z_1\leq \min\left\{z_{1-\frac{\alpha-\gamma}{2}},\Omega_{bd^{(s_1^*)}}^*\Omega_{d^{(s_1^*)}d^{(s_1^*)}}^{*-1}Y_\delta^{(s_1^*)}+\tilde c \left(c_u^*\left({\tilde \Omega}^{*(s_1^*,s_2^*)}\right),{\tilde \Omega}^{*(s_1^*,s_2^*)}\right)\right\}\right) \notag \\
&=P\left(-\min\left\{z_{1-\frac{\alpha-\gamma}{2}},-\Omega_{bd^{(s_2^*)}}^*\Omega_{d^{(s_2^*)}d^{(s_2^*)}}^{*-1}\delta^{(s_2^*)}-\tilde Z_3+c_u^*\left({\tilde \Omega}^{*(s_1^*,s_2^*)}\right)\right\}\right. \notag \\
&\quad \left.\leq Z_1\leq \min\left\{z_{1-\frac{\alpha-\gamma}{2}},\Omega_{bd^{(s_1^*)}}^*\Omega_{d^{(s_1^*)}d^{(s_1^*)}}^{*-1}\delta^{(s_1^*)}+\tilde Z_2+\tilde c \left(c_u^*\left({\tilde \Omega}^{*(s_1^*,s_2^*)}\right),{\tilde \Omega}^{*(s_1^*,s_2^*)}\right)\right\}\right) \notag \\
& \geq P\left(-\min\left\{z_{1-\frac{\alpha-\gamma}{2}},-\tilde Z_3+c_u^*\left({\tilde \Omega}^{*(s_1^*,s_2^*)}\right)\right\}\right. \notag \\
&\quad \left.\leq Z_1\leq \min\left\{z_{1-\frac{\alpha-\gamma}{2}},\tilde Z_2+\tilde c \left(c_u^*\left({\tilde \Omega}^{*(s_1^*,s_2^*)}\right),{\tilde \Omega}^{*(s_1^*,s_2^*)}\right)\right\}\right)=1-\alpha \label{two-sided both finite d case}
\end{align}
by the definition of $\tilde c(\cdot)$ in \eqref{cl_cc} and Lemma \ref{lem:zero var size}, where
\begin{equation*}
\left(\begin{array}{c}
Z_1 \\
\tilde Z_2 \\
\tilde Z_3
\end{array}\right)\sim \mathcal{N}
\left(\left(\begin{array}{c}
0 \\
0 \\
0
\end{array}\right), 
\left(\begin{array}{ccc}
1 & \Omega_{\beta \delta^{(s_1^*)}}^* & \Omega_{\beta \delta^{(s_2^*)}}^* \\
\Omega_{ \delta^{(s_1^*)}\beta}^* & \Omega_{ \delta^{(s_1^*)} \delta^{(s_1^*)}}^* & \Omega_{ \delta^{(s_1^*)} \delta^{(s_2^*)}}^* \\
\Omega_{ \delta^{(s_2^*)}\beta}^* & \Omega_{ \delta^{(s_2^*)} \delta^{(s_1^*)}}^* & \Omega_{ \delta^{(s_2^*)} \delta^{(s_2^*)}}^*
\end{array}\right)\right).
\end{equation*}
If $\|\mathfrak{d}^{(s_1^*)}\|<\infty$, $\|\mathfrak{d}^{(s_2^*)}\|=\infty$, since $\Omega_{bd^{(s_1^*)}}^*\Omega_{d^{(s_1^*)}d^{(s_1^*)}}^{*-1}\delta^{(s_1^*)}\geq 0$ by (i)--(v) and \eqref{bd drifting seq}, \eqref{two-sided subsequence} is equal to
\begin{gather}
P\left(-z_{1-\frac{\alpha-\gamma}{2}}\leq Z_1\leq \min\left\{z_{1-\frac{\alpha-\gamma}{2}},\Omega_{bd^{(s_1^*)}}^*\Omega_{d^{(s_1^*)}d^{(s_1^*)}}^{*-1}Y_\delta^{(s_1^*)}+\tilde c \left(c_u^*\left({\tilde \Omega}^{*(s_1^*,s_2^*)}\right),{\tilde \Omega}^{*(s_1^*,s_2^*)}\right)\right\}\right) \notag \\
\geq P\left(-z_{1-\frac{\alpha-\gamma}{2}}\leq Z_1\leq \min\left\{z_{1-\frac{\alpha-\gamma}{2}},\tilde Z_2+\tilde c \left(c_u^*\left({\tilde \Omega}^{*(s_1^*,s_2^*)}\right),{\tilde \Omega}^{*(s_1^*,s_2^*)}\right)\right\}\right) \notag \\
\geq P\left(-\min\left\{z_{1-\frac{\alpha-\gamma}{2}},-\tilde Z_3+c_u^*\left({\tilde \Omega}^{*(s_1^*,s_2^*)}\right)\right\}\leq Z_1\leq \min\left\{z_{1-\frac{\alpha-\gamma}{2}},\tilde Z_2+\tilde c \left(c_u^*\left({\tilde \Omega}^{*(s_1^*,s_2^*)}\right),{\tilde \Omega}^{*(s_1^*,s_2^*)}\right)\right\}\right)=1-\alpha \label{two-sided first finite d case}
\end{gather}
by the definition of $\tilde c(\cdot)$ in \eqref{cl_cc} and Lemma \ref{lem:zero var size}.  If $\|\mathfrak{d}^{(s_1^*)}\|=\infty$, $\|\mathfrak{d}^{(s_2^*)}\|<\infty$, since $\Omega_{bd^{(s_2^*)}}^*\Omega_{d^{(s_2^*)}d^{(s_2^*)}}^{*-1}\delta^{(s_2^*)}\leq 0$ by (i)--(v) and \eqref{bd drifting seq}, \eqref{two-sided subsequence} is equal to
\begin{gather}
P\left(-\min\left\{z_{1-\frac{\alpha-\gamma}{2}},-\Omega_{bd^{(s_2^*)}}^*\Omega_{d^{(s_2^*)}d^{(s_2^*)}}^{*-1}Y_\delta^{(s_2^*)}+c_u^*\left({\tilde \Omega}^{*(s_1^*,s_2^*)}\right)\right\}\leq Z_1\leq z_{1-\frac{\alpha-\gamma}{2}}\right) \notag \\
\geq P\left(-\min\left\{z_{1-\frac{\alpha-\gamma}{2}},-\tilde Z_3+c_u^*\left({\tilde \Omega}^{*(s_1^*,s_2^*)}\right)\right\}\leq Z_1\leq z_{1-\frac{\alpha-\gamma}{2}}\right) \notag\\
 \geq P\left(-\min\left\{z_{1-\frac{\alpha-\gamma}{2}},-\tilde Z_3+c_u^*\left({\tilde \Omega}^{*(s_1^*,s_2^*)}\right)\right\}\leq Z_1\leq \min\left\{z_{1-\frac{\alpha-\gamma}{2}},\tilde Z_2+\tilde c \left(c_u^*\left({\tilde \Omega}^{*(s_1^*,s_2^*)}\right),{\tilde \Omega}^{*(s_1^*,s_2^*)}\right)\right\}\right)=1-\alpha \label{two-sided second finite d case}
\end{gather}
by the definition of $\tilde c(\cdot)$ in \eqref{cl_cc} and Lemma \ref{lem:zero var size}.  Finally, if $\|\mathfrak{d}^{(s_1^*)}\|,\|\mathfrak{d}^{(s_2^*)}\|=\infty$, \eqref{two-sided subsequence} is equal to
\begin{equation}
P\left(z_{1-\frac{\alpha-\gamma}{2}}\leq Z_1\leq z_{1-\frac{\alpha-\gamma}{2}}\right)=1-\alpha+\gamma>1-\alpha. \label{two-sided infinite d case}
\end{equation}
by (ii)-(v).  Together, \eqref{two-sided case 4}--\eqref{two-sided infinite d case} yield the lower bound in the statement of the theorem for $CI_{t,n}(\cdot)$.

To prove the upper bound, note that by nearly identical arguments to those used to establish \eqref{two-sided subsequence},
\begin{align*}
&\limsup_{n\rightarrow\infty}\text{ }\sup_{\lambda\in\Lambda}P_\lambda\left(b\in CI_{t,n}(\hat b_n,\hat d_n;\widehat\Sigma_n)\right)= \\
&\lim_{n\rightarrow\infty}P_{\lambda_{m_n,\mathfrak{b},\mathfrak{d},\Sigma^*}}\left(-\min\left\{z_{1-\frac{\alpha-\gamma}{2}},-\widehat\Omega_{m_n,bd^{(s_2^*)}}\widehat\Omega_{m_n,d^{(s_2^*)}d^{(s_2^*)}}^{-1}\diag(\widehat\Sigma_{m_n,d^{(s_2^*)}d^{(s_2^*)}})^{-1/2}\sqrt{m_n}\hat d_{m_n}^{(s_2^*)}+c_u^*\left({\tilde \Omega}^{*(s_1^*,s_2^*)}\right)\right\} \right.  \\
&\quad  \leq\frac{\sqrt{m_n}(\hat b_{m_n}-b_{m_n,\mathfrak{b}})}{\sqrt{\widehat\Sigma_{m_n,bb}}}  \\
&\quad \left.\leq \min\left\{z_{1-\frac{\alpha-\gamma}{2}},\widehat\Omega_{m_n,bd^{(s_1^*)}}\widehat\Omega_{m_n,d^{(s_1^*)}d^{(s_1^*)}}^{-1}\diag(\widehat\Sigma_{m_n,d^{(s_1^*)}d^{(s_1^*)}})^{-1/2}\sqrt{m_n}\hat d_{m_n}^{(s_1^*)}+\tilde c \left(c_u^*\left({\tilde \Omega}^{*(s_1^*,s_2^*)}\right),{\tilde \Omega}^{*(s_1^*,s_2^*)}\right)\right\}\right)
\end{align*}
for a subsequence $\{m_n:n\geq 1\}$ of $\{n:n\geq 1\}$ such that $\lambda_{m_n,\mathfrak{b},\mathfrak{d},\Sigma^*}\in\Lambda$ for all $n\geq 1$, $\sqrt{m_n}(b_{m_n,\mathfrak{b}},d_{m_n,\mathfrak{d}})\rightarrow (\mathfrak{b},\mathfrak{d}) $ and $\Sigma_{m_n,\Sigma^*}\rightarrow \Sigma^*$ for some $(\mathfrak{b},\mathfrak{d},\Sigma^*)\in \mathbb{R}_\infty\times \mathbb{R}_{+,\infty}^k\times \Phi$ with $\lambda_{\min}(\Sigma^*)\geq \kappa$ and $\lambda_{\max}(\Sigma^*)\leq \kappa^{-1}$ and some $c_u^*\left({\tilde \Omega}^{*(s_1^*,s_2^*)}\right)\in \tilde c_u\left({\tilde \Omega}^{*(s_1^*,s_2^*)}\right)$.  Note that for the probability to the left of the inequality in \eqref{two-sided both finite d case},
\begin{gather*}
P\left(-\min\left\{z_{1-\frac{\alpha-\gamma}{2}},-\Omega_{bd^{(s_2^*)}}^*\Omega_{d^{(s_2^*)}d^{(s_2^*)}}^{*-1}Y_\delta^{(s_2^*)}+c_u^*\left({\tilde \Omega}^{*(s_1^*,s_2^*)}\right)\right\}\right. \\
\leq \left. Z_1\leq \min\left\{z_{1-\frac{\alpha-\gamma}{2}},\Omega_{bd^{(s_1^*)}}^*\Omega_{d^{(s_1^*)}d^{(s_1^*)}}^{*-1}Y_\delta^{(s_1^*)}+\tilde c \left(c_u^*\left({\tilde \Omega}^{*(s_1^*,s_2^*)}\right),{\tilde \Omega}^{*(s_1^*,s_2^*)}\right)\right\}\right) \\
\leq P\left(-z_{1-\frac{\alpha-\gamma}{2}}\leq Z_1\leq z_{1-\frac{\alpha-\gamma}{2}}\right)=1-\alpha+\gamma;
\end{gather*}
for the probability to the left of the first inequality in \eqref{two-sided first finite d case},
\begin{gather*}
P\left(-z_{1-\frac{\alpha-\gamma}{2}}\leq Z_1\leq \min\left\{z_{1-\frac{\alpha-\gamma}{2}},\Omega_{bd^{(s_1^*)}}^*\Omega_{d^{(s_1^*)}d^{(s_1^*)}}^{*-1}Y_\delta^{(s_1^*)}+\tilde c \left(c_u^*\left({\tilde \Omega}^{*(s_1^*,s_2^*)}\right),{\tilde \Omega}^{*(s_1^*,s_2^*)}\right)\right\}\right) \\
\leq P\left(-z_{1-\frac{\alpha-\gamma}{2}}\leq Z_1\leq z_{1-\frac{\alpha-\gamma}{2}}\right)=1-\alpha+\gamma;
\end{gather*}
and for the probability to the left of the first inequality in \eqref{two-sided second finite d case},
\begin{gather*}
P\left(-\min\left\{z_{1-\frac{\alpha-\gamma}{2}},-\Omega_{bd^{(s_2^*)}}^*\Omega_{d^{(s_2^*)}d^{(s_2^*)}}^{*-1}Y_\delta^{(s_2^*)}+c_u^*\left({\tilde \Omega}^{*(s_1^*,s_2^*)}\right)\right\}\leq Z_1\leq z_{1-\frac{\alpha-\gamma}{2}}\right) \\
\leq P\left(-z_{1-\frac{\alpha-\gamma}{2}}\leq Z_1\leq z_{1-\frac{\alpha-\gamma}{2}}\right)=1-\alpha+\gamma.
\end{gather*}
Then, nearly identical reasoning used to establish \eqref{two-sided case 4}--\eqref{two-sided infinite d case}, replacing ``$\liminf_{n\rightarrow\infty}\inf_{\lambda\in\Lambda}$'' with ``$\limsup_{n\rightarrow\infty}\sup_{\lambda\in\Lambda}$'' and the subsequences $\{k_n:n\geq 1\}$ and $\{\lambda_{k_n,\mathfrak{b},\mathfrak{d},\Sigma^*}\in\Lambda:n\geq 1\}$ with $\{m_n:n\geq 1\}$ and $\{\lambda_{m_n,\mathfrak{b},\mathfrak{d},\Sigma^*}\in\Lambda:n\geq 1\}$, yields the upper bound in the statement of the theorem for $CI_{t,n}(\cdot)$. 
\end{proof}

\section{Parameter Space for the Standard Linear Regression Model} \label{sec:reg_space}

In this section, we provide details for parameter spaces satisfying (i)--(v) in Section \ref{sec:asymptotics} in the context of the standard linear regression model.  Recall in this setting we are interested in conducting inference on a regression coefficient of interest $b$ in the standard linear regression model for observations $i=1,\ldots,n$
\[y_i=bz_i+x_i^{\prime}d+w_i^{\prime}c+\varepsilon_i,\]
where $y_i$ is the dependent variable, $z_i$ is the scalar regressor of interest, $x_i\in\mathbb{R}^{\mathcal D_x}$ are control variables with \emph{known positive partial effects} $d\geq 0$ on $y_i$, $w_i\in\mathbb{R}^{\mathcal D_w}$ are control variables with unrestricted partial effects $c$ and $\varepsilon_i$ is the error term. 

Define $h_i=(z_i,x_i^{\prime},w_i^{\prime})^{\prime}$ so that the ordinary least squares estimator of $(b,d^{\prime})^{\prime}$, $(\hat b_n,\hat d_n^{\prime})^{\prime}$, is equal to the first $\mathcal D_x+1$ entries of $(\sum_{i=1}^nh_ih_i^{\prime})^{-1}\sum_{i=1}^n h_iy_i$.  Let $F$ denote the joint distribution of the stationary random vectors $\{(h_i^{\prime},\varepsilon_i)^{\prime}:i\geq 1\}$ and define the parameter $\tilde\lambda=(b,d,c,\mathcal{V},Q,F)$.  The parameter space $\widetilde\Lambda$ for $\tilde\lambda$ is defined to include parameters $\tilde\lambda=(b,d,c,\mathcal{V},Q,F)$ such that for some finite $\kappa>0$, the following conditions hold:

(i') $b\in\mathbb{R}$, $d\in\mathbb{R}_+^{\mathcal D_x}$ and $c\in\mathbb{R}^{\mathcal D_w}$;

(ii') $\lim_{n\rightarrow \infty}n^{-1}\sum_{i=1}^n\sum_{j=1}^nE_F[h_ih_j^{\prime}\varepsilon_i\varepsilon_j]$ exists and equals $\mathcal{V}\in\Phi$ with $\lambda_{\max}(\mathcal V)\leq \kappa^{-1}$;

(iii') $E_F[h_ih_i^{\prime}]$ exists and equals $Q\in\Phi$ with $\lambda_{\min}(Q)\geq \kappa$.

In addition, under any sequence of parameters $\{\tilde \lambda_{n,\mathfrak{b},\mathfrak{d},\mathcal V^*,Q^*}=(b_{n,\mathfrak{b}},d_{n,\mathfrak{d}},\mathcal{V}_{n,\mathcal{V}^*},Q_{n,Q^*},F_{n,\mathfrak{b},\mathfrak{d},\mathcal{V}^*,Q^*}):n\geq 1\}$ in $\widetilde\Lambda$ such that \eqref{bd drifting seq} holds, 
\begin{equation}
\mathcal{V}_{n,\mathcal{V}^*}\rightarrow \mathcal{V}^* \label{V conv}
\end{equation}
 and 
\begin{equation}
 Q_{n,Q^*}\rightarrow Q^* \label{Q conv}
 \end{equation} 
 for some $\mathcal{V}^*,Q^*\in \Phi$, the following remaining conditions hold:

(iv') $\widehat{\mathcal V}_n$ and $\widehat Q_n\equiv n^{-1}\sum_{i=1}^nh_ih_i^{\prime}$ exist and $\lambda_{\min}(n^{-1}\sum_{i=1}^nh_ih_i^{\prime})>0$ with probability one for all $n\geq 1$;

(v') $\widehat{\mathcal V}_n\overset{p}\longrightarrow \mathcal{V}^*$;

(vi') $\widehat Q_n\overset{p}\longrightarrow Q^*$;

(vii') $n^{-1/2}\sum_{i=1}^n h_i\varepsilon_i\overset{d}\longrightarrow \mathcal{N}(0,\mathcal{V}^*)$;

(viii') for any sequence $\{\tilde \lambda_{n,\mathfrak{b},\mathfrak{d},\mathcal V^*,Q^*}\}$ in $\widetilde\Lambda$ and any subsequence $\{s_n:n\geq 1\}$ of $\{n:n\geq 1\}$ for which \eqref{bd drifting seq}, \eqref{V conv}--\eqref{Q conv} hold along the subsequence, conditions (iv')--(vii') also hold along the subsequence.

Note that for $\Sigma$ equal to the upper left $(\mathcal D_x+1)\times (\mathcal D_x+1)$ submatrix of $Q^{-1}\mathcal VQ^{-1}$ and $\widehat\Sigma_n$ equal to the upper left $(\mathcal D_x+1)\times (\mathcal D_x+1)$ submatrix of $\widehat Q_n^{-1}\widehat{\mathcal V}_n\widehat Q_n^{-1}$, the conditions (i')--(viii') on the parameter space $\widetilde\Lambda$ imply (i)--(v) for the parameter space $\Lambda$.  More specifically, (i') implies (i), (ii')--(iii') imply (ii), (iv')--(vi') imply (iii), (vi')--(vii') imply (iv) and (viii') implies (v).

In conjunction with a suitable choice of covariance matrix estimator $\widehat{\mathcal V}_n$, the above definition of the parameter space $\widetilde\Lambda$ effectively serves as a set of assumptions on the underlying DGP in the context of the standard linear regression model when using ordinary least squares for estimation.  Part (i') imposes known sign restrictions for the nuisance coefficients $d$ while letting the coefficient of interest $b$ and the other nuisance coefficients $c$ remain unrestricted.  Parts (ii')--(iii') are standard conditions ensuring the existence of asymptotic covariance matrices while parts (iv')--(vi') are high level assumptions that guarantee consistent estimators of these covariance matrices are available, typically shown via application of a law of large numbers.  Part (vii') is a high level assumption that directly assumes a central limit theorem holds for the product of the regressors and error term in the regression model, a result that is typically invoked when proving asymptotic normality of ordinary least squares estimators.  Finally, part (viii') is a mild technical condition used to establish results under relevant drifting sequences of DGPs.

The definition of the parameter space $\widetilde\Lambda$ is written at such a level of generality to allow for heteroskedasticity and/or weak dependence in the data, enabling the use of our CIs in both cross-sectional and time series settings.  We refer the interested reader to Section 3 of \cite{McC20} for two sets of weak low-level sufficient conditions that guarantee the high-level assumptions (iv')--(vii') hold when using standard covariance matrix estimators in the context of estimation robust to heteroskedasticity for randomly sampled data and estimation robust to heteroskedasticity and autocorrelation for time series data.

\section{Additional Tables} \label{AT}

\begin{table}[h!]											
\begin{center}									
\caption{Coefficients for $6^\text{th}$ order polynomial approximation of $c_u(\omega)$ for $\alpha = 0.01$ and $\gamma = \alpha/10$}									
\label{Coefficients_2s_table_1}									
\begin{tabular}{c|rrrrrrr}								
\hline								
\hline			
 & \multicolumn{1}{c}{1} & \multicolumn{1}{c}{$\omega_{12}$} & \multicolumn{1}{c}{$\omega_{12}^2$} & \multicolumn{1}{c}{$\omega_{12}^3$} & \multicolumn{1}{c}{$\omega_{12}^4$} & \multicolumn{1}{c}{$\omega_{12}^5$} & \multicolumn{1}{c}{$\omega_{12}^6$}  \\ 
\hline
1 & $2.5710$ & $1.4378$ & $-4.7977$ & $12.2591$ & $-20.5823$ & $18.2815$ & $-6.5866$ \\ 
$\omega_{13}$ & $1.1854$ & $-1.1672$ & $3.6035$ & $-2.5234$ & $0.2467$ & $0.6751$ & $$ \\ 
$\omega_{13}^2$ & $-16.4621$ & $-2.1843$ & $-2.6765$ & $0.8411$ & $-0.6847$ & $$ & $$ \\ 
$\omega_{13}^3$ & $63.1856$ & $8.4153$ & $1.0849$ & $0.7850$ & $$ & $$ & $$ \\ 
$\omega_{13}^4$ & $-128.0372$ & $-9.2032$ & $-0.3625$ & $$ & $$ & $$ & $$ \\ 
$\omega_{13}^5$ & $123.3096$ & $3.1479$ & $$ & $$ & $$ & $$ & $$ \\ 
$\omega_{13}^6$ & $-45.5050$ & $$ & $$ & $$ & $$ & $$ & $$ \\ 
\hline								
\end{tabular}
\end{center}						
\end{table}

\begin{table}[h!]											
\begin{center}									
\caption{Coefficients for $6^\text{th}$ order polynomial approximation of $c_u(\omega)$ for $\alpha = 0.1$ and $\gamma = \alpha/10$}									
\label{Coefficients_2s_table_10}									
\begin{tabular}{c|rrrrrrr}								
\hline								
\hline			
 & \multicolumn{1}{c}{1} & \multicolumn{1}{c}{$\omega_{12}$} & \multicolumn{1}{c}{$\omega_{12}^2$} & \multicolumn{1}{c}{$\omega_{12}^3$} & \multicolumn{1}{c}{$\omega_{12}^4$} & \multicolumn{1}{c}{$\omega_{12}^5$} & \multicolumn{1}{c}{$\omega_{12}^6$}  \\ 
\hline
1 & $1.6348$ & $1.2890$ & $-4.8501$ & $14.0485$ & $-23.9082$ & $20.3891$ & $-7.0186$ \\ 
$\omega_{13}$ & $1.2271$ & $0.0224$ & $-0.6555$ & $0.7875$ & $1.0308$ & $-0.5813$ & $$ \\ 
$\omega_{13}^2$ & $-11.7243$ & $-2.0585$ & $3.7550$ & $-5.0051$ & $1.5399$ & $$ & $$ \\ 
$\omega_{13}^3$ & $43.6253$ & $3.2898$ & $-1.7097$ & $1.1221$ & $$ & $$ & $$ \\ 
$\omega_{13}^4$ & $-87.8291$ & $-2.6854$ & $0.6640$ & $$ & $$ & $$ & $$ \\ 
$\omega_{13}^5$ & $84.6893$ & $0.5102$ & $$ & $$ & $$ & $$ & $$ \\ 
$\omega_{13}^6$ & $-31.4176$ & $$ & $$ & $$ & $$ & $$ & $$ \\ 
\hline								
\end{tabular}
\end{center}						
\end{table}

\newpage 

\bibliographystyle{apalike}
\bibliography{sign_restrict}


\end{document}